%% file: paper.tex
\RequirePackage{rotating} 

\documentclass[acmsmall,nonacm]{acmart}

\input{bib/macros}

\renewcommand\sfsmaller{}

\usepackage[english]{babel}
\usepackage{amsmath}
\usepackage{amsthm} 
\usepackage{amssymb} 
\usepackage{array}
\usepackage{enumerate}
\usepackage{enumitem} 
\usepackage{graphicx} 
\usepackage{algorithm}
\usepackage{algpseudocode}
\algtext*{EndFor} 
\algtext*{EndIf} 
\algtext*{EndWhile} 
\usepackage{mathtools} 
\usepackage{centernot} 
\usepackage{varwidth} 
\usepackage{url}
\usepackage{scalerel} 

\newcolumntype{H}{>{\setbox0=\hbox\bgroup}c<{\egroup}@{}} 
\newcolumntype{Z}{>{\setbox0=\hbox\bgroup}c<{\egroup}@{\hspace*{-\tabcolsep}}} 

\newtheorem*{definition*}{Definition}

\algnewcommand{\LineComment}[1]{\State \(\triangleright\) #1}
\makeatletter
\algnewcommand{\LineCommentx}[1]{\Statex \hskip\ALG@thistlm \(\triangleright\) #1}
\algnewcommand{\LineCommentxx}[1]{\Statex \hskip\ALG@tlm \(\triangleright\) #1}
\algrenewcommand\ALG@beginalgorithmic{\small}
\makeatother
\algnewcommand{\lForEach}[2] {\State \algorithmicforeach\ #1 \algorithmicdo\ #2} 
\algnewcommand{\lIf}[2] {\State \algorithmicif\ #1 \algorithmicthen\ #2} 
\algnewcommand{\lElse}[1]{\State \algorithmicelse\ #1} 
\algnewcommand{\lElsIf}[2] {\State \algorithmicelse\ \algorithmicif\ #1 \algorithmicthen\ #2} 
\algnewcommand{\lForAll}[2]{\State \algorithmicforall\ #1 \algorithmicdo\ #2} 

\newcommand\lines[2]{\range{line}{lines}{#1}{#2}}
\newcommand\Lines[2]{\range{Line}{Lines}{#1}{#2}}
\newcommand\range[4]{%
  \ifthenelse{\equal{#3}{#4}}
    {#1~#3}
    {#2~#3--#4}%
}

\usepackage{tikz} 
\newcommand{\tikzmark}[1]{\tikz[remember picture, baseline] \node[inner sep=0pt, outer sep=0pt] (#1){};}
\newcommand{\link}[2]{
\begin{tikzpicture}[remember picture, overlay, >=stealth, shift={(0,0)}]
	\draw[arrows=->] ([xshift=#1pt,yshift=#2pt]#1.center) -- ([xshift=#1pt,yshift=#2pt]#2.center);
\end{tikzpicture}}
\newcommand{\textlink}[3]{%
\begin{tikzpicture}[remember picture, overlay, >=stealth, shift={(0,0)}] 
	\draw[arrows=->] (#1) to node[sloped,anchor=center,above] {{\smaller #3}} (#2);
\end{tikzpicture}}
\newcommand{\undertextlink}[3]{%
\begin{tikzpicture}[remember picture, overlay, >=stealth, shift={(0,0)}] 
	\draw[arrows=->] (#1) to node[sloped,anchor=center,below] {{\smaller #3}} (#2);
\end{tikzpicture}}

\newcommand{\hbFull}{happens-before\xspace}
\newcommand{\HbFull}{Happens-before\xspace}
\newcommand{\HBFull}{Happens-Before\xspace}
\newcommand{\HB}{HB\xspace}
\newcommand{\cpFull}{causally-precedes\xspace}
\newcommand{\CpFull}{Causally-precedes\xspace}
\newcommand{\CPFull}{Causally-Precedes\xspace}
\newcommand{\CP}{CP\xspace}
\newcommand{\POFull}{Program Order\xspace}
\newcommand{\PoFull}{Program-order\xspace}

\newcommand{\PO}{PO\xspace}

\newcommand{\PCP}{CCP\xspace}
\newcommand{\CCPFull}{Conditionally Causally-Precedes\xspace}

\newcommand{\CCP}{CCP\xspace}
\newcommand{\GoldCP}{Raptor\xspace} 
\newcommand{\goldCP}{\GoldCP}
\newcommand{\raptor}{\goldCP}
\newcommand{\Raptor}{\GoldCP}
 
\newcommand{\goldHB}{Raptor-HB\xspace}

\newcommand{\FTFull}{FastTrack\xspace}
\newcommand\Add[2]{\ensuremath{#2^+ \gets #2^+ \cup \{ #1 \}}}

\newcommand{\CPThreadPlain}[1]{\ensuremath{#1}}
\newcommand{\CPThread}[1]{\CPThreadPlain{\code{#1}}}
\newcommand{\HBThread}[1]{\ensuremath{\code{#1}}}
\newcommand{\HBThreadPlain}[1]{\ensuremath{#1}}

\newcommand{\CPLock}[1]{\ensuremath{\code{#1}}}
\newcommand{\HBLock}[2]{\ensuremath{\code{#1}^{#2}}}
\newcommand{\LSLock}[2]{\ensuremath{\code{#1}^{#2}_{*}}}
\newcommand{\PCPLockPlain}[3]{\ensuremath{#1\!:\!\HBLock{#2}{#3}}}
\newcommand{\CCPLockPlain}[3]{\PCPLockPlain{#1}{#2}{#3}}
\newcommand{\PCPLock}[3]{\PCPLockPlain{\code{#1}}{#2}{#3}}
\newcommand{\CCPLock}[3]{\PCPLock{#1}{#2}{#3}}
\newcommand{\PCPThreadPlain}[3]{\ensuremath{#1\!:\!\HBLock{#2}{#3}}}
\newcommand{\CCPThreadPlain}[3]{\PCPThreadPlain{#1}{#2}{#3}} 
\newcommand{\PCPThread}[3]{\PCPThreadPlain{\code{#1}}{#2}{#3}}
\newcommand{\CCPThread}[3]{\PCPThread{#1}{#2}{#3}}

\newcommand{\OwnerElement}[2]{\ensuremath{\code{#1}^{#2}}}
\newcommand{\OwnerElementRead}[3]{\ensuremath{\code{#1}^{#2}_{\code{#3}}}}
\newcommand{\Set}{Set\xspace}
\newcommand{\set}{set\xspace}
\newcommand{\Sets}{Sets\xspace}
\newcommand{\sets}{sets\xspace}

\newcommand{\POLockSetPlain}[2]{\ensuremath{\mathit{PO}({#1}^{#2})}}
\newcommand{\POLockSet}[2]{\ensuremath{\mathit{PO}(\code{#1}^{#2})}}
\newcommand{\POLockSetRead}[3]{\ensuremath{\mathit{PO}(\code{#1}^{#2}_\thr{#3})}}
\newcommand{\HBLockSetPlain}[2]{\ensuremath{\mathit{HB}({#1}^{#2})}}
\newcommand{\HBLockSet}[2]{\HBLockSetPlain{\code{#1}}{#2}}
\newcommand{\HBLockSetRead}[3]{\ensuremath{\mathit{HB}(\code{#1}^{#2}_\thr{#3})}}
\newcommand{\CPLockSetPlain}[2]{\ensuremath{\mathit{CP}({#1}^{#2})}}
\newcommand{\CPLockSet}[2]{\ensuremath{\CPLockSetPlain{\code{#1}}{#2}}\xspace}
\newcommand{\CPLockSetRead}[3]{\ensuremath{\mathit{CP}(\code{#1}^{#2}_\thr{#3})}}
\newcommand{\PCPLockSetPlain}[2]{\ensuremath{\mathit{CCP}({#1}^{#2})}}
\newcommand{\CCPLockSetPlain}[2]{\PCPLockSetPlain{#1}{#2}}
\newcommand{\PCPLockSet}[2]{\PCPLockSetPlain{\code{#1}}{#2}}
\newcommand{\CCPLockSet}[2]{\PCPLockSet{#1}{#2}}
\newcommand{\PCPLockSetRead}[3]{\ensuremath{\mathit{CCP}(\code{#1}^{#2}_\thr{#3})}}
\newcommand{\CCPLockSetRead}[3]{\PCPLockSetRead{#1}{#2}{#3}}
\newcommand{\xxrightarrow}[1]{\raisebox{-3pt}{$\xrightarrow{#1}$\;}}

\newcommand\xiT[1]{\ensuremath{\xi_\thr{#1}}}
\newcommand{\tr}{\ensuremath{\mathit{tr}}\xspace}

\newcommand{\Relation}[1]{\ensuremath{\prec_\textsc{\tiny{#1}}}\xspace}
\newcommand{\NotRelation}[1]{\ensuremath{\not\prec_\textsc{\tiny{#1}}}\xspace}
\newcommand{\EqRelation}[1]{\ensuremath{\preccurlyeq_\textsc{\tiny{#1}}}\xspace}
\newcommand{\PORelation}{\Relation{\PO}}
\newcommand{\NotPORelation}{\NotRelation{\PO}}
\newcommand{\EqPORelation}{\EqRelation{\PO}}

\newcommand{\HBRelation}{\Relation{\HB}}
\newcommand{\NotHBRelation}{\NotRelation{\HB}}
\newcommand{\EqHBRelation}{\EqRelation{\HB}}
\newcommand{\CPRelation}{\Relation{\CP}}
\newcommand{\NotCPRelation}{\NotRelation{\CP}}

\newcommand{\POOrdered}[2]{\Ordered{#1}{\PORelation}{#2}}
\newcommand{\NotPOOrdered}[2]{\Ordered{#1}{\NotPORelation}{#2}}
\newcommand{\EqPOOrdered}[2]{\Ordered{#1}{\EqPORelation}{#2}}

\newcommand{\HBOrdered}[2]{\Ordered{#1}{\HBRelation}{#2}}
\newcommand{\NotHBOrdered}[2]{\Ordered{#1}{\NotHBRelation}{#2}}
\newcommand{\EqHBOrdered}[2]{\Ordered{#1}{\EqHBRelation}{#2}}
\newcommand{\CPOrdered}[2]{\Ordered{#1}{\CPRelation}{#2}}
\newcommand{\NotCPOrdered}[2]{\Ordered{#1}{\NotCPRelation}{#2}}

\newcommand{\Ordered}[3]{\ensuremath{#1 #2 #3}}
\newcommand\getThread[1]{\ensuremath{\mathit{thr}(#1)}}
\newcommand\EqTotalOrder{\ensuremath{\le_\mathit{tr}}\xspace}
\newcommand\totalOrder{\ensuremath{<_\mathit{tr}}\xspace}
\newcommand{\conflicts}[2]{\ensuremath{#1 \asymp #2}}

\newcommand{\Read}[3]{\ensuremath{\code{rd(#1)}^{#2}_\thr{#3}}} 
\newcommand{\Write}[2]{\ensuremath{\code{wr(#1)}^{#2}}}
\newcommand{\Acquire}[2]{\ensuremath{\code{acq(#1)}^{#2}}} 
\newcommand{\Release}[2]{\ensuremath{\code{rel(#1)}^{#2}}}

\newcommand{\race}{\CP-race\xspace}
\newcommand{\races}{\race{}s\xspace}

\newcommand{\RuleA}{Rule~(a)\xspace}
\newcommand{\RuleB}{Rule~(b)\xspace}
\newcommand{\RuleC}{Rule~(c)\xspace}
\newcommand{\Rules}{Rules\xspace}
\newcommand{\erho}{\ensuremath{e_{\rho}}\xspace}
\newcommand{\exi}{\ensuremath{e_{\xi}}\xspace} 
\newcommand{\exiT}[1]{\ensuremath{e_{\xiT{#1}}}\xspace}
\newcommand\rhoprime{\ensuremath{\sigma}}
\newcommand\heldBy[1]{\ensuremath{\mathit{heldBy}(\thr{#1})}}

\newcommand\CPdistance{\CP-distance\xspace}
\newcommand\CPDistance{\CP-Distance\xspace}
\newcommand\dist[3]{\ensuremath{d(\OwnerElement{#1}{#3}\!\leadsto\!\OwnerElement{#1}{#2})}}

\renewcommand{\originalgrumbler}[2]{\begin{quote}\textcolor{blue}{{\bf #1 says:} #2}\end{quote}}

\newcommand\notes[1]{\begin{quote}\textcolor{darkgreen}{\textbackslash \textbf{notes\{}} #1 \textcolor{darkgreen}{\}}\end{quote}}

\newcommand\later[1]{\begin{quote}\textcolor{darkgreen}{\textbackslash \textbf{later\{}} #1 \textcolor{darkgreen}{\}}\end{quote}}

\input{toggle-comments}

\acmJournal{PACMPL}
\acmVolume{1}
\acmNumber{1}
\acmArticle{1}
\acmYear{2018}
\acmMonth{1}
\acmDOI{10.1145/nnnnnnn.nnnnnnn}
\startPage{1}

\setcopyright{none}

\begin{document}


\title[{Online Set-Based Dynamic Analysis for Sound Predictive Race Detection}]
{Online Set-Based Dynamic Analysis for \\ Sound Predictive Race Detection}


%
\author{Jake Roemer}
\affiliation{
	\institution{Ohio State University}
}
\email{roemer.37@osu.edu}
\author{Michael D.\ Bond}
\affiliation{
	\institution{Ohio State University}
}
\email{mikebond@cse.ohio-state.edu}

\input{abstract}

\maketitle

%

\input{Intro}

\input{Definitions}

\input{Overview}

\input{State}

\input{Analysis} 

\input{Removal}

\input{Evaluation}

\input{RelatedWorks}

\input{Conclusion} 

\input{acks}

\appendix

\input{Correctness}

\bibliographystyle{abbrv}
\bibliography{bib/conf-abbrv,bib/plass}

\end{document}

%% file: bib/macros.tex
\makeatletter
\@ifpackageloaded{subfig}
{}
{\usepackage[labelformat=simple,subrefformat=parens,caption=false]{subfig}
 
 }
\makeatother
\usepackage{graphicx}

\usepackage{color}
\usepackage{listings}
\usepackage{relsize}
\usepackage{multirow}
\usepackage[normalem]{ulem} 

\usepackage{xspace}

\newcommand{\originalgrumbler}[2]{\begin{quote}\textcolor{blue}{\sl{\bf #1 says:} #2}\end{quote}}
\newcommand{\grumbler}[2]{\originalgrumbler{#1}{#2}}

\newcommand{\mike}[1]{\grumbler{Mike}{#1}}

\newcommand{\jake}[1]{\grumbler{Jake}{#1}}

\definecolor{darkgreen}{rgb}{0,0.4,0}

\newcommand{\sfsmaller}{}

\newcommand{\bench}[1]{\textsf{\sfsmaller#1}}
\newcommand{\code}[1]{\textsf{\sfsmaller#1}}

\newcommand{\mc}[3]{\multicolumn{#1}{#2}{#3}}

\newcommand{\eg}{e.g.\xspace}
\newcommand{\ie}{i.e.\xspace}
\newcommand{\cf}{cf.\xspace}
\newcommand{\etal}{et al.\xspace}

\newcommand{\thr}[1]{\textsf{\sfsmaller#1}}

\makeatletter
\@ifundefined{state}{%
}{}
\makeatother

\usepackage{etoolbox}



%% file: toggle-comments.tex
\renewcommand{\notes}[1]{}
\renewcommand{\later}[1]{}



%% file: abstract.tex
\begin{abstract}

Predictive data race detectors find data races that exist in executions
\emph{other than} the observed execution.
Smaragdakis \etal\ introduced the \emph{\cpFull} (\CP) relation 
and a polynomial-time analysis 
for \emph{sound} (no false races) predictive data race detection.
However, 
their analysis cannot scale beyond analyzing bounded windows of execution traces.

This work introduces a novel dynamic analysis called \emph{\raptor} 
that computes \CP soundly and completely.
\Raptor is inherently an online analysis 
that analyzes and finds all \CP-races of 
an execution trace in its entirety.
An evaluation of a prototype implementation of \raptor shows 
that it scales to program executions that the prior \CP analysis cannot handle, 
finding data races that the prior \CP analysis cannot find.

\end{abstract}

%% file: Intro.tex
\section{Introduction}
\label{sec:introduction}

A shared-memory program has a \emph{data race} 
if an execution of the program can perform two memory accesses that are 
\emph{conflicting} and \emph{concurrent},
 meaning that the accesses are executed by different threads,
at least one access is a write,
and there are no interleaving program operations.
Data races often lead to atomicity, order, and sequential consistency violations 
that cause programs to crash, hang, or corrupt data~\cite{boehm-miscompile-hotpar-11,
jmm-broken,portend-toplas15,conc-bug-study-2008,portend-asplos12,benign-races-2007,prescient-memory,adversarial-memory,racefuzzer,relaxer,
therac-25,blackout-2003-tr,nasdaq-facebook}.
Modern shared-memory programming languages including C++ and Java 
provide undefined or ill-defined semantics
for executions containing data races~\cite{java-memory-model,c++-memory-model-2008,memory-models-cacm-2010,
you-dont-know-jack,data-races-are-pure-evil,out-of-thin-air-MSPC14,jmm-broken}.

\emph{\HbFull} (\HB) analysis detects data races by tracking the \emph{\hbFull} relation~\cite{happens-before,adve-weak-racedet-1991}
and reports conflicting, concurrent accesses unordered by \HB as data races~\cite{fasttrack,multirace,goldilocks-pldi-2007}.
However, the coverage of \HB analysis is inherently limited to data races 
that manifest in the current execution.
Consider both example executions in Figure~\ref{fig:SimpleExample}.
The \HB relation, which is
the union of program and synchronization order~\cite{happens-before},
orders the accesses to \code{x} and 
would \emph{not} detect a data race for either observed execution.
However, for Figure~\ref{fig:SimpleExample:has-CP-race},
we can know \emph{from the observed execution} that a data race exists
in a different interleaving of events of the program
(if Thread~2 acquires \code{m} first, the accesses to \code{x} can occur concurrently).
On the other hand, it is unclear if Figure~\ref{fig:SimpleExample:no-CP-race} has a data race,
since the execution of \Read{x}{1}{T2} may depend on the order of accesses to \code{y}
(\eg, suppose \Read{x}{1}{T2}'s execution depends on the value read by \Read{y}{1}{T2}).

\input{examples/SimpleExample}

\emph{Predictive analyses} detect data races that are possible in executions
\emph{other than} the observed execution.
Notably, Smaragdakis \etal introduce a predictive analysis that tracks
the \emph{\cpFull} (\CP) relation~\cite{causally-precedes},
a subset of \HB that conservatively orders conflicting accesses 
that cannot race in some other, unobserved execution.
A \CP-race exists between conflicting accesses \emph{not} ordered by \CP.\footnote{More precisely, two conflicting
accesses unordered by \CP imply either a data race or a deadlock in another execution~\cite{causally-precedes}.}
In Figure~\ref{fig:SimpleExample:has-CP-race}, 
\HBOrdered{\Write{x}{1}}{\Read{x}{1}{T2}},
but \NotCPOrdered{\Write{x}{1}}{\Read{x}{1}{T2}},
\ie, the execution has a \emph{\race} (two conflicting accesses unordered by \CP).
In contrast,
Figure~\ref{fig:SimpleExample:no-CP-race} has no \CP-race
(\CPOrdered{\Write{x}{1}}{\Read{x}{1}{T2}})
because \CP correctly orders the critical sections that \emph{may}
result in different behavior if executed in the reverse order,
\ie, \Read{y}{1}{T2} \emph{may} read a different value
and the data race on \code{x} \emph{might not} occur.

Smaragdakis \etal\ show how to compute \CP in polynomial time in the execution length.
Nonetheless, their analysis cannot scale to full executions,
and instead analyzes bounded execution windows of 500 consecutive events~\cite{causally-precedes},
missing \CP-races 
involving accesses that are ``far apart'' in the observed execution.
Their \CP analysis is inherently offline;
in contrast, an \emph{online} dynamic analysis would summarize the execution so far 
in the form of \emph{analysis state}, without needing to ``look back'' at the entire trace. 
Like Smaragdakis's \CP analysis,
most existing predictive analyses are offline, 
needing access to the entire execution trace, and cannot scale to full execution traces~\cite{causally-precedes, rvpredict-pldi-2014, jpredictor,
maximal-causal-models, rdit-oopsla-2016, said-nfm-2011, ipa, wcp} (Section~\ref{sec:related-works}).

Two recent approaches introduce online predictive analyses~\cite{wcp,vindicator}.
This article's contributions are concurrent with or precede each of these prior approaches.
In particular, Kini \etal's \emph{weak-causally-precedes} (WCP) submission to PLDI 2017~\cite{wcp}
is concurrent with this article's work,
as established by our November 2016 technical report~\cite{raptor}.
Furthermore, none of the related work shows how to compute \emph{\CP} online,
a challenging proposition~\cite{causally-precedes} (Section~\ref{sec:motivation}).

\paragraph{Our approach.}

This article introduces \emph{\raptor} (\underline{Ra}ce \underline{p}redic\underline{tor}),
a novel dynamic analysis that computes the \CP relation soundly and completely.
\Raptor is inherently an \emph{online} analysis because
it summarizes an execution's behavior
so far in the form of \emph{analysis state}
that captures the recursive nature of the \CP relation,
rather than needing to look at the entire execution so far.
We introduce analysis invariants 
and prove that \raptor soundly and completely tracks \CP
by maintaining the invariants after each step of a program's execution.

We have implemented \raptor as a dynamic analysis for Java programs.
Our unoptimized prototype implementation
can analyze executions of real programs with 
hundreds of thousands or millions of events within an hour or two.
In contrast, Smaragdakis \etal's analysis generally cannot scale beyond bounded windows 
of thousands of events~\cite{causally-precedes}.
As a result, \raptor detects \CP-races
that are too ``far apart'' for the offline \CP analysis to detect.

While concurrent work's WCP analysis~\cite{wcp} 
and later work's DC analysis~\cite{vindicator} are faster and
(as a result of using weaker relations than \CP) detect more races than
\raptor,
computing \CP online is a challenging problem that 
prior work has been unable to solve~\cite{wcp,causally-precedes} (Section~\ref{sec:motivation}).
Furthermore, \raptor provides the first set-based algorithm
for partial-order-based predictive analysis.
Though recent advances in predictive analysis have subsumed \raptor's race coverage and performance,
the alternative technique of a set-based approach provides unique avenues for future development.

\Raptor advances the state of the art by 
(1) being the first online analysis for computing \CP soundly and completely and
(2) demonstrably scaling to longer executions than the prior \CP analysis and 
finding real \CP-races that the prior \CP analysis cannot detect.

%% file: examples/SimpleExample.tex
\begin{figure}[t]
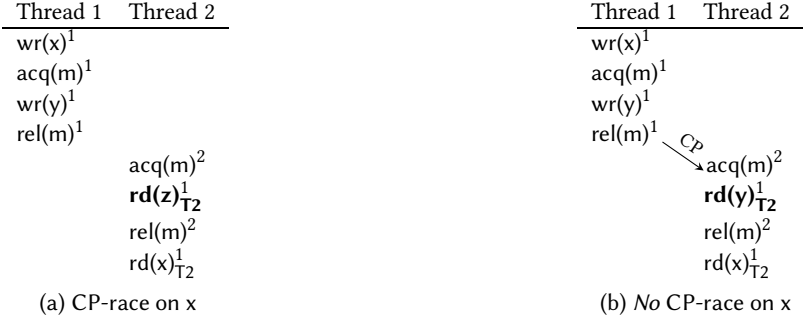

\small
\centering
\subfloat[CP-race on \code{x}]{
\begin{minipage}{0.4\textwidth}
\centering
\begin{tabular}{ll}
Thread 1 & Thread 2 \\\hline
\Write{x}{1} & \\
\Acquire{m}{1} & \\
\Write{y}{1} & \\
\Release{m}{1} & \\
& \Acquire{m}{2} \\
& \textbf{\Read{z}{1}{T2}} \\[.3em]
& \Release{m}{2} \\
& \Read{x}{1}{T2}
\end{tabular}
\end{minipage}
\label{fig:SimpleExample:has-CP-race}
}
\hspace{1.5cm}
\subfloat[\emph{No} \CP-race on \code{x}]{
\begin{minipage}{0.45\textwidth}
\centering
\begin{tabular}{ll}
Thread 1 & Thread 2 \\\hline
\Write{x}{1} & \\
\Acquire{m}{1} & \\
\Write{y}{1} & \\
\Release{m}{1}\tikzmark{1} \\
& \tikzmark{2}\Acquire{m}{2} \\
& \textbf{\Read{y}{1}{T2}} \\[.3em]
& \Release{m}{2} \\
& \Read{x}{1}{T2}
\end{tabular}
\textlink{1}{2}{\CP}
\label{fig:SimpleExample:no-CP-race}
\end{minipage}
}


\caption{\label{fig:SimpleExample} Two observed program executions (of potentially different programs) in which
top-to-bottom ordering represents the observed order of events;
column placement represents the executing thread;
\code{x}, \code{y}, and \code{z} are shared variables;
and \code{m} is a lock.
For each event, the superscript and subscript differentiate repeat operations;
only reads (\eg, \Read{y}{1}{T2}) have subscripts, which differentiate reads by different threads since the last write to the same variable.
The arrow represents a \CP edge established by \RuleA of the \CP definition (Definition~\ref{def:cp}).
Bold text shows the (sole) difference between the two executions.
In both figures, \HB orders conflicting accesses to \code{x}, 
but in Figure~\ref{fig:SimpleExample:has-CP-race}
\protect\NotCPOrdered{\Write{x}{1}}{\Read{x}{1}{T2}},
whereas in Figure~\ref{fig:SimpleExample:no-CP-race}
\CPOrdered{\Write{x}{1}}{\Read{x}{1}{T2}}.}
\end{figure}

%% file: Definitions.tex
\section{Background and Motivation}
\label{sec:definitions}

This section defines the \cpFull (\CP) relation from prior work~\cite{causally-precedes}
and motivates the challenges of computing \CP online.
First, we introduce the execution model and notation used throughout the article.

\subsection{Execution Model}

An execution trace \tr consists of events
observed in a total order, denoted \EqTotalOrder (reflexive) or \totalOrder (irreflexive),
that represents
a linearization of a sequentially consistent (SC) execution~\cite{sequential-consistency}.\footnote{It is safe to assume
SC because language memory models provide SC up to the first data race~\cite{memory-models-cacm-2010}.}
Every event in \tr has an associated executing thread.

An event is one of \Write{x}{i}, \Read{x}{i}{T}, \Acquire{m}{i}, or \Release{m}{i},
where \code{x} is a variable, \code{m} is a lock,
and $i$ specifices the $i$th instance of the event,
\ie, \Write{x}{i} is the $i$th write to variable \code{x}.
\Read{x}{i}{T} is a read by thread \thr{T} to variable \code{x}
such that $\Write{x}{i} \totalOrder \Read{x}{i}{T} \totalOrder \Write{x}{i+1}$.

We assume that any observed execution trace is \emph{well formed},
meaning a thread only acquires a lock that is not held
and only releases a lock it holds,
and lock release order is last in, first out (\ie, critical sections are well-nested).


Let \getThread{e} be a helper function
that returns the thread identifier that executed event $e$.
Two access (read/write) events $e$ and $e'$ to the same variable
are \emph{conflicting} (denoted \conflicts{e}{e'})
if at least one event is a write and $\getThread{e} \neq \getThread{e'}$.

\subsection{Relations over an Execution Trace}

The following presentation
is based on prior work's presentation~\cite{causally-precedes}.

\emph{\PoFull} (\PO) is a partial order over events executed by the same thread.
Given a trace \tr,
\EqPORelation is the smallest relation such that
for any two events $e$ and $e'$,
\EqPOOrdered{e}{e'} if $e \EqTotalOrder e' \land \getThread{e} = \getThread{e'}$.


%

\begin{definition}[\HbFull]
The \emph{\hbFull} (\HB) relation is a partial order over events in an execution trace~\cite{happens-before}.
Given a trace \tr, 
\EqHBRelation is the smallest relation such that:

\begin{itemize}
  \item Two events are \HB ordered if they are \PO ordered.
  That is, \EqHBOrdered{e}{e'} if \EqPOOrdered{e}{e'}.
   
  \item Release and acquire operations on the same lock 
  (\ie, synchronization order) are \HB ordered.
  That is, \EqHBOrdered{\Release{m}{i}}{\Acquire{m}{j}} if
  $\Release{m}{i} \totalOrder \Acquire{m}{j}$ (which implies $i < j$).
  
  \item \HB is closed under composition with itself.
  That is, \EqHBOrdered{e}{e'} if $\exists e'' \mid \EqHBOrdered{e}{\EqHBOrdered{e''}{e'}}$.
\end{itemize}
\label{def:hb}
\end{definition}

\noindent
Throughout the article, we generally use irreflexive variants of \EqPORelation and \EqHBRelation,
\PORelation and \HBRelation, respectively, when it is correct to do so.

\begin{definition}[\CpFull]
\begin{samepage}
The \emph{\cpFull} (\CP) relation is a strict (\ie, irreflexive) partial order that is strictly weaker than \HB~\cite{causally-precedes}.
Given a trace \tr,
\CPRelation is the smallest relation such that:

\begin{enumerate}[label=(\alph*)]
  \item Release and acquire operations on the same lock 
  containing conflicting events are \CP ordered.
  That is, \CPOrdered{\Release{m}{i}}{\Acquire{m}{j}} if
  $\exists e \exists e' \mid e \totalOrder e' \land \conflicts{e}{e'} \land
  \POOrdered{\POOrdered{\Acquire{m}{i}}{e}}{\Release{m}{i}} \land
  \POOrdered{\POOrdered{\Acquire{m}{j}}{e'}}{\Release{m}{j}}$.

  \item Two critical sections on the same lock are \CP ordered
  if they contain \CP-ordered events.
  Because of the next rule,
  this rule can be expressed simply as follows:
  \CPOrdered{\Release{m}{i}}{\Acquire{m}{j}} if \CPOrdered{\Acquire{m}{i}}{\Release{m}{j}}.

  \item \CP is closed under left and right composition with \HB.
  That is,
  \CPOrdered{e}{e'} if $\exists e'' \mid \CPOrdered{\HBOrdered{e}{e''}}{e'}$ or
  if $\exists e'' \mid \HBOrdered{\CPOrdered{e}{e''}}{e'}$
\end{enumerate}
\label{def:cp}
\end{samepage}
\end{definition}

\noindent
The rest of this article refers to the above rules of the \CP definition
as \emph{\Rules (a), (b), and (c)}.
An execution trace \tr has a \emph{\CP-race} 
if it has two events $e \totalOrder e'$ such that 
$\conflicts{e}{e'} \land \NotPOOrdered{e}{e'} \land \NotCPOrdered{e}{e'}$.

\paragraph{Examples.}

In the execution traces in
Figures~\ref{fig:SimpleExample:has-CP-race} and \ref{fig:SimpleExample:no-CP-race} (page~\pageref{fig:SimpleExample:has-CP-race}),
\HBOrdered{\Write{x}{1}}{\Read{x}{1}{T2}} because 
$\POOrdered{\Write{x}{1}}{\HBOrdered{\Release{m}{1}}{\POOrdered{\allowbreak\Acquire{m}{2}}{\Read{x}{1}{T2}}}}$.
In Figure~\ref{fig:SimpleExample:no-CP-race},
\CPOrdered{\Release{m}{1}}{\Acquire{m}{2}} by \RuleA, and it follows that
\CPOrdered{\Write{x}{1}}{\Read{x}{1}{T2}} by \RuleC because
$\HBOrdered{\Write{x}{1}}{\CPOrdered{\Release{m}{1}}{\HBOrdered{\Acquire{m}{2}}{\Read{x}{1}{T2}}}}$. 
In contrast,
in Figure~\ref{fig:SimpleExample:has-CP-race} the critical sections do not have conflicting accesses,
so \NotCPOrdered{\Release{m}{1}}{\Acquire{m}{2}} and 
\NotCPOrdered{\Write{x}{1}}{\Read{x}{1}{T2}}.

Next consider the execution in Figure~\ref{fig:LessComplexExample} 
(page~\pageref{fig:LessComplexExample}),
ignoring the rightmost column (explained in Section~\ref{sec:motivation}).
The accesses to \code{x} are \CP
ordered through the following logic:
\CPOrdered{\Release{u}{1}}{\Acquire{u}{2}} by \RuleA implies
\CPOrdered{\Acquire{m}{1}}{\Release{m}{2}} by \RuleC, which implies
\CPOrdered{\Release{m}{1}}{\Acquire{m}{2}} by \RuleB.
Since \HBOrdered{\Write{x}{1}}{\CPOrdered{\Release{m}{1}}{\HBOrdered{\Acquire{m}{2}}{\Read{x}{1}{T3}}}},
\CPOrdered{\Write{x}{1}}{\Read{x}{1}{T3}} by \RuleC.

Prior work proves that the \CP relation is sound\footnote{Following prior work on predictive analysis~\cite{causally-precedes,rvpredict-pldi-2014},
a \emph{sound} analysis reports only true data races.}~\cite{causally-precedes}.
In particular, if a \race exists, there exists an execution that has
an \HB-race (two conflicting accesses unordered by \HB) or a deadlock~\cite{causally-precedes}.

\subsection{Limitations of Recursive Ordering}
\label{sec:motivation}

This article targets the challenge of developing an \emph{online} analysis for
tracking the \CP relation and detecting \CP-races. An online analysis must
(1) compute \CP soundly and completely;
(2) maintain analysis state that summarizes the execution so far, without needing to maintain and refer to the entire execution trace; and
(3) analyze real program execution traces using time and space that is acceptable
for heavyweight in-house testing.

The main difficulty in tracking the \CP relation online 
is in summarizing the execution so far as analysis state.
An analysis can compute the \PO and \HB relations for events executed so far based \emph{only on the events executed so far}.
In contrast, an online \CP analysis must handle the fact that \CP may order two events because of later events.
For example, \CPOrdered{e}{e'} only because of a future event $e''$ ($e' \totalOrder e''$);
we provide a concrete example shortly.
The analysis must summarize the \emph{possible} order between $e$ and $e'$ at least
until $e''$ executes, without needing access to the entire execution trace.
Smaragdakis \etal\ explain the inherent challenge of
developing an online analysis for \CP as follows~\cite{causally-precedes}:

\begin{center}
\begin{tabular}{@{}p{0.9\linewidth}}
\em CP reasoning, based on \textnormal{[the definition of \CP{}],} is
highly recursive. Notably, Rule (c) can feed into Rule (b), which can feed back
into Rule (c). As a result, we have not implemented CP using techniques such as
vector clocks, nor have we yet discovered a full CP implementation that only
does online reasoning (i.e., never needs to ``look back'' in the execution
trace).
\end{tabular}
\end{center}

\noindent
Smaragdakis \etal's \CP algorithm encodes the recursive definition
of \CP in Datalog,
guaranteeing polynomial-time execution in the size of the execution trace.
However, the algorithm is inherently \emph{offline} because it
fundamentally needs to ``look back'' at the entire execution trace.
Experimentally, Smaragdakis \etal\ find that their algorithm does not scale to full program traces.
Instead, they limit their algorithm's computation 
to bounded windows of 500 consecutive events~\cite{causally-precedes}.

Figure~\ref{fig:LessComplexExample} illustrates the 
challenge of developing an online analysis that computes \CP soundly and completely
while handling the recursive nature of the \CP definition.
The last column shows the \emph{orderings}
relevant to \CPOrdered{\Write{x}{1}}{\Read{x}{1}{T3}}
that are ``knowable'' after each event $e$.
More formally, these are orderings that exist for a subtrace 
comprised of events up to and including $e$.

As Section~\ref{sec:definitions} explained,
\CPOrdered{\Write{x}{1}}{\Read{x}{1}{T3}} because \CPOrdered{\Release{m}{1}}{\Acquire{m}{2}}.
However, at \Read{x}{1}{T3}, an online analysis \emph{cannot} determine
that \CPOrdered{\Release{m}{1}}{\Acquire{m}{2}} 
(and thus \CPOrdered{\Write{x}{1}}{\Read{x}{1}{T3}})
based on the execution subtrace so far.
Not until \Read{y}{1}{T2} is it knowable that
\CPOrdered{\Release{u}{1}}{\Acquire{u}{2}} and thus
\CPOrdered{\Acquire{m}{1}}{\Release{m}{2}},
\CPOrdered{\Release{m}{1}}{\Acquire{m}{2}}, and
\CPOrdered{\Write{x}{1}}{\Read{x}{1}{T3}}.
A sound and complete online analysis for \CP must track analysis state that captures
ordering once it is knowable without maintaining the entire execution trace.

\input{examples/LessComplexExample}

Alternatively, consider instead that \thr{T1} executed the critical section on \code{u}
\emph{before} the critical section on \code{m}. In that subtly different
execution, \NotCPOrdered{\Write{x}{1}}{\Read{x}{1}{T3}}. 
A sound and complete online analysis for \CP must track analysis state that
captures the difference between these two execution variants.

We note that more challenging examples exist.
For instance, it is possible to modify the example 
so that it is unknowable even at \Release{m}{2} that \CPOrdered{\Acquire{m}{1}}{\Release{m}{2}}.
Section~\ref{sec:analysis} presents three such examples.

%% file: examples/LessComplexExample.tex
\begin{figure}
\smaller
\centering
\begin{tabular}{lll|l}
\thr{T1} & \thr{T2} & \thr{T3} & Relevant orderings ``knowable'' after event\\\hline 
\Write{x}{1} & & & \\
\Acquire{m}{1} & & & \HBOrdered{\Write{x}{1}}{\Release{m}{1}} ~~~~~ [knowable at \Acquire{m}{1} since \Release{m}{1} is inevitable]\\
\Release{m}{1} & & & \\
\Acquire{u}{1} & & & \HBOrdered{\Acquire{m}{1}}{\Release{u}{1}}\\
\Write{y}{1} & & & \POOrdered{\Acquire{u}{1}}{\POOrdered{\Write{y}{1}}{\Release{u}{1}}}\\
\Release{u}{1}\tikzmark{1} & & & \\
& \Acquire{m}{2} & & \\
& \Acquire{v}{1} & & \\
& \Release{v}{1} & & \\
& & \Acquire{v}{2} & \\
& & \Release{v}{2} & \\
& & \Read{x}{1}{T3} & \HBOrdered{\Acquire{m}{2}}{\Read{x}{1}{T3}}\\[.3em]
& \tikzmark{2}\Acquire{u}{2} & & \HBOrdered{\Acquire{u}{2}}{\Release{m}{2}}\\
& \Read{y}{1}{T2} & & \POOrdered{\Acquire{u}{2}}{\POOrdered{\Read{y}{1}{T2}}{\Release{u}{2}}},
\CPOrdered{\Release{u}{1}}{\Acquire{u}{2}},\\[.3em]&&&
\CPOrdered{\Acquire{m}{1}}{\Release{m}{2}}, \CPOrdered{\Release{m}{1}}{\Acquire{m}{2}}, 
\CPOrdered{\Write{x}{1}}{\Read{x}{1}{T3}}\\
& \Release{u}{2} & &\\
& \Release{m}{2} & &
\end{tabular}
\textlink{1}{2}{\CP}

%

\caption{An example execution in which \CPOrdered{\Write{x}{1}}{\Read{x}{1}{T3}}.
The last column shows, for each event $e$, orderings relevant to \CPOrdered{\Write{x}{1}}{\Read{x}{1}{T3}}
for the subtrace up to and including $e$.
The arrow represents a \CP ordering established by \RuleA of the \CP definition.}
\label{fig:LessComplexExample}
\end{figure}

%% file: Overview.tex
\section{\Raptor Overview}
\label{sec:overview}

\emph{\Raptor}
(\underline{Ra}ce \underline{p}redic\underline{tor})
is a new online dynamic analysis that
computes the \CP relation soundly and completely
by maintaining analysis state that captures \CP orderings
knowable for a subtrace of events up to and including the latest event in the execution.
This section overviews the components of \raptor's analysis state.

\paragraph{Terminology.}

Throughout the rest of the article,
we say that an event $e$ is \emph{\CP ordered to} a lock \code{m} or thread \thr{T} if
there exists an event $e'$ that
releases \code{m} ($\exists i \mid e' = \Release{m}{i}$) or
is executed by \thr{T} ($\getThread{e'} = \thr{T}$), respectively,
and \CPOrdered{e}{e'}.
This property, in turn, implies that for any future event $e''$ (\ie, $e' \totalOrder e''$), \CPOrdered{e}{e''}
if $e''$ acquires \code{m} ($e'' = \Acquire{m}{j}$) or is executed by \thr{T} ($\getThread{e''} = \thr{T}$), respectively,
since \CP composes with \HB.
Similarly, $e$ is \emph{\HB ordered to} a lock or thread if the same conditions hold for \HBRelation instead of \CPRelation.

\paragraph{\Sets.}
\label{Sec:overview:locksets}



Existing \HB analyses typically represent analysis state using \emph{vector clocks}~\cite{vector-clocks,multirace,fasttrack}.
Since the \CP relation conditionally orders critical sections,
conditional information is required on synchronization objects to accurately track \CP.
Using sets to track the \HB and \CP relations in terms of synchronization objects
naturally manages conditional information, compared with using vector clocks.
\Raptor's analysis state is represented by 
sets containing synchronization objects---locks and threads---that represent \CP, \HB, and \PO orderings.
For example, if a lock \code{m} is an element of the \HB set $\HBLockSet{x}{8}$, 
it means that the $8$th write of \code{x} event, \Write{x}{8}, is \HB ordered to \code{m}.
Similarly, the thread element $\CPThread{T2}{} \in \CPLockSetRead{y}{3}{T1}$ 
means that the event \Read{y}{3}{T1} (a read by \thr{T1} to \code{y} between the 3rd and 4th writes to \code{y}) is \CP ordered to thread \thr{T2}.
\Raptor's sets are most related to the sets used by \emph{Goldilocks},
a sound and complete \emph{\HB} data race detector~\cite{goldilocks-pldi-2007} (Section~\ref{Sec:related:goldilocks}).

\paragraph{\Sets for each access to a variable.}

As implied above, rather than each variable \OwnerElement{x}{} having \CP, \HB, and \PO sets,
every \emph{access} \Write{x}{i} and \Read{x}{i}{T}
has its own \CP, \HB, and \PO sets.
Per-access \emph{\CP} sets are necessary because of the nature of the \CP relation:
at \Write{x}{i+1}, it is not in general knowable whether
\CPOrdered{\Write{x}{i}}{\Write{x}{i+1}} or
$\forall_\thr{T} \, \CPOrdered{\Read{x}{i}{T}}{\Write{x}{i+1}}$.
Similarly, it is not generally knowable at \Read{x}{i}{T} whether \CPOrdered{\Write{x}{i}}{\Read{x}{i}{T}}.
In Figure~\ref{fig:LessComplexExample},
even after \Read{x}{1}{T3} executes,
\raptor must continue to maintain \sets for \Write{x}{1} 
because \CPOrdered{\Write{x}{1}}{\Read{x}{1}{T3}} has not yet been established.

Maintaining per-access sets would seem to require massive time and space (proportional to the length of the execution),
making it effectively an offline analysis like prior work's \CP analysis~\cite{causally-precedes}.
However, as we show in Section~\ref{sec:removal}, \goldCP can safely remove sets under detectable conditions,
\eg, it can remove \Write{x}{i}'s sets
once it determines that \CPOrdered{\Write{x}{i}}{\Write{x}{i+1}}
and $\forall_{\thr{T}} \, \CPOrdered{\Write{x}{i}}{\Read{x}{i}{T}}$.

\paragraph{\Sets for lock acquires.}

\Raptor tracks \CP, \HB, and \PO sets not just for variable accesses, 
but also for lock acquire operations
to compute \CP order by \RuleB
(\ie, \CPOrdered{\Acquire{m}{i}}{\Release{m}{j}} implies \CPOrdered{\Release{m}{i}}{\Acquire{m}{j}}).
For example, $\CPThread{T3} \in \CPLockSet{m}{i}$ means
the event \Acquire{m}{i} is \CP ordered to thread \code{T3}.

Similar to sets for variable accesses,
maintaining a \CP, \HB, and \PO \set for each lock acquire 
might consume high time and space proportional to the execution's length.
In Section~\ref{sec:removal}, we show how \raptor 
can safely remove an acquire \Acquire{m}{i}'s \sets once they are no longer needed---once
no other \CP ordering is dependent on the possibility 
of \Acquire{m}{i} being \CP ordered with a future \Release{m}{}.

\paragraph{Conditional \CP sets.}

As mentioned earlier, 
it is unknowable in general at an event $e'$ whether \CPOrdered{e}{e'}.
This recursive nature of the \CP definition prevents immediate determination of \CP ordering
at $e'$.
This delayed knowledge is unavoidable due to \RuleB,
which states that \CPOrdered{\Release{m}{i}}{\allowbreak\Acquire{m}{j}} if
\CPOrdered{\Acquire{m}{i}}{\Release{m}{j}}.
A \CP ordering might not be known until \Release{m}{j} executes---or even longer
because \RuleC can ``feed into'' \RuleB, which can
feed back into \RuleC~\cite{causally-precedes}. 

\Raptor maintains \emph{conditional \CP} (\CCP) \sets to track the fact that,
at a given event in an execution,
\CP ordering may or may not exist, 
depending on whether some \emph{other} \CP ordering exists.
For example,
an element \CCPLock{n}{m}{j} (or \CCPThread{T}{m}{j}) 
in the \CCP \set \CCPLockSet{x}{i} means that
\Write{x}{i} is \CP ordered to lock \code{n} (or thread \thr{T})
\textbf{if} \CPOrdered{\Acquire{m}{j}}{\Release{m}{k}} 
for some future event \Release{m}{k}.

\medskip
\noindent
In contrast with the above,
\emph{Goldilocks} does not need or use \sets for
each variable access, \sets for lock acquires, or conditional \sets, 
since it maintains sets that track only the \HB relation~\cite{goldilocks-pldi-2007}.

\paragraph{Outline of \goldCP presentation.}

Section~\ref{sec:state} describes \raptor's \sets and their elements in detail, and
it presents invariants maintained by \raptor's \sets at every event in an execution trace.
Section~\ref{sec:analysis} introduces the \raptor analysis that
adds and, in some cases, removes \set elements at each execution event.
Section~\ref{sec:removal} describes how \raptor removes ``obsolete'' \sets and detects \CP-races.

%% file: State.tex
\section{\Raptor's Analysis State and Invariants}
\label{sec:state}

This section describes the analysis state that \raptor maintains.
Every \emph{set owner} $\rho$, 
which can be a variable write instance \OwnerElement{x}{i},
a variable read instance \OwnerElementRead{x}{i}{T},
or lock acquire instance \OwnerElement{m}{i},
has the following \sets:
\POLockSetPlain{\rho}{},
\HBLockSetPlain{\rho}{},
\CPLockSetPlain{\rho}{}, and
\CCPLockSetPlain{\rho}{}.
In general, elements of each \set are threads \thr{T} and locks \code{m}, with a few caveats:
\HBLockSetPlain{\rho}{} maintains an index for each lock element (\eg, \HBLock{m}{j}), and each
\CCPLockSetPlain{\rho}{} element includes an associated lock instance upon which \CP
ordering is conditional (\eg, \CCPLock{m}{n}{j} or \CCPThread{T}{n}{j}).
In addition, each \set for a variable write instance \OwnerElement{x}{i} 
and read instance \OwnerElementRead{x}{i}{T} can contain a special element $\xi$,
which indicates ordering between \Write{x}{i} and \Write{x}{i+1} and
between \Read{x}{i}{T} and \Write{x}{i+1}.
Similarly, each set for \OwnerElement{x}{i} can also contain a special element $\xiT{T}$ for each thread \thr{T},
which indicates ordering between \Write{x}{i} and \Read{x}{i}{T}.
Since knowledge of \CP ordering may be delayed,
a write or read instance could establish \CP order to a thread \thr{T}
at an event later than the conflicting write or read instance.
The special elements are necessary to distinguish \CP ordering
to the conflicting write or read instance
from \CP ordering to a later event.

Figure~\ref{Fig:invariants} shows invariants 
that the \raptor analysis maintains for every set owner $\rho$.
The rest of this section explains these invariants in detail,
using events $e$, \exi, \exiT{T}, and \erho as defined in the figure.

\input{Invariants}

\subsection{\POFull \Set: \boldmath\POLockSetPlain{\rho}{}\unboldmath}\label{sub:sec:po-state}

According to the \emph{\POinv} in Figure~\ref{Fig:invariants},
\POLockSetPlain{\rho}{} contains all threads that the event \erho is \PO ordered to.
That said, we know from the definition of \PO that \erho will be \PO ordered to
only one thread (the thread that executed $\erho$).
In addition, for any $\rho = \OwnerElement{x}{h}$ or $\rho = \OwnerElementRead{x}{h}{T}$,
\POLockSetPlain{\rho}{} may contain the special element $\xi$,
indicating that \Write{x}{h} or \Read{x}{h}{T}, respectively, 
is \PO ordered to the next write access to \code{x} by the same thread,
\ie, \POOrdered{\Write{x}{h}}{\Write{x}{h+1}} or \POOrdered{\Read{x}{h}{T}}{\Write{x}{h+1}}.
Similarly, for any $\rho = \OwnerElement{x}{h}$,
\POLockSetPlain{\rho}{} may contain the special element \xiT{T},
indicating that \Write{x}{h} is \PO ordered to the next read access to \code{x} by thread \code{T},
\ie, \POOrdered{\Write{x}{h}}{\Read{x}{h}{T}}.
Note that \raptor does not really need $\xi$ and \xiT{T} to indicate \POOrdered{\Write{x}{h}}{\Write{x}{h+1}},
\POOrdered{\Read{x}{h}{T}}{\Write{x}{h+1}}, or \POOrdered{\Write{x}{h}}{\Read{x}{h}{T}}, 
since \PO order is knowable at the next (read/write) access,
but \raptor uses these elements for consistency with the \CP and \CCP \sets, which \emph{do} need $\xi$ and \xiT{T}
as explained later in this section.

\subsection{\HBFull \Set: \boldmath\HBLockSetPlain{\rho}{}\unboldmath}\label{sub:sec:hb-state}

The \HBLockSetPlain{\rho}{} \set contains threads and locks that the event \erho is \HB ordered to.
Figure~\ref{Fig:invariants} states three invariants for \HBLockSetPlain{\rho}{}:
the \emph{\invHB}, \emph{\invHBindex}, and \emph{\invHBcriticalsection} invariants.

The \HBinv defines which threads and locks are in \HBLockSetPlain{\rho}{}.
If \erho is \HB ordered to a thread or lock,
then \HBLockSetPlain{\rho}{} contains that thread or lock.
This property implies that \erho will be 
\HB ordered to any future event that executes on the same thread
or acquires the same lock, respectively.
Similar to \PO sets, for $\rho = \OwnerElement{x}{h}$ or $\rho = \OwnerElementRead{x}{h}{T}$,
$\xi \in \HBLockSetPlain{\rho}{}$ means
\HBOrdered{\Write{x}{h}}{\Write{x}{h+1}} or \HBOrdered{\Read{x}{h}{T}}{\Write{x}{h+1}}, respectively.
Additionally, for $\rho = \OwnerElement{x}{h}$,
$\xiT{T} \in \HBLockSetPlain{\rho}{}$ means
\HBOrdered{\Write{x}{h}}{\Read{x}{h}{T}}.
Though the $\xi$ and \xiT{T} elements are superfluous (HB ordering is knowable at the next (read/write) access),
\raptor maintains these elements for consistency with the \CP and \CCP \sets that need it.

According to the \HBindexinv,
every lock \code{m} in \HBLockSetPlain{\rho}{} has a superscript $i$ (\eg, \HBLock{m}{i}) that specifies the earliest
release of \code{m} that \erho is \HB ordered to.
For example, $\HBLock{m}{i} \in \HBLockSetPlain{\rho}{}$ means that
\HBOrdered{\erho}{\Release{m}{i}} but \NotHBOrdered{\erho}{\Release{m}{i-1}}.
This property tracks
\emph{which} instance of the critical section on lock \code{m}, \OwnerElement{m}{i},
would need to be \CP ordered to \Release{m}{j}
to imply that \CPOrdered{\erho}{\Acquire{m}{j}} (by \Rules~(b) and (c)).

According to the \HBcriticalsectioninv,
for read/write accesses ($\rho = \OwnerElement{x}{h}$ or $\rho = \OwnerElementRead{x}{h}{T}$) only,
\HBLock{m}{i} in \HBLockSetPlain{\rho}{} may have a subscript $*$ (\ie, \LSLock{m}{i}), 
indicating that, in addition to \HBOrdered{\erho}{\Release{m}{i}},
$\erho$ executed \emph{inside} the critical section on lock \OwnerElement{m}{i},
\ie, $\Acquire{m}{i} \POOrdered{}{} \erho \POOrdered{}{} \Release{m}{i}$.
Notationally,
whenever $\LSLock{m}{i} \in \HBLockSetPlain{\rho}{}$,
$\HBLock{m}{i} \in \HBLockSetPlain{\rho}{}$ is also implied.
\Raptor tracks this property to establish \RuleA precisely.

\subsection{\CPFull \Set: \boldmath\CPLockSetPlain{\rho}{}\unboldmath}\label{sub:sec:cp-state}

Analogous to \HBLockSetPlain{\rho}{} for \HB ordering,
each \CPLockSetPlain{\rho}{} \set contains locks and threads 
that the event \erho is \CP ordered to.
However, at an event $e'$, since \RuleB may delay establishing \CPOrdered{\erho}{e'}, 
\CPLockSetPlain{\rho}{} does \emph{not necessarily} contain $\rhoprime$
such that $\comp{\rhoprime}{e'} \land \CPOrdered{\erho}{e'}$.
This property of \CP presents two main challenges.
First, \RuleB may delay establishing \CP order
that is dependent on other \CP orders.
\Raptor introduces the \CCPLockSetPlain{\rho}{} \set (described below) 
to track potential \CP ordering that may be established later.
\Raptor tracks every lock and thread that \erho \emph{is} \CP ordered to,
either eagerly using \CPLockSetPlain{\rho}{} or lazily using \PCPLockSetPlain{\rho}{},
according to the \emph{\CPinv} in Figure~\ref{Fig:invariants}.

Second, as a result of computing \CP lazily,
\raptor may not be able to determine that there is a \CP-race 
between conflicting events \conflicts{\Write{x}{i}}{\Write{x}{i+1}},
\conflicts{\Write{x}{i}}{\Read{x}{i}{T}}, or
\conflicts{\Read{x}{i}{T}}{\Write{x}{i+1}} 
until after the second conflicting access event \Read{x}{i}{T} or \Write{x}{i+1}.
For example, if the analysis adds \CPThread{T} to \CPLockSet{x}{i}
sometime \emph{after} \thr{T} executed \Write{x}{i+1},
that does not necessarily mean that \CPOrdered{\Write{x}{i}}{\Write{x}{i+1}}
(it means only that \Write{x}{i} is \CP ordered to some event by \thr{T} after \Write{x}{i+1}).
\Raptor uses the special thread-like element $\xi$ that
represents the thread \thr{T} \emph{up to event \Write{x}{i+1} only}, so 
$\xi \in \CPLockSet{x}{i}$ or $\xi \in \CPLockSetRead{x}{i}{T}$
only if \CPOrdered{\Write{x}{i}}{\Write{x}{i+1}} or \CPOrdered{\Read{x}{i}{T}}{\Write{x}{i+1}}, respectively.
\Raptor also uses the special thread-like element \xiT{T} that
represents the thread \thr{T} \emph{up to event \Read{x}{i}{T} only}, so
$\xiT{T} \in \CPLockSet{x}{i}$ only if \CPOrdered{\Write{x}{i}}{\Read{x}{i}{T}}.

The \emph{\CPruleAinv} (Figure~\ref{Fig:invariants}) covers the case 
for which \raptor always computes \CP eagerly: 
when two critical sections on the same lock have 
conflicting events,
according to \RuleA.
In this case, the invariant states that if two critical sections, \OwnerElement{m}{h} and \OwnerElement{m}{i},
are \CP ordered by \RuleA alone,
then $\thr{T} \in \CPLockSet{m}{h}$ as soon as the second conflicting access executes.
The \CPruleAinv is useful in proving that \goldCP maintains the \CPinv
(Appendix~\ref{sec:correctness}).

\subsection{\CCPFull \Set: \boldmath\CCPLockSetPlain{\rho}{}\unboldmath}\label{sub:sec:ccp-state}

Section~\ref{sec:overview} overviewed \CCP \sets.
In general,
$\CCPLockPlain{\rhoprime}{n}{k} \in \CCPLockSetPlain{\rho}{}$ 
means that the event \erho is \CP ordered to $\rhoprime$ 
\textbf{if} \CPOrdered{\Acquire{n}{k}}{\Release{n}{j}},
where \OwnerElement{n}{j} is an ongoing critical section 
(\ie, $\Acquire{n}{j} \totalOrder e \land \Release{n}{j} \not\totalOrder e$).

As mentioned above,
the \CPinv
says that for every \CP ordering, \goldCP captures
it eagerly in a \CP \set \textbf{or} lazily in a \PCP \set (or both).
A further constraint, codified in the \emph{\PCPconstraintinv},
is that $\PCPLockPlain{\rhoprime}{n}{k} \in \PCPLockSetPlain{\rho}{}$
\emph{only if} a critical section on lock \code{n} is ongoing.
As Section~\ref{sec:analysis} shows,
when \code{n}'s current critical section ends (at \Release{n}{j}),
\goldCP either (1) determines whether \CPOrdered{\Acquire{n}{k}}{\Release{n}{j}}
or (2) identifies another lock \code{q} that has an ongoing critical section
such that it is correct to add some \PCPLockPlain{\rhoprime}{q}{f}
to \PCPLockSetPlain{\rho}{}.

Like \CPLockSetPlain{\rho}{}, 
when $\rho = \OwnerElement{x}{i}$ or $\rho = \OwnerElementRead{x}{i}{T}$, 
\CCPLockSetPlain{\rho}{} can contain 
special thread-like elements of the form \CCPLockPlain{\xi}{n}{k}.
The element $\CCPLockPlain{\xi}{n}{k} \in \CCPLockSet{x}{i}$ or
$\CCPLockPlain{\xi}{n}{k} \in \CCPLockSetRead{x}{i}{T}$ means that
\CPOrdered{\Write{x}{i}}{\Write{x}{i+1}} \textbf{if} \CPOrdered{\Acquire{n}{k}}{\Release{n}{j}}, or
\CPOrdered{\Read{x}{i}{T}}{\Write{x}{i+1}} \textbf{if} \CPOrdered{\Acquire{n}{k}}{\Release{n}{j}}, respectively,
where \OwnerElement{n}{j} is the current ongoing critical section of \code{n}.
Similarly, for $\rho = \OwnerElement{x}{i}$,
\CCPLockSetPlain{\rho}{} can contain
special thread-like elements of the form \CCPLockPlain{\xiT{T}}{n}{k}.
The element $\CCPLockPlain{\xiT{T}}{n}{k} \in \CCPLockSet{x}{i}$ means that
\CPOrdered{\Write{x}{i}}{\Read{x}{i}{T}} \textbf{if} \CPOrdered{\Acquire{n}{k}}{\Release{n}{j}}
where \OwnerElement{n}{j} is the current ongoing critical section of \code{n}.

\input{examples/SimpleExample_State}

\paragraph{\Raptor state example.} Figure~\ref{fig:SimpleExampleState} shows
updates to \raptor's analysis state at each event for the execution from Figure~\ref{fig:SimpleExample:no-CP-race}.
Directly after \Write{y}{1},
$\HBThread{T1} \in \POLockSet{y}{1}$, satisfying the \POinv;
and $\{\HBThread{T1}, \LSLock{m}{1}\} \in \HBLockSet{y}{1}$,
satisfying the \invHB, \invHBindex, and \invHBcriticalsection invariants
since \POOrdered{\Acquire{m}{1}}{\POOrdered{\Write{y}{1}}{\Release{m}{1}}}.
Directly after \Acquire{m}{2},
\CCPLockSet{x}{1}, \CCPLockSet{m}{1}, and \CCPLockSet{y}{1} contain $\CCPThread{T2}{m}{1}$
satisfying the \invCP and \invPCPconstraint invariants,
capturing the fact that \thr{T1}'s events are \CP ordered to \thr{T2} if \CPOrdered{\Acquire{m}{1}}{\Release{m}{2}}.
Directly after \Read{y}{1}{T2},
$\CPThread{T2} \in \CPLockSet{m}{1}$
satisfies the \CPruleAinv, and
$\CCPThreadPlain{\xiT{T2}}{m}{1} \in \CCPLockSet{y}{1}$
satisfies the \CPinv.
Finally, after \Release{m}{2} and \Read{x}{1}{T2}, 
$\CPThread{T2} \in \CPLockSet{x}{1}$ and $\CPThreadPlain{\xiT{T2}} \in \CPLockSet{x}{1}$,
respectively, satisfying the \CPinv, indicating \CPOrdered{\Write{x}{1}}{\Read{x}{1}{T2}}.

%% file: Invariants.tex
\newcommand\myset[1]{\ensuremath{ \big\{ \, #1 \, \big\} }}
\newcommand\mysetsm[1]{\ensuremath{ \{ \, #1 \, \} }}

\newcommand\comp[2]{\ensuremath{\mathit{appl}(#1,#2)}}

\newcommand\invPO{\textbf{[PO]}\xspace}
\newcommand\invHB{\textbf{[HB]}\xspace}
\newcommand\invHBindex{\textbf{[HB-index]}\xspace}
\newcommand\invHBcriticalsection{\textbf{[HB-cri\-ti\-cal-sec\-tion]}\xspace}
\newcommand\invCP{\textbf{[CP]}\xspace}
\newcommand\invCPruleA{\textbf{[CP-rule-A]}\xspace}
\newcommand\invPCPconstraint{\textbf{[\PCP-con\-straint]}\xspace}
\newcommand\invCCPconstraint{\invPCPconstraint}

\newcommand\inv{invariant\xspace}
\newcommand\POinv{\invPO \inv}
\newcommand\HBinv{\invHB \inv}
\newcommand\HBindexinv{\invHBindex \inv}
\newcommand\HBcriticalsectioninv{\invHBcriticalsection \inv}
\newcommand\CPinv{\invCP \inv}
\newcommand\CPruleAinv{\invCPruleA \inv}
\newcommand\PCPconstraintinv{\invPCPconstraint \inv}
\newcommand\CCPconstraintinv{\PCPconstraintinv}

\newcommand\Access[2]{\ensuremath{ \textcolor{red}{a_{#2}} }}

\begin{figure}
\fbox{
\parbox[c]{0.98\textwidth}{
\flushleft
\small
Let $e$ be any event in the program trace.
The following invariants hold for the point in the trace immediately \emph{before} $e$.

Let $\exi = \Write{x}{h+1}$ if $e = \Write{x}{h}$ or $e = \Read{x}{h}{T}$.
Let $\exiT{T} = \Read{x}{h}{T}$ if $e = \Write{x}{h}$.
Otherwise ($e$ is a lock acquire/release event), 
\exi and \exiT{T} are ``invalid events'' that match no real event.

We define a boolean function $\comp{\rhoprime}{e'}$ that 
evaluates to true iff event $e'$ ``applies to'' set element $\rhoprime$:\\
$\comp{\rhoprime}{e'} \coloneqq
     \begin{cases}
       \getThread{e'} = \thr{T}           & \quad \text{if $\rhoprime$ is a thread \thr{T}} \\
       \exists i \mid e' = \Release{m}{i} & \quad \text{if $\rhoprime$ is a lock \thr{m}} \\
       e' = \exi			              & \quad \text{if $\rhoprime$ is $\xi$} \\
       e' = \exiT{T}					  & \quad \text{otherwise ($\rhoprime$ is \xiT{T})}
     \end{cases}$

The following invariants hold for every set owner $\rho$.
For each set owner $\rho$,
let \erho be the event corresponding to $\rho$, \ie,
$\erho = \Write{x}{h}$ if $\rho = \OwnerElement{x}{h}$, 
$\erho = \Read{x}{h}{T}$ if $\rho = \OwnerElementRead{x}{h}{T}$, or
$\erho = \Acquire{m}{h}$ if $\rho = \OwnerElement{m}{h}$.

\begin{description}

\item[\invPO]
$\POLockSetPlain{\rho}{} = \myset{\rhoprime \mid \rhoprime \text{ is not a lock} \land \big(\exists e' \mid \comp{\rhoprime}{e'} \land \EqPOOrdered{\erho}{e'} \totalOrder e\big)}$
\item[\invHB]
$\HBLockSetPlain{\rho}{} = \myset{\rhoprime \mid \big(\exists e' \mid \comp{\rhoprime}{e'} \land \EqHBOrdered{\erho}{e'} \totalOrder e\big)}$


\item[\invHBindex]
$\HBLock{m}{i} \in \HBLockSetPlain{\rho}{}
\Longleftrightarrow
\big(\NotHBOrdered{\erho}{\Release{m}{i-1}} \land \HBOrdered{\erho}{\Release{m}{i}} \totalOrder e\big)$

\item[\invHBcriticalsection]
$\LSLock{m}{i} \in \HBLockSetPlain{\rho}{}
\Longleftrightarrow
\big(\POOrdered{\POOrdered{\Acquire{m}{i}}{\erho}}{\Release{m}{i}} \land (\rho = \OwnerElement{x}{h} \lor \rho = \OwnerElementRead{x}{h}{T}) \land
\erho \totalOrder e\big)$

\item[\invCP]
$\CPLockSetPlain{\rho}{} \cup \myset{\rhoprime \mid \big( \exists \HBLock{n}{k} \mid \PCPLock{\rhoprime}{n}{k} \in \PCPLockSetPlain{\rho}{} \land \exists j \mid \CPOrdered{\Release{n}{k}}{\Acquire{n}{j}} \totalOrder e\big)} =
\allowbreak
\myset{\rhoprime \mid \big(\exists e' \mid \comp{\rhoprime}{e'} \land \CPOrdered{\erho}{e'} \totalOrder e\big)}$

\item[\invCPruleA]
$\big(
\exists \code{m}, \thr{T}, h, i, e', e'' \mid \rho = \HBLock{m}{h} \land \POOrdered{\POOrdered{\erho}{e'}}{\Release{m}{h}} \totalOrder \POOrdered{\POOrdered{\Acquire{m}{i}}{e''}}{\Release{m}{i}} 
\land \allowbreak e'' \totalOrder e
\land \getThread{e''} = \thr{T} 
\land \conflicts{e'}{e''}\big) \implies 
\CPThread{T} \in \CPLockSetPlain{\rho}{}$

\item[\invPCPconstraint]
$\exists \PCPLockPlain{\rhoprime}{n}{k} \in \PCPLockSetPlain{\rho}{} \implies
\big(\exists j \mid \Acquire{n}{j} \totalOrder e \land \Release{n}{j} \not\totalOrder e\big)$

\end{description}
}}
\caption{The invariants maintained by the \goldCP analysis at every event in the observed total order.}
\label{Fig:invariants}
\end{figure}

%% file: examples/SimpleExample_State.tex
\begin{figure}
\newcommand{\sna}{---} 
\newcommand{\nochange}{\sna}
\newcommand{\sss}[1]{\scaleto{\bf #1}{4pt}}
\newcommand{\setE}[2]{\rule{0pt}{3ex}\ensuremath{\dfrac{\sss{#1}}{#2}}\rule{0pt}{4ex}\xspace}
\smaller
\centering
\begin{tabular}{ll|llllll}
\mc{2}{c|}{Execution} & \mc{6}{c}{Analysis state changes} \\
\thr{T1} & \thr{T2} & \OwnerElement{x}{1} & \OwnerElementRead{x}{1}{T2} & \OwnerElement{y}{1} & \OwnerElementRead{y}{1}{T2} & \OwnerElement{m}{1} & \OwnerElement{m}{2} \\[.3em]\hline
\Write{x}{1} 				&& \setE{\PO{\;}\HB}{\HBThread{T1}} & \sna & \sna & \sna & \sna & \sna \\ 
\Acquire{m}{1} 				&& \nochange & \sna & \sna & \sna & \setE{\PO{\;}\HB}{\HBThread{T1}} & \sna \\ 
\Write{y}{1} 				&& \nochange & \sna & \setE{\PO{\;}\HB}{\HBThread{T1}} \setE{\HB}{\LSLock{m}{1}} & \sna & \nochange & \sna \\
\Release{m}{1}\tikzmark{1}	&& \setE{\HB}{\HBLock{m}{1}} & \sna & \nochange & \sna & \setE{\HB}{\HBLock{m}{1}} & \sna \\
& \tikzmark{2}\Acquire{m}{2}& \setE{\HB}{\HBThread{T2}} \setE{\CCP}{\CCPThread{T2}{m}{1}} & \sna & \setE{\HB}{\HBThread{T2}} \setE{\CCP}{\CCPThread{T2}{m}{1}} & \sna & \setE{\HB}{\HBThread{T2}} \setE{\CCP}{\CCPThread{T2}{m}{1}} & \setE{\PO{\;}\HB}{\HBThread{T2}} \\
& \Read{y}{1}{T2}			& \nochange & \sna & \setE{\HB}{\HBThreadPlain{\xiT{T2}}} \setE{\CCP}{\CCPThreadPlain{\xiT{T2}}{m}{1}} & \setE{\PO{\;}\HB}{\HBThread{T2}} \setE{\HB}{\LSLock{m}{2}} & \setE{\CP}{\CPThread{T2}} & \nochange \\
& \Release{m}{2}			& \setE{\CP}{\CPThread{T2}{\;}\CPLock{m}} $\setminus \setE{\CCP}{\CCPThread{T2}{m}{1}}$ & \sna & \setE{\CP}{\CPThread{T2}{\;}\CPLock{m}{\;}\CPThreadPlain{\xiT{T2}}} $\setminus \setE{\CCP}{\CCPThreadPlain{\{\CPThread{T2}{\;}\xiT{T2}\}}{m}{1}}$ & \nochange & \setE{\CP}{\CPLock{m}} $\setminus \setE{\CCP}{\CCPThread{T2}{m}{1}}$ & \setE{\HB}{\HBLock{m}{2}} \\
& \Read{x}{1}{T2}			& \setE{\HB{\;}\CP}{\CPThreadPlain{\xiT{T2}}} & \setE{\PO{\;}\HB}{\HBThread{T2}} & \nochange & \nochange & \nochange & \nochange 
\end{tabular}
\textlink{1}{2}{\CP}

\caption{\Raptor's analysis state updates for the execution from Figure~\ref{fig:SimpleExample:no-CP-race}.
The cells under \emph{Analysis state changes} show, after each event, the changes to each set owner's sets
(``\nochange'' indicates no changes).
By default, the cells show additions to a set; the prefix ``$\setminus$'' indicates removal from a set.
For example, after \Release{m}{2}, \raptor adds \thr{T2} and \code{m} to \CPLockSet{x}{1} and
removes \CCPLock{T2}{m}{1} from \CCPLockSet{x}{1}.}
\label{fig:SimpleExampleState}
\end{figure}

%% file: Analysis.tex
\section{The \Raptor Analysis}
\label{sec:analysis}

This section details \Raptor, our novel dynamic analysis that 
maintains the invariants shown in Figure~\ref{Fig:invariants} 
and explained in Section~\ref{sec:state}.
For each event $e$ in the observed trace \tr,
\raptor updates its analysis state
by adding and (in some cases) removing elements from each set owner $\rho$'s \sets.
Assuming that the analysis state satisfies the invariants immediately \emph{before} $e$,
then at event $e$, 
\raptor modifies the analysis state 
so that it satisfies the invariants immediately \emph{after} $e$:

\begin{theorem}
After every event, \raptor
maintains the invariants in Figure~\ref{Fig:invariants}.
\label{thm:invariants}
\end{theorem}
\noindent
Appendix~\ref{sec:correctness} proves the theorem.

To represent the new analysis state immediately \emph{after} $e$,
we use the notation $\POLockSetPlain{\rho}{}^+$, $\HBLockSetPlain{\rho}{}^+$, $\CPLockSetPlain{\rho}{}^+$, and $\PCPLockSetPlain{\rho}{}^+$.
Initially,
before \raptor starts updating analysis state at $e$ 
(\ie, immediately \emph{before} $e$),
$\POLockSetPlain{\rho}{}^+ \coloneqq \POLockSetPlain{\rho}{}$,
$\HBLockSetPlain{\rho}{}^+ \coloneqq \HBLockSetPlain{\rho}{}$,
$\CPLockSetPlain{\rho}{}^+ \coloneqq \CPLockSetPlain{\rho}{}$, and
$\CCPLockSetPlain{\rho}{}^+ \coloneqq \CCPLockSetPlain{\rho}{}$.

\paragraph{Initial analysis state.}
\label{Sec:analysis:initialization}

Before the first event in \tr,
each $\rho$'s \sets are initially empty,
\ie, $\POLockSetPlain{\rho}{} = \HBLockSetPlain{\rho}{} = \CPLockSetPlain{\rho}{} = \PCPLockSetPlain{\rho}{} = \emptyset$.
This initial state conforms to Figure~\ref{Fig:invariants}'s invariants for the point in execution before the first event.

To simplify checking for \CP-races, the analysis assumes, for every program variable \code{x},
a ``fake'' initial access \Write{x}{0}.
The analysis initializes \POLockSet{x}{0} to $\{\xi\} \cup \{\xiT{T} \mid \thr{T} \textnormal{ is a thread}\}$;
all other \sets of \OwnerElement{x}{0}
are initially $\emptyset$.
This initial state ensures that the first write access to \code{x}, \Write{x}{1},
appears to be ordered to the prior write access to \code{x} (\Write{x}{0}),
and any read accesses to \code{x} before the first write
(\Read{x}{0}{T})
appear to be ordered to the prior write (\Write{x}{0}),
without requiring special logic to handle this corner case.
(The analysis does not need fake accesses \Read{x}{0}{T} because
the analysis at the first write \Write{x}{1} only checks for ordering with prior reads that actually executed.)

\subsection{Handling Write Events}\label{Sec:analysis:write}

At a write event to shared variable \code{x}, 
\ie, $e = \Write{x}{i}$ by thread \thr{T}, 
the analysis performs the actions in Algorithm~\ref{alg:write}.
The analysis establishes \CP \RuleA;
checks for \PO, \HB, \CP, and conditional \CP (\CCP) order with prior accesses \Write{x}{i-1}
and \Read{x}{i-1}{t} for all threads \thr{t};
and initializes \OwnerElement{x}{i}'s \sets.

\input{WriteEvent}

\paragraph{Establishing \RuleA.}

Lines~\ref{line:write:forall-locks}--\ref{line:write:endfor-locks} of Algorithm~\ref{alg:write}
show how the analysis establishes \RuleA
(conflicting critical sections are \CP ordered).
The helper function \heldBy{T} returns the set of locks currently held by thread \thr{T}
(\ie, locks with active critical sections executed by \thr{T}).
For each lock \code{m} held by \thr{T},
the analysis checks whether a prior conflicting access to \code{x} executed
in an earlier critical section on \code{m}.
\Raptor adds \thr{T} to \CPLockSet{m}{j} if a prior critical section \OwnerElement{m}{j} has executed a conflicting access,
establishing \RuleA, which satisfies the \CPruleAinv (Figure~\ref{Fig:invariants}).
When \thr{T} later releases \code{m},
the analysis will update each \CPLockSetPlain{\rho}{} \set 
that depends on \CPOrdered{\Release{m}{j}}{\Acquire{m}{k}},
as Section~\ref{Sec:analysis:release} describes.

\paragraph{Checking ordering with prior accesses.}

Lines~\ref{line:write:POEdge-Exists}--\ref{line:write:end-xi} of Algorithm~\ref{alg:write}
add $\xi$ to each \set of \OwnerElement{x}{i-1} that already contains \thr{T},
indicating \PO, \HB, and/or \CP ordering from \Write{x}{i-1} to \Write{x}{i},
satisfying the invariants, \eg, if \HBOrdered{\Write{x}{i-1}}{\Write{x}{i}}, then $\xi \in \HBLockSet{x}{i-1}$ (part of the \HBinv).
Notably, for any \OwnerElement{m}{j} such that $\PCPThread{T}{m}{j} \in \PCPLockSet{x}{i-1}$,
\CPOrdered{\Write{x}{i-1}}{\Write{x}{i}} if \CPOrdered{\Acquire{m}{j}}{\Release{m}{k}}
(where \OwnerElement{m}{k} is the current critical section on \code{m}),
and so the analysis adds
\CCPThreadPlain{\xi}{m}{j} to \PCPLockSet{x}{i-1}.

Similar to the prior write access, 
\raptor checks for ordering with each \Read{x}{i-1}{t} by each thread \thr{t}.
In general, reads are not totally ordered in a \race-free execution.
Thus \raptor must check for ordering between \Write{x}{i} and 
each prior read by another thread \Read{x}{i-1}{t}
($\Write{x}{i-1} \totalOrder \Read{x}{i-1}{t} \totalOrder \Write{x}{i}$). 
Lines~\ref{line:write:forall-read-threads}--\ref{line:write:forall-read-threads:end} of Algorithm~\ref{alg:write}
add $\xi$ to indicate \PO, \HB, \CP, and/or \CCP ordering 
from \Read{x}{i-1}{t} by each thread \thr{t} to \Write{x}{i},
satisfying the invariants.

\paragraph{Initializing \sets for current access.}

Lines~\ref{line:write:initialize-PO}--\ref{line:write:initialize-HB} 
initialize \PO and \HB \sets for \OwnerElement{x}{i}.
(Before this event, all \sets for \OwnerElement{x}{i} are $\emptyset$.)
In addition to adding \thr{T} to \POLockSet{x}{i} and \HBLockSet{x}{i},
the analysis adds \LSLock{m}{j} to \HBLockSet{x}{i} for 
each ongoing critical section on \OwnerElement{m}{j} by \thr{T},
satisfying the \HBcriticalsectioninv.

\subsection{Handling Read Events}\label{sec:reads}

At a read to shared variable \code{x}, \ie, $e = \Read{x}{i}{T}$, 
the analysis performs the actions in Algorithm~\ref{alg:read},
which is analogous to Algorithm~\ref{alg:write}.
The analysis establishes \CP \RuleA; 
checks for \PO, \HB, \CP, and \CCP order with the prior access \Write{x}{i};
and initializes \OwnerElementRead{x}{i}{T}'s \sets.

\input{ReadEvent}

\paragraph{Establishing \RuleA.}

Lines~\ref{line:read:RuleA:start}--\ref{line:read:RuleA:end} of Algorithm~\ref{alg:read}
add \thr{T} to \CPLockSet{m}{j} if a prior critical section executed a conflicting write access to \code{x},
establishing \RuleA, satisfying the \CPruleAinv (Figure~\ref{Fig:invariants}).
As we mentioned for write events, when \thr{T} later releases lock \code{m},
the analysis will update the \CPLockSet{$\rho$}{} \sets that depend on 
\Release{m}{j} \CPOrdered{}{} \Acquire{m}{k}, as Section~\ref{Sec:analysis:release} describes.

\paragraph{Checking ordering with prior write access.}

To represent \CPOrdered{\Write{x}{i}}{\Read{x}{i}{T}},
lines~\ref{line:read:establishPO}--\ref{line:read:establishCP} add \xiT{T} to each \set
of \OwnerElement{x}{i} that already contains \thr{T},
indicating \PO, \HB, \CP, and/or \CCP ordering from \Write{x}{i} to \Read{x}{i}{T},
satisfying the invariants.

\paragraph{Initializing \sets for current read access.}

Lines~\ref{line:read:initialize-PO}--\ref{line:read:initialize-HB} 
initialize \PO and \HB \sets for \OwnerElementRead{x}{i}{T}.
In addition to adding \thr{T} to \POLockSetRead{x}{i}{T} and \HBLockSetRead{x}{i}{T},
the analysis adds \LSLock{m}{j} to \HBLockSetRead{x}{i}{T} for 
each ongoing critical section on \OwnerElement{m}{j} by \thr{T},
satisfying the \HBcriticalsectioninv.

If a thread \thr{T} performs multiple reads to \code{x} between \Write{x}{i} and \Write{x}{i+1},
\raptor only needs to track \sets for the latest read:
if the earlier read races with \Write{x}{i+1}, then so does the later read.
Thus \raptor maintains \OwnerElementRead{x}{i}{T}'s \sets for the latest \Read{x}{i}{T} only,
which requires resetting \OwnerElementRead{x}{i}{T}'s \CP and \CCP \sets to $\emptyset$ on each read (line~\ref{line:read:resetRead}).

\subsection{Handling Acquire Events}\label{Sec:analysis:acquire}

At an acquire of a lock \code{m}, \ie, $e = \Acquire{m}{i}$ by thread \thr{T},
the analysis performs the actions in Algorithm~\ref{alg:acquire}. 
The analysis establishes \HB and \CP ordering from \code{m} to \thr{T} for all $\rho$; 
adds \CCPLockSetPlain{\rho}{} elements for conditionally \CP (\CCP) ordered critical sections;
and initializes \OwnerElement{m}{i}'s \sets.

\input{AcquireEvent}

\paragraph{Establishing \HB and \CP order}

Both \HB and \CP are closed under right-composition with \HB.
Thus, after the current event $e = \Acquire{m}{i}$ by \thr{T},
any $\erho$ that is \HB or \CP ordered to \code{m} 
(\ie, \HBOrdered{\erho}{\Release{m}{j}} or \CPOrdered{\erho}{\Release{m}{j}} for some $j < i$) 
is now also \HB or \CP ordered,
respectively, to \thr{T}.
Specifically, if $\CPLock{m} \in \CPLockSetPlain{\rho}{}$, then
\CPOrdered{\erho}{\Acquire{m}{i}} and the analysis adds \thr{T} to \CPLockSetPlain{\rho}{}
(\lines{\ref{line:acq:CP-condition}}{\ref{line:acq:end-CP-condition}}), satisfying the \CPinv.
Similarly, \lines{\ref{line:acq:HB-condition}}{\ref{line:acq:add-HBThread}} 
establishes \HB order from \code{m} to \thr{T},
satisfying the \HBinv.
For the condition at line~\ref{line:acq:HB-condition},
recall that $\HBLock{m}{j} \in \HBLockSetPlain{\rho}{}$ if $\LSLock{m}{j} \in \HBLockSetPlain{\rho}{}$.


\paragraph{Establishing \CCP order.}

The analysis establishes \CCP order
at \lines{\ref{line:acq:CCP-condition}}{\ref{line:acq:end-CCP-condition}}.
For any critical section on lock \HBLock{n}{k},
if $\CCPLock{m}{n}{k} \in \CCPLockSetPlain{\rho}{}$,
then \CPOrdered{\erho}{\Release{m}{i-1}}
if \CPOrdered{\Acquire{n}{k}}{\Release{n}{j}}
where \HBLock{n}{j} is an ongoing critical section.
After the current event \Acquire{m}{i}, since \HB right-composes with \CP,
\CPOrdered{\erho}{\Acquire{m}{i}} if \CPOrdered{\Acquire{n}{k}}{\Release{n}{j}}.
Thus, the analysis adds \PCPThread{T}{n}{k} to \CCPLockSetPlain{\rho}{}.

Additionally, if \HBOrdered{\erho}{\Release{m}{j}} for some $j < i$,
then by \RuleB,
\CPOrdered{\erho}{\Acquire{m}{i}}
\textbf{if} \Acquire{m}{j} \CPOrdered{}{} \Release{m}{i}.
Line~\ref{line:acq:add-CCPThread} handles
this case by adding \CCPThread{T}{m}{j} to \CCPLockSetPlain{\rho}{}
when \HBOrdered{\erho}{\Release{m}{j}},
satisfying the \CCPconstraintinv.

\paragraph{Initializing \sets for current lock.}

Lines~\ref{line:acq:initialize-PO}--\ref{line:acq:initialize-HB} initialize
\OwnerElement{m}{i}'s \sets.
The analysis adds \thr{T} to \POLockSet{m}{i} and \HBLockSet{m}{i},
since \OwnerElement{m}{i} will be \PO and \HB ordered to any event by \thr{T} after \Acquire{m}{i},
satisfying the \invPO and \HBinv{}s.
(The analysis never needs or uses \POLockSet{m}{i}.
\Raptor adds \thr{T} to \POLockSet{m}{i}
only to satisfy Figure~\ref{Fig:invariants}'s \POinv.)

\subsection{Handling Release Events}
\label{Sec:analysis:release}

At a lock release, \ie, $e = \Release{m}{i}$ by thread \thr{T},
the analysis performs the actions in Algorithm~\ref{alg:pre-release}, called the ``pre-release'' algorithm,
followed by the actions in Algorithm~\ref{alg:release}, called the ``release'' algorithm.
We divide \raptor's analysis into two algorithms 
to separate the changes to \CCPLockSetPlain{\rho}{} elements: 
the pre-release algorithm adds elements to \CPLockSetPlain{\rho}{} and \CCPLockSetPlain{\rho}{},
and the release algorithm uses the updated \sets.

\subsubsection{Pre-release Algorithm (Algorithm~\ref{alg:pre-release})}

Due to \RuleB, whether \erho is \CP ordered to some lock or thread
can depend on whether \CPOrdered{\Acquire{m}{j}}{\Release{m}{i}} (for some $j$).
The pre-release algorithm establishes \RuleB by updating \CP and \CCP \sets 
that depend on whether \CPOrdered{\Acquire{m}{j}}{\Release{m}{i}}.
Lines~\ref{line:prerel:CCP-condition}--\ref{line:prerel:end-CCP-condition}
establish \RuleB directly by detecting that \CPOrdered{\Acquire{m}{j}}{\Release{m}{i}},\footnote{More precisely,
the analysis checks for any $l \ge j \mid \CPOrdered{\Acquire{m}{l}}{\Release{m}{i}}$,
since \HB left-composes with \CP.}
then updating \CPLockSetPlain{\rho}{} \sets for every dependent element \rhoprime,
satisfying the \CPinv.

\input{PreReleaseEvent} 

Even if \CPOrdered{\Acquire{m}{j}}{\Release{m}{i}}, it may not be knowable at
\Release{m}{i} because it may be dependent on if \CPOrdered{\Acquire{n}{k}}{\Release{n}{h}}
(where \OwnerElement{n}{h} is the ongoing critical section on \code{n}).
\Lines{\ref{line:prerel:CCP-transfer-condition}}{\ref{line:prerel:end-CCP-transfer-condition}}
detect such cases and ``transfer'' the \PCP dependence from \code{m} to \code{n}.
This transfer is necessary because \code{m}'s critical section is ending, and the release algorithm
will remove all \PCPLockPlain{\rhoprime}{m}{j} elements.
After the \emph{pre-release} algorithm, Figure~\ref{Fig:invariants}'s invariants
still hold---for the point in time \emph{prior to $e$}.

\subsubsection{Release Algorithm (Algorithm~\ref{alg:release})}

The release algorithm operates on the analysis state 
modified by the pre-release algorithm.
The analysis establishes \CP and \HB order from \thr{T} to \code{m};
and removes all \PCP elements that are dependent on \code{m}.

\input{ReleaseEvent}

Since \HB and \CP are closed under right-composition with \HB,
if $\erho$ is \HB or \CP ordered to \thr{T} before \Release{m}{i},
then $\erho$ is \HB or \CP ordered to \code{m} after \Release{m}{i}.
The analysis establishes order from \thr{T} to \code{m} (\lines{\ref{line:rel:CP-condition}}{\ref{line:rel:end-HB-condition}}),
satisfying the \invHB and \invCP invariants.
This is analogous to lines~\ref{line:acq:CP-condition}--\ref{line:acq:end-HB-condition} of Algorithm~\ref{alg:acquire}'s
establishing order from \code{m} to \code{T} at \Acquire{m}{i}.
\Lines{\ref{line:rel:CCP-condition}}{\ref{line:rel:end-CCP-condition}}
establishes \CCP order from \thr{T} to \code{m}
for any lock instance \HBLock{n}{k} ($\code{n} \ne \code{m}$) such that
$\CCPThread{T}{n}{k} \in \CCPLockSetPlain{\rho}{}$
similar to Algorithm~\ref{alg:acquire}'s line~\ref{line:acq:CCP-condition}.

Line~\ref{line:rel:removal} removes all \CCP elements dependent on \code{m},
\ie, all \CCPLockPlain{\rhoprime}{m}{j} elements from \CCPLockSetPlain{\rho}{},
satisfying the \CCPconstraintinv.
Removal is necessary: it would be incorrect for the analysis
to retain these elements, \eg, \CPOrdered{\Acquire{m}{j}}{\Release{m}{i+1}} does \emph{not}
imply that \erho is \CP ordered to $\rhoprime$.
It is safe to remove all \CCPLockPlain{\rhoprime}{m}{j} elements,
even if \CPOrdered{\Acquire{m}{j}}{\Release{m}{i}} is not knowable at $e = \Release{m}{i}$,
because the pre-release algorithm has already handled transferring
all such \rhoprime\ whose \CCP order depends on a lock other than \code{m}.

After the \emph{release} algorithm,
Figure~\ref{Fig:invariants}'s invariants hold---for the point in time \emph{after $e$}.

\subsection{Examples}
\label{Sec:analysis:examples}


Figure~\ref{fig:LessComplexExampleState} 
extends the example from Figure~\ref{fig:LessComplexExample}
with \raptor's analysis state after each event.
At \Acquire{m}{2},
\Raptor adds
\CCPThread{T2}{m}{1} to all \CCPLockSetPlain{\rho}{} (line~\ref{line:acq:add-CCPThread} in Algorithm~\ref{alg:acquire}) such that
\HBOrdered{\erho}{\Release{m}{1}} (line~\ref{line:acq:HB-condition} in Algorithm~\ref{alg:acquire}), 
since \CPOrdered{\erho}{\Acquire{m}{2}} if \CPOrdered{\Acquire{m}{1}}{\Release{m}{2}}.
At \Read{x}{1}{T3}, \raptor adds \CCPThreadPlain{\xiT{T3}}{m}{1}
to \CCPLockSet{x}{1} (line~\ref{line:read:establishCCP} in Algorithm~\ref{alg:read}), capturing that
\CPOrdered{\Write{x}{1}}{\Read{x}{1}{T3}}
if \CPOrdered{\Release{m}{1}}{\Acquire{m}{2}}.
However, it is not knowable at this point that
\CPOrdered{\Write{x}{1}}{\Read{x}{1}{T3}}.
At \Read{y}{1}{T2}, \Raptor establishes \RuleA by adding \CPThread{T2} to
\CPLockSet{u}{1} (line~\ref{line:read:RuleA:add} in Algorithm~\ref{alg:read}). 
Although it is possible to infer
\CPOrdered{\Write{x}{1}}{\Read{x}{1}{T3}}
at this point, \Raptor defers this logic
until \Release{m}{2}, when the pre-release algorithm (Algorithm~\ref{alg:pre-release})
detects \CPOrdered{\Acquire{m}{1}}{\Release{m}{2}} 
and handles \CCP elements of the form \CCPLockPlain{\rhoprime}{m}{1}.
At \Release{m}{2},
since $\CPThread{T2} \in \CPLockSet{m}{1}$ (line~\ref{line:prerel:trigger-condition} in Algorithm~\ref{alg:pre-release}) and
$\CCPThreadPlain{\xiT{T3}}{m}{1} \in \CCPLockSet{x}{1}$ (line~\ref{line:prerel:CCP-condition} in Algorithm~\ref{alg:pre-release}),
\raptor adds \xiT{T3} to \CPLockSet{x}{1} (line~\ref{line:prerel:trigger-condition} in Algorithm~\ref{alg:pre-release}),
indicating \CPOrdered{\Write{x}{1}}{\Read{x}{1}{T3}}.

Figure~\ref{fig:ruleb-transfer-needed} shows a more complex execution requiring ``transfer'' of \CCP ordering, in which
\CPOrdered{\Write{x}{1}}{\Write{x}{2}} because
\CPOrdered{\Acquire{m}{1}}{\Release{m}{2}}, which in turn depends on
\CPOrdered{\Acquire{o}{1}}{\Release{o}{2}}.
Even when \OwnerElement{m}{2}'s critical section ends at
\Release{m}{2}, it is not knowable that
\CPOrdered{\Acquire{m}{1}}{\Release{m}{2}}.
At \Release{m}{2}, the pre-release algorithm ``transfers'' \CCP ordering from \code{m} to \code{o}:
it adds $\PCPLock{p}{o}{1}$ to \PCPLockSet{x}{1} (line~\ref{line:prerel:add-dependent-transfers} in Algorithm~\ref{alg:pre-release}) because
$\CCPThread{T3}{o}{1} \in \PCPLockSet{m}{1}$ (line~\ref{line:prerel:CCP-transfer-condition} in Algorithm~\ref{alg:pre-release}) and
$\CCPLock{p}{m}{1} \in \PCPLockSet{x}{1}$ (line~\ref{line:prerel:CCP-condition} in Algorithm~\ref{alg:pre-release}).
As a result, at \Write{x}{2}, \raptor adds \CCPLockPlain{\xi}{o}{1} to \CCPLockSet{x}{1} (line~\ref{line:write:add-conditional-xi} in Algorithm~\ref{alg:write}).
Finally, at \Release{o}{2}, the analysis adds $\xi$ to \CPLockSet{x}{1} (line~\ref{line:prerel:trigger-condition} in Algorithm~\ref{alg:pre-release})
because $\CPThread{T2} \in \CPLockSet{o}{1}$ (line~\ref{line:prerel:trigger-condition} in Algorithm~\ref{alg:pre-release}) and
$\CCPLockPlain{\xi}{o}{1} \in \CCPLockSet{x}{1}$ (line~\ref{line:prerel:CCP-condition} in Algorithm~\ref{alg:pre-release}).


Figure~\ref{fig:ComplexExample} presents an even more complex execution involving ``transfer'' of \CCP ordering,
in which \CPOrdered{\Write{x}{1}}{\Write{x}{2}} because \CPOrdered{\Acquire{m}{1}}{\Release{m}{2}},
which in turn depends on \CPOrdered{\Release{q}{1}}{\Acquire{q}{2}}.

At event \Release{m}{2},
an online analysis \emph{cannot} determine that
\CPOrdered{\Acquire{m}{1}}{\Release{m}{2}}
because it is \emph{not knowable} that 
\CPOrdered{\Release{q}{1}}{\Acquire{q}{2}}.
At \Release{m}{2}, \raptor's pre-release algorithm
``transfers'' \CCP ordering from \code{m} to \code{q} by adding
\CCPThread{T5}{q}{1} to \CCPLockSet{x}{1} (line~\ref{line:prerel:add-dependent-transfers} in Algorithm~\ref{alg:pre-release}) because
$\CCPThread{T4}{q}{1} \in \CCPLockSet{m}{1}$ (line~\ref{line:prerel:CCP-transfer-condition} in Algorithm~\ref{alg:pre-release}) and
$\CCPThread{T5}{m}{1} \in \CCPLockSet{x}{1}$ (line~\ref{line:prerel:CCP-condition} in Algorithm~\ref{alg:pre-release}).
As a result, at \Write{x}{2}, 
the analysis thus adds \CCPThreadPlain{\xi}{q}{1} to \CCPLockSet{x}{1} (line~\ref{line:write:add-conditional-xi} in Algorithm~\ref{alg:write}).
Finally, at \Release{q}{2}, 
\raptor adds $\xi$ to \CPLockSet{x}{1} (line~\ref{line:prerel:trigger-condition} in Algorithm~\ref{alg:pre-release})
because $\CPThread{T3} \in \CPLockSet{q}{1}$ (line~\ref{line:prerel:trigger-condition} in Algorithm~\ref{alg:pre-release}) and
$\CCPThread{T5}{q}{1} \in \CCPLockSet{x}{1}$ (line~\ref{line:prerel:CCP-condition} in Algorithm~\ref{alg:pre-release}).

Alternatively, 
suppose that \thr{T1} executed its critical section on \code{o}
\emph{before} its critical section on \code{m}.
In that subtly different execution, \NotCPOrdered{\Write{x}{1}}{\Write{x}{2}}.
\Raptor tracks analysis state that achieves capturing
the difference between these two execution variants.

\input{examples/LessComplexExample_State}

\input{examples/RuleBTransferNeeded}

\input{examples/ComplexExample}

%% file: WriteEvent.tex
\newcommand{\algcaption}[1]{\hfill \Raptor's analysis for #1}

\begin{algorithm}[ht] \caption{\label{modWriteEvent}\algcaption{\Write{x}{i} by \code{T}}}
\begin{algorithmic}[1]
    \LineCommentx{Establish \RuleA}
    \ForAll {$\HBLock{m}{} \in \heldBy{T}$} \label{line:write:forall-locks}
    	\lIf{$\exists j \exists h \mid \LSLock{m}{j} \in \HBLockSet{x}{h} \land \CPThread{T} \notin \POLockSet{x}{h}$}{\Add{\CPThread{T}}{\CPLockSet{m}{j}} s.t.\ $j$ is max satisfying index} \label{line:write:LSlock-check} \label{line:write:CPEdge-established}
		\ForAll {threads \thr{t}} \label{line:write:forall-read-threads:RuleA}
			\lIf{$\exists j \exists h \mid \LSLock{m}{j} \in \HBLockSetRead{x}{h}{\code{t}} \land \CPThread{T} \notin \POLockSetRead{x}{h}{\code{t}}$}{\Add{\CPThread{T}}{\CPLockSet{m}{j}} {\smaller s.t.\ $j$ is max satisfying index}} \label{line:write:LSlock-check:read} \label{line:write:CPEdge-established:read}
		\EndFor
    \EndFor \label{line:write:endfor-locks}
    
	\LineCommentx{Add $\xi$ to represent \thr{T} at \Write{x}{i} from prior write \Write{x}{i-1}}
	\lIf{$\CPThread{T} \in \POLockSet{x}{i-1}$}{\Add{\xi}{\POLockSet{x}{i-1}}} \label{line:write:POEdge-Exists} \label{line:write:add-existingPO-xi}
	\lIf{$\HBThread{T} \in \HBLockSet{x}{i-1}$}{\Add{\xi}{\HBLockSet{x}{i-1}}} \label{line:write:HBEdge-Exists} \label{line:write:add-existingHB-xi}
	\lIf{$\CPThread{T} \in \CPLockSet{x}{i-1}$}{\Add{\xi}{\CPLockSet{x}{i-1}}} \label{line:write:CPEdge-Exists} \label{line:write:add-existingCP-xi}
	\lForAll{$\HBLock{m}{} \mid \PCPThread{T}{m}{} \in \PCPLockSet{x}{i-1}$}{\Add{\PCPThreadPlain{\xi}{m}{}}{\PCPLockSet{x}{i-1}}} \label{line:write:forall-dependent-locks} \label{line:write:add-conditional-xi} \label{line:write:end-xi}
	
	\LineCommentx{Add $\xi$ to represent \thr{T} at \Write{x}{i} from prior reads \Read{x}{i-1}{t} by all threads \thr{t}}
	\ForAll {threads \thr{t}} \label{line:write:forall-read-threads}
		\lIf{$\CPThread{T} \in \POLockSetRead{x}{i-1}{\code{t}}$}{\Add{\xi}{\POLockSetRead{x}{i-1}{\code{t}}}} \label{line:write:POEdge-Exists:read} \label{line:write:add-existingPO-xi:read}
		\lIf{$\HBThread{T} \in \HBLockSetRead{x}{i-1}{\code{t}}$}{\Add{\xi}{\HBLockSetRead{x}{i-1}{\code{t}}}} \label{line:write:HBEdge-Exists:read} \label{line:write:add-existingHB-xi:read}
		\lIf{$\CPThread{T} \in \CPLockSetRead{x}{i-1}{\code{t}}$}{\Add{\xi}{\CPLockSetRead{x}{i-1}{\code{t}}}} \label{line:write:CPEdge-Exists:read} \label{line:write:add-existingCP-xi:read}
		\lForAll{$\HBLock{m}{} \mid \PCPThread{T}{m}{} \in \PCPLockSetRead{x}{i-1}{\code{t}}$}{\Add{\PCPThreadPlain{\xi}{m}{}}{\PCPLockSetRead{x}{i-1}{\code{t}}}}  \label{line:write:forall-dependent-locksRead} \label{line:write:add-conditional-xiRead} \label{line:write:end-xiRead}
	\EndFor \label{line:write:forall-read-threads:end}
	\LineCommentx{Initialize \sets for \HBLock{x}{i}}
	\State $\POLockSet{x}{i}^+ \gets \{\code{T}\}$ \label{line:write:initialize-PO}
	\State $\HBLockSet{x}{i}^+ \gets \{\code{T}\} \cup \{\LSLock{m}{j} \mid \HBLock{m}{j} \in \heldBy{T}\}$ \label{line:write:initialize-HB}
\end{algorithmic}
\label{alg:write}
\end{algorithm}

%% file: ReadEvent.tex
\begin{algorithm}[ht] \caption{\algcaption{\Read{x}{i}{T} by \thr{T}}}
	\begin{algorithmic}[1]	
		\LineCommentx{Establish \RuleA}
		\ForAll {$\HBLock{m}{} \in \heldBy{T}$} \label{line:read:RuleA:start}
			\lIf{$\exists j \exists h \mid \LSLock{m}{j} \in \HBLockSet{x}{h} \land \CPThread{T} \notin \POLockSet{x}{h}$}{\Add{\CPThread{T}}{\CPLockSet{m}{j}} s.t.\ $j$ is max satisfying index} \label{line:read:RuleA:add}
		\EndFor \label{line:read:RuleA:end}
		\LineCommentx{Add \xiT{T} to represent \thr{T} at \Read{x}{i}{T} from prior write \Write{x}{i}}
		
		\lIf{$\CPThread{T} \in \POLockSet{x}{i}$}{\Add{\xi_{\code{T}}}{\POLockSet{x}{i}}} \label{line:read:establishPO}
		\lIf{$\HBThread{T} \in \HBLockSet{x}{i}$}{\Add{\xi_{\code{T}}}{\HBLockSet{x}{i}}} \label{line:read:establishHB}
		\lIf{$\CPThread{T} \in \CPLockSet{x}{i}$}{\Add{\xi_{\code{T}}}{\CPLockSet{x}{i}}} \label{line:read:establishCP}
		\lForAll{$\HBLock{m}{} \mid \PCPThread{T}{m}{} \in \PCPLockSet{x}{i}$}{\Add{\PCPThreadPlain{\xi_{\code{T}}}{m}{}}{\PCPLockSet{x}{i}}} \label{line:read:establishCCP}
		
		\LineCommentx{Initialize \sets for \OwnerElementRead{x}{i}{T} and reset
		\CPLockSetRead{x}{i}{\code{T}} and \CCPLockSetRead{x}{i}{\code{T}} to handle the case of a prior \Read{x}{i}{T}}
		\State $\POLockSetRead{x}{i}{\code{T}}^+ \gets \{\code{T}\}$ \label{line:read:initialize-PO} 
		\State $\HBLockSetRead{x}{i}{\code{T}}^+ \gets \{\code{T}\} \cup \{\LSLock{m}{j} \mid \HBLock{m}{j} \in \heldBy{T}\}$ \label{line:read:initialize-HB}
		\State $\CPLockSetRead{x}{i}{\code{T}}^+ \gets \CCPLockSetRead{x}{i}{T}^+ \gets \emptyset$ \label{line:read:resetRead}
	\end{algorithmic}
	\label{alg:read}
\end{algorithm}

%% file: AcquireEvent.tex
\begin{algorithm}[ht] \caption{\label{acquireEvent}\algcaption{\Acquire{m}{i} by \code{T}}}
	\begin{algorithmic} [1]
		\ForAll {$\rho$}\label{line:acq:active-only}
			\LineCommentxx{Establish order from \code{m} to \thr{T}}
			\lIf{$\CPLock{m} \in \CPLockSetPlain{\rho}{}$}{\Add{\CPThread{T}}{\CPLockSetPlain{\rho}{}}} \label{line:acq:CP-condition} \label{line:acq:add-CPThread} \label{line:acq:end-CP-condition}
			\lForAll{$\HBLock{n}{k} \mid \PCPLock{m}{n}{k} \in \PCPLockSetPlain{\rho}{}$}{\Add{\PCPThread{T}{n}{k}}{\PCPLockSetPlain{\rho}{}} {\smaller \Comment{No effect if $\exists k'<k \mid \PCPThread{T}{n}{k'} \in \PCPLockSetPlain{\rho}{}^+$}}}\label{line:acq:CCP-condition}\label{line:acq:add-transfer-CCPThread}\label{line:acq:end-CCP-condition}
			\If {$\exists j \mid \HBLock{m}{j} \in \HBLockSetPlain{\rho}{}$}\label{line:acq:HB-condition}
  			   	\State \Add{\HBThread{T}}{\HBLockSetPlain{\rho}{}}\label{line:acq:add-HBThread}
   			   	\LineCommentx{Add \CCP ordering}
 				\State \Add{\CCPThread{T}{m}{j}}{\PCPLockSetPlain{\rho}{}}\label{line:acq:add-CCPThread} {\smaller \Comment {No effect if $\exists j'<j \mid \PCPThread{T}{m}{j'} \in \PCPLockSetPlain{\rho}{}^+$}}
 			\EndIf\label{line:acq:end-HB-condition}
		\EndFor \label{acq:endallP}
		\LineCommentx{Initialize \sets for \OwnerElement{m}{i}}
		\State $\POLockSet{m}{i}^+ \leftarrow \{\HBThread{T}\}$\label{line:acq:initialize-PO}
		\State $\HBLockSet{m}{i}^+ \leftarrow \{\HBThread{T}\}$\label{line:acq:initialize-HB}
	\end{algorithmic} \label{alg:acquire}
\end{algorithm}

%% file: PreReleaseEvent.tex
\begin{algorithm}[ht] \caption{\label{pre-releaseEvent}\algcaption{pre-release of \Release{m}{i} by \code{T}}}
	\begin{algorithmic} [1]
		\ForAll {$\rho$}\label{line:prerel:active-only}
			\LineCommentxx{Handle \CCP ordering according to Rule (b)}
			\ForAll {$\rhoprime, j \mid \PCPLockPlain{\rhoprime}{m}{j} \in \PCPLockSetPlain{\rho}{}$}\label{line:prerel:CCP-condition}
				\lIf{$\exists l \ge j \mid \thr{T} \in \CPLockSet{m}{l}$}{\Add{$\rhoprime$}{\CPLockSetPlain{\rho}{}}} \label{line:prerel:trigger-condition} \label{line:prerel:add-CPSigma} \label{line:prerel:end-CCP-condition}
				\LineCommentx{Transfer \PCP to depend on other lock(s)}
				\ForAll{$\HBLock{n}{k} \mid (\exists l \ge j \mid \PCPThread{T}{n}{k} \in \PCPLockSet{m}{l})$} \label{line:prerel:CCP-transfer-condition}
					\State \Add{\PCPLockPlain{\rhoprime}{n}{k}}{\PCPLockSetPlain{\rho}{}} \label{line:prerel:add-dependent-transfers} \label{line:prerel:end-CCP-transfer-condition}
					{\smaller \Comment{No effect if $\exists k'<k \mid \PCPLockPlain{\rhoprime}{n}{k'} \in \PCPLockSetPlain{\rho}{}^+$}}
				\EndFor
			\EndFor
		\EndFor
	\end{algorithmic} \label{alg:pre-release}
\end{algorithm}

%% file: ReleaseEvent.tex
\begin{algorithm}[ht] \caption{\label{releaseEvent}\algcaption{\Release{m}{i} by \code{T}}}
	\begin{algorithmic} [1]
 		\ForAll {$\rho$}\label{line:rel:active-only}
 			\LineCommentxx {Establish order from \code{T} to \code{m}}
 			\lIf{$\CPThread{T} \in \CPLockSetPlain{\rho}{}$}{\Add{\CPLock{m}}{\CPLockSetPlain{\rho}{}}} \label{line:rel:CP-condition} \label{line:rel:add-CPlock} \label{line:rel:end-CP-condition}
			\lForAll{$\HBLock{n}{k} \mid \code{n} \ne \code{m} \land \PCPThread{T}{n}{k} \in \PCPLockSetPlain{\rho}{}$}{\Add{\PCPLock{m}{n}{k}}{\PCPLockSetPlain{\rho}{}}} \label{line:rel:CCP-condition} \label{line:rel:add-dependent-locks} \label{line:rel:end-CCP-condition}
			\lIf{$\HBThread{T} \in \HBLockSetPlain{\rho}{}$}{\Add{\HBLock{m}{i}}{\HBLockSetPlain{\rho}{}} {\smaller \Comment{No effect if $\exists i'<i \mid \HBLock{m}{i'} \in \HBLockSetPlain{\rho}{}^+$}}} \label{line:rel:HB-condition} \label{line:rel:add-HBlock} \label{line:rel:end-HB-condition}
			\LineCommentx {Remove \PCP elements conditional on \code{m}}
			\State $\PCPLockSetPlain{\rho}{}^+ \gets \PCPLockSetPlain{\rho}{}^+ \setminus \{ \, \PCPLockPlain{\rhoprime}{m}{j} \in \PCPLockSetPlain{\rho}{}\, \}$\label{line:rel:removal}
 		\EndFor
	\end{algorithmic} \label{alg:release}
\end{algorithm}

%% file: examples/LessComplexExample_State.tex
\begin{sidewaysfigure}
\newcommand{\sna}{---} 
\newcommand{\nochange}{\sna}
\newcommand{\sss}[1]{\scaleto{\bf #1}{4pt}}
\newcommand{\setE}[2]{\rule{0pt}{3ex}\ensuremath{\dfrac{\sss{#1}}{#2}}\rule{0pt}{4ex}\xspace}
\smaller
\centering
\begin{tabular}{@{}l@{\;}l@{\;}l@{\;}|l@{\;\;\;}l@{\;\;\;}l@{\;\;\;}l@{\;\;\;}l@{\;\;\;}l@{\;\;\;}l@{\;\;\;}l@{\;\;\;}l@{\;\;\;}l@{}}
\mc{3}{c|}{Execution} & \mc{10}{c}{Analysis state changes} \\
\thr{T1} & \thr{T2} & \thr{T3} & \OwnerElement{x}{1} & \OwnerElementRead{x}{1}{T3} & \OwnerElement{y}{1} & \OwnerElementRead{y}{1}{T2} & \OwnerElement{m}{1} & \OwnerElement{m}{2} & \OwnerElement{u}{1} & \OwnerElement{u}{2} & \OwnerElement{v}{1} & \OwnerElement{v}{2} \\[.3em]\hline
\Write{x}{1} 				&&& \setE{\HB}{\HBThread{T1}} & \sna & \sna & \sna & \sna & \sna & \sna & \sna & \sna & \sna \\
\Acquire{m}{1} 				&&& \nochange & \sna & \sna & \sna & \setE{\HB}{\HBThread{T1}} & \sna & \sna & \sna & \sna & \sna \\
\Release{m}{1} 				&&&	\setE{\HB}{\HBLock{m}{1}} & \sna & \sna & \sna & \setE{\HB}{\HBLock{m}{1}} & \sna & \sna & \sna & \sna & \sna \\
\Acquire{u}{1} 				&&&	\nochange & \sna & \sna & \sna & \nochange & \sna & \setE{\HB}{\HBThread{T1}} & \sna & \sna & \sna \\
\Write{y}{1} 				&&&	\nochange & \sna & \setE{\HB}{\HBThread{T1}{\;}\LSLock{u}{1}} & \sna & \nochange & \sna & \nochange & \sna & \sna & \sna \\
\Release{u}{1}\tikzmark{1}	&&&	\setE{\HB}{\HBLock{u}{1}} & \sna & \nochange & \sna & \setE{\HB}{\HBLock{u}{1}} & \sna & \setE{\HB}{\HBLock{u}{1}} & \sna & \sna & \sna \\
& \Acquire{m}{2} 			&&	\setE{\HB}{\HBThread{T2}} \setE{\CCP}{\CCPThread{T2}{m}{1}} & \sna & \nochange & \sna & \setE{\HB}{\HBThread{T2}} \setE{\CCP}{\CCPThread{T2}{m}{1}} & \setE{\HB}{\HBThread{T2}} & \nochange & \sna & \sna & \sna \\
& \Acquire{v}{1} 			&&	\nochange & \sna & \nochange & \sna & \nochange & \nochange & \nochange & \sna & \setE{\HB}{\HBThread{T2}} & \sna \\
& \Release{v}{1} 			&&	\setE{\HB}{\HBLock{v}{1}} \setE{\CCP}{\CCPLock{v}{m}{1}} & \sna & \nochange & \sna & \setE{\HB}{\HBLock{v}{1}} \setE{\CCP}{\CCPLock{v}{m}{1}} & \setE{\HB}{\HBLock{v}{1}} & \nochange & \sna & \setE{\HB}{\HBLock{v}{1}} & \sna \\
& & \Acquire{v}{2} 			&	\setE{\HB}{\HBThread{T3}} \setE{\CCP}{\CCPThread{T3}{m}{1}{\;}\CCPThread{T3}{v}{1}} & \sna & \nochange & \sna & \setE{\HB}{\HBThread{T3}} \setE{\CCP}{\CCPThread{T3}{m}{1}{\;}\CCPThread{T3}{v}{1}} & \setE{\HB}{\HBThread{T3}} \setE{\CCP}{\CCPThread{T3}{v}{1}} & \nochange & \sna & \setE{\HB}{\HBThread{T3}} \setE{\CCP}{\CCPThread{T3}{v}{1}} & \setE{\HB}{\HBThread{T3}} \\
& & \Release{v}{2} 			&	$\setminus \setE{\CCP}{\CCPThread{T3}{v}{1}}$ & \sna & \nochange & \sna & $\setminus \setE{\CCP}{\CCPThread{T3}{v}{1}}$ & \nochange & \nochange & \sna & $\setminus \setE{\CCP}{\CCPThread{T3}{v}{1}}$ & \setE{\HB}{\HBLock{v}{2}} \\
& & \Read{x}{1}{T3}			&	\setE{\HB}{\HBThreadPlain{\xiT{T3}}} \setE{\CCP}{\CCPThreadPlain{\xiT{T3}}{m}{1}} & \setE{\HB}{\HBThread{T3}} & \nochange & \sna & \nochange & \nochange & \nochange & \sna & \nochange & \nochange \\
& \tikzmark{2}\Acquire{u}{2}&&	\setE{\CCP}{\CCPThread{T2}{u}{1}} & \nochange & \setE{\HB}{\HBThread{T2}} \setE{\CCP}{\CCPThread{T2}{u}{1}} & \sna & \setE{\CCP}{\CCPThread{T2}{u}{1}} & \nochange & \setE{\HB}{\HBThread{T2}} \setE{\CCP}{\CCPThread{T2}{u}{1}} & \setE{\HB}{\HBThread{T2}} & \nochange & \nochange \\
& \Read{y}{1}{T2}			&&	\nochange & \nochange & \setE{\HB}{\HBThreadPlain{\xiT{T2}}} \setE{\CCP}{\CCPThreadPlain{\xiT{T2}}{u}{1}} & \setE{\HB}{\HBThread{T2}{\;}\LSLock{u}{2}{\;}\LSLock{m}{2}} & \nochange & \nochange & \setE{\CP}{\CPThread{T2}} & \nochange & \nochange & \nochange \\
& \Release{u}{2} 			&&	\setE{\CP}{\CPThread{T2}{\;}\CPLock{u}} & \nochange & \setE{\CP}{\CPThread{T2}{\;}\CPLock{u}{\;}\CPThreadPlain{\xiT{T2}}} $\setminus \setE{\CCP}{\CCPThreadPlain{\{\HBThreadPlain{\xiT{T2}}{\;}\HBThread{T2}\}}{u}{1}}$ & \nochange & \setE{\CP}{\CPThread{T2}{\;}\CPLock{u}} & \setE{\HB}{\HBLock{u}{2}} & \setE{\CP}{\CPThread{T2}{\;}\CPLock{u}} $\setminus \setE{\CCP}{\CCPThread{T2}{u}{1}}$ & \setE{\HB}{\HBLock{u}{2}} & \setE{\HB}{\HBLock{u}{2}} & \nochange \\
& \Release{m}{2} 			&&	\setE{\CP}{\CPThread{T3}{\;}\CPLock{m}{\;}\CPLock{v}{\;}\CPThreadPlain{\xiT{T3}}} $\setminus \setE{\CCP}{\CCPThreadPlain{\{\HBThreadPlain{\xiT{T3}}{\;}\HBThread{T2}{\;}\HBThread{T3}{\;}\HBLock{v}{}\}}{m}{1}}$ & \nochange & \setE{\HB}{\HBLock{m}{2}} \setE{\CP}{\CPLock{m}} & \nochange & \setE{\CP}{\CPThread{T3}{\;}\CPLock{m}{\;}\CPLock{v}} $\setminus \setE{\CCP}{\CCPThreadPlain{\{\HBThread{T2}{\;}\HBThread{T3}{\;}\HBLock{v}{}\}}{m}{1}}$ & \setE{\HB}{\HBLock{m}{2}} & \setE{\HB}{\HBLock{m}{2}} \setE{\CP}{\CPLock{m}} & \setE{\HB}{\HBLock{m}{2}} & \setE{\HB}{\HBLock{m}{2}} & \nochange
\end{tabular}
\textlink{1}{2}{\CP}


\caption{The execution from Figure~\ref{fig:LessComplexExample}, in which \CPOrdered{\Write{x}{1}}{\Read{x}{1}{T3}},
with \raptor's analysis state updates shown in the same format as Figure~\ref{fig:SimpleExampleState}.
For brevity, this figure and the article's remaining figures omit showing updates to \POLockSetPlain{\rho}{}.}
\label{fig:LessComplexExampleState}
\end{sidewaysfigure}

%% file: examples/RuleBTransferNeeded.tex
\begin{sidewaysfigure}
\newcommand{\sna}{---} 
\newcommand{\nochange}{\sna}
\newcommand{\sss}[1]{\scaleto{\bf #1}{4pt}}
\newcommand{\setE}[2]{\rule{0pt}{3ex}\ensuremath{\dfrac{\sss{#1}}{#2}}\rule{0pt}{3.15ex}\xspace}
\small
\smaller
\centering
\begin{tabular}{@{}l@{\;}l@{}l@{}l@{}|l@{}l@{}l@{}l@{}l@{}l@{}l@{}H@{}H@{}l@{}l@{}l@{}l@{}l@{}@{}}
\mc{4}{c|}{Execution} & \mc{12}{c}{Analysis state changes} \\
\thr{T1} & \thr{T2} & \thr{T3} & \thr{T4} & \OwnerElement{x}{1} & \OwnerElement{y}{1} & \OwnerElement{m}{1} & \OwnerElement{o}{1} & \OwnerElement{q}{1} & \OwnerElement{p}{1} & \OwnerElement{r}{1} & \OwnerElement{x}{2} & \OwnerElement{y}{2} & \OwnerElement{m}{2} & \OwnerElement{o}{2} & \OwnerElement{q}{2} & \OwnerElement{p}{2} & \OwnerElement{r}{2} \\\hline
\Write{x}{1}				&&&& \setE{\HB}{\HBThread{T1}} & \sna & \sna & \sna & \sna & \sna & \sna & \sna & \sna & \sna & \sna & \sna & \sna & \sna \\
\Acquire{m}{1}				&&&& \nochange & \sna & \setE{\HB}{\HBThread{T1}} & \sna & \sna & \sna & \sna & \sna & \sna & \sna & \sna & \sna & \sna & \sna \\
\Release{m}{1}				&&&& \setE{\HB}{\HBLock{m}{1}} & \sna & \setE{\HB}{\HBLock{m}{1}} & \sna & \sna & \sna & \sna & \sna & \sna & \sna & \sna & \sna & \sna & \sna \\
\Acquire{o}{1}				&&&& \nochange & \sna & \nochange & \setE{\HB}{\HBThread{T1}} & \sna & \sna & \sna & \sna & \sna & \sna & \sna & \sna & \sna & \sna \\
\Release{o}{1}				&&&& \setE{\HB}{\HBLock{o}{1}} & \sna & \setE{\HB}{\HBLock{o}{1}} & \setE{\HB}{\HBLock{o}{1}} & \sna & \sna & \sna & \sna & \sna & \sna & \sna & \sna & \sna & \sna \\
\Acquire{q}{1}\tikzmark{1}	&&&& \nochange & \sna & \nochange & \nochange & \setE{\HB}{\HBThread{T1}} & \sna & \sna & \sna & \sna & \sna & \sna & \sna & \sna & \sna \\
\Write{y}{1}				&&&& \nochange & \setE{\HB}{\HBThread{T1}{\;}\LSLock{q}{1}} & \nochange & \nochange & \nochange & \sna & \sna & \sna & \sna & \sna & \sna & \sna & \sna & \sna \\
\Release{q}{1}				&&&& \setE{\HB}{\HBLock{q}{1}} & \nochange & \setE{\HB}{\HBLock{q}{1}} & \setE{\HB}{\HBLock{q}{1}} & \setE{\HB}{\HBLock{q}{1}} & \sna & \sna & \sna & \sna & \sna & \sna & \sna & \sna & \sna \\
& \Acquire{o}{2}			&&& \setE{\HB}{\HBThread{T2}} \setE{\CCP}{\CCPThread{T2}{o}{1}} & \nochange & \setE{\HB}{\HBThread{T2}} \setE{\CCP}{\CCPThread{T2}{o}{1}} & \setE{\HB}{\HBThread{T2}} \setE{\CCP}{\CCPThread{T2}{o}{1}} & \nochange & \sna & \sna & \sna & \sna & \sna & \setE{\HB}{\HBThread{T2}} & \sna & \sna & \sna \\
& & \Acquire{m}{2}			&& \setE{\HB}{\HBThread{T3}} \setE{\CCP}{\CCPThread{T3}{m}{1}} & \nochange & \setE{\HB}{\HBThread{T3}} \setE{\CCP}{\CCPThread{T3}{m}{1}} & \nochange & \nochange & \sna & \sna & \sna & \sna & \setE{\HB}{\HBThread{T3}} & \nochange & \sna & \sna & \sna \\
& & \Acquire{p}{1}			&& \nochange & \nochange & \nochange & \nochange & \nochange & \setE{\HB}{\HBThread{T3}} & \sna & \sna & \sna & \nochange & \nochange & \sna & \sna & \sna \\
& & \Release{p}{1}			&& \setE{\HB}{\HBLock{p}{1}} \setE{\CCP}{\CCPLock{p}{m}{1}} & \nochange & \setE{\HB}{\HBLock{p}{1}} \setE{\CCP}{\CCPLock{p}{m}{1}} & \nochange & \nochange & \setE{\HB}{\HBLock{p}{1}} & \sna & \sna & \sna & \setE{\HB}{\HBLock{p}{1}} & \nochange & \sna & \sna & \sna \\
& \Acquire{r}{1}			&&& \nochange & \nochange & \nochange & \nochange & \nochange & \nochange & \setE{\HB}{\HBThread{T2}} & \sna & \sna & \nochange & \nochange & \sna & \sna & \sna \\
& \Release{r}{1}			&&& \setE{\HB}{\HBLock{r}{1}} \setE{\CCP}{\CCPLock{r}{o}{1}} & \nochange & \setE{\HB}{\HBLock{r}{1}} \setE{\CCP}{\CCPLock{r}{o}{1}} & \setE{\HB}{\HBLock{r}{1}} \setE{\CCP}{\CCPLock{r}{o}{1}} & \nochange & \nochange & \setE{\HB}{\HBLock{r}{1}} & \sna & \sna & \nochange & \setE{\HB}{\HBLock{r}{1}} & \sna & \sna & \sna \\
& & \Acquire{r}{2}			&& \setE{\CCP}{\CCPThread{T3}{o}{1}{\;}\CCPThread{T3}{r}{1}} & \nochange & \setE{\CCP}{\CCPThread{T3}{o}{1}{\;}\CCPThread{T3}{r}{1}} & \setE{\HB}{\HBThread{T3}} \setE{\CCP}{\CCPThread{T3}{o}{1}{\;}\CCPThread{T3}{r}{1}} & \nochange & \nochange & \setE{\HB}{\HBThread{T3}} \setE{\CCP}{\CCPThread{T3}{r}{1}} & \sna & \sna & \nochange & \setE{\HB}{\HBThread{T3}} \setE{\CCP}{\CCPThread{T3}{r}{1}} & \sna & \sna & \setE{\HB}{\HBThread{T3}} \\
& & \Release{r}{2}			&& \setE{\CCP}{\CCPLock{r}{m}{1}} $\setminus \setE{\CCP}{\CCPThread{T3}{r}{1}}$ & \nochange & \setE{\CCP}{\CCPLock{r}{m}{1}} $\setminus \setE{\CCP}{\CCPThread{T3}{r}{1}}$ & $\setminus \setE{\CCP}{\CCPThread{T3}{r}{1}}$ & \nochange & \setE{\HB}{\HBLock{r}{2}} & $\setminus \setE{\CCP}{\CCPThread{T3}{r}{1}}$ & \sna & \sna & \setE{\HB}{\HBLock{r}{2}} & $\setminus \setE{\CCP}{\CCPThread{T3}{r}{1}}$ & \sna & \sna & \setE{\HB}{\HBLock{r}{2}}\\
& & \Release{m}{2}			&& \setE{\CCP}{\CCPLock{\{\CPLock{m}{}{\;}\CPLock{p}{}\}}{o}{1}} $\setminus \setE{\CCP}{\CCPThreadPlain{\{\HBThread{T3}{\;}\CPLock{p}{\;}\CPLock{r}\}}{m}{1}}$ & \nochange & \setE{\CCP}{\CCPLock{\{\CPLock{m}{\;}\CPLock{p}\}}{o}{1}} $\setminus \setE{\CCP}{\CCPThreadPlain{\{\HBThread{T3}{\;}\CPLock{p}{\;}\CPLock{r}\}}{m}{1}}$ & \setE{\HB}{\HBLock{m}{2}} \setE{\CCP}{\CCPLock{m}{o}{1}} & \nochange & \setE{\HB}{\HBLock{m}{2}} & \setE{\HB}{\HBLock{m}{2}} & \sna & \sna & \setE{\HB}{\HBLock{m}{2}} & \setE{\HB}{\HBLock{m}{2}} & \sna & \sna & \setE{\HB}{\HBLock{m}{2}} \\
& & & \Acquire{p}{2}		& \setE{\HB}{\HBThread{T4}} \setE{\CCP}{\CCPThread{T4}{o}{1}{\;}\CCPThread{T4}{p}{1}} & \nochange & \setE{\HB}{\HBThread{T4}} \setE{\CCP}{\CCPThread{T4}{o}{1}{\;}\CCPThread{T4}{p}{1}} & \nochange & \nochange & \setE{\HB}{\HBThread{T4}} \setE{\CCP}{\CCPThread{T4}{p}{1}} & \nochange & \sna & \sna & \setE{\HB}{\HBThread{T4}} \setE{\CCP}{\CCPThread{T4}{p}{1}} & \nochange & \sna & \setE{\HB}{\HBThread{T4}} & \nochange \\
& & & \Release{p}{2}		& $\setminus \setE{\CCP}{\CCPThread{T4}{p}{1}}$ & \nochange & $\setminus \setE{\CCP}{\CCPThread{T4}{p}{1}}$ & \nochange & \nochange & $\setminus \setE{\CCP}{\CCPThread{T4}{p}{1}}$ & \nochange & \sna & \sna & $\setminus \setE{\CCP}{\CCPThread{T4}{p}{1}}$ & \nochange & \sna & \setE{\HB}{\HBLock{p}{2}} & \nochange \\
& & & \Write{x}{2}			& \setE{\HB}{\HBThreadPlain{\xi}} \setE{\CCP}{\CCPThreadPlain{\xi}{o}{1}} & \nochange & \nochange & \nochange & \nochange & \nochange & \nochange & \setE{\HB}{\HBThread{T4}} & \sna & \nochange & \nochange & \sna & \nochange & \nochange \\
& \tikzmark{2}\Acquire{q}{2}&&& \setE{\CCP}{\CCPThread{T2}{q}{1}} & \setE{\HB}{\HBThread{T2}} \setE{\CCP}{\CCPThread{T2}{q}{1}} & \setE{\CCP}{\CCPThread{T2}{q}{1}} & \setE{\CCP}{\CCPThread{T2}{q}{1}} & \setE{\HB}{\HBThread{T2}} \setE{\CCP}{\CCPThread{T2}{q}{1}} & \nochange & \nochange & \nochange & \sna & \nochange & \nochange & \setE{\HB}{\HBThread{T2}} & \nochange & \nochange \\
& \Write{y}{2}				&&& \nochange & \setE{\HB}{\HBThreadPlain{\xi}} \setE{\CCP}{\CCPThreadPlain{\xi}{q}{1}} & \nochange & \nochange & \setE{\CP}{\CPThread{T2}} & \nochange & \nochange & \nochange & \setE{\HB}{\HBThread{T2}{\;}\LSLock{o}{2}{\;}\LSLock{q}{2}} & \nochange & \nochange & \nochange & \nochange & \nochange \\
& \Release{q}{2}			&&& \setE{\CP}{\CPThread{T2}{\;}\CPLock{q}} $\setminus \setE{\CCP}{\CCPThread{T2}{q}{1}}$ & \setE{\CP}{\CPThreadPlain{\xi}{\;}\CPThread{T2}{\;}\CPLock{q}} $\setminus \setE{\CCP}{\CCPThreadPlain{\{\HBThreadPlain{\xi}{\;}\HBThread{T2}\}}{q}{1}}$ & \setE{\CP}{\CPThread{T2}{\;}\CPLock{q}} $\setminus \setE{\CCP}{\CCPThread{T2}{q}{1}}$ & \setE{\CP}{\CPThread{T2}{\;}\CPLock{q}} $\setminus \setE{\CCP}{\CCPThread{T2}{q}{1}}$ & \setE{\CP}{\CPThread{T2}{\;}\CPLock{q}} $\setminus \setE{\CCP}{\CCPThread{T2}{q}{1}}$ & \nochange & \setE{\HB}{\HBLock{q}{2}} & \nochange & \nochange & \nochange & \setE{\HB}{\HBLock{q}{2}} & \setE{\HB}{\HBLock{q}{2}} & \nochange & \nochange \\
& \Release{o}{2}			&&& \setE{\CP}{\CPThreadPlain{\xi}{\;}\CPThread{T3}{\;}\CPThread{T4}{\;}\CPLock{r}{\;}\CPLock{m}{\;}\CPLock{p}{\;}\CPLock{o}} & \setE{\HB}{\HBLock{o}{2}} \setE{\CP}{\CPLock{r}{\;}\CPLock{m}{\;}\CPLock{p}{\;}\CPLock{o}} & \setE{\CP}{\CPThread{T3}{\;}\CPThread{T4}} & \setE{\CP}{\CPThread{T3}{\;}\CPLock{r}{\;}\CPLock{m}{\;}\CPLock{o}} & \setE{\HB}{\HBLock{o}{2}} \setE{\CP}{\CPLock{o}} & \nochange & \setE{\HB}{\HBLock{o}{2}} & \nochange & \nochange & \nochange & \setE{\HB}{\HBLock{o}{2}} & \setE{\HB}{\HBLock{o}{2}} & \nochange & \nochange \\
							&&&& $\setminus \setE{\CCP}{\CCPThreadPlain{\{\HBThreadPlain{\xi}{\;}\HBThread{T4}{\;}\HBThread{T3}{\;}\HBThread{T2}{\;}\HBLock{r}{}{\;}\HBLock{m}{}{\;}\HBLock{p}{}\}}{o}{1}}$ & & $\setminus \setE{\CCP}{\CCPThreadPlain{\{\HBThread{T4}{\;}\HBThread{T3}{\;}\HBThread{T2}{\;}\HBLock{r}{}{\;}\HBLock{m}{}{\;}\HBLock{p}{}\}}{o}{1}}$
\end{tabular}
\undertextlink{1}{2}{\CP}
\caption{An execution illustrating \CCP transfer in which \CPOrdered{\Write{x}{1}}{\Write{x}{2}},
in the same format as Figure~\ref{fig:SimpleExampleState}.
For space, the figure omits \set updates for \OwnerElement{x}{2} and \OwnerElement{y}{2}.}
\label{fig:ruleb-transfer-needed}
\end{sidewaysfigure}

%% file: examples/ComplexExample.tex
\begin{sidewaysfigure}

\newcommand{\sna}{---} 
\newcommand{\nochange}{\sna}
\newcommand{\sss}[1]{\scaleto{\bf #1}{4pt}}
\newcommand{\setE}[2]{\rule{0pt}{3ex}\ensuremath{\dfrac{\sss{#1}}{#2}}\rule{0pt}{3.15ex}\xspace}
\smaller
\smaller
\centering
\begin{tabular}{@{}l@{\;}l@{\;}l@{\;}l@{\;}l@{\;}|l@{\;\;}l@{\;\;}l@{\;\;}l@{\;\;}l@{\;\;}l@{\;\;}l@{\;\;}H@{}H@{}l@{\;\;}l@{\;\;}l@{\;\;}l@{\;\;}l@{}}
\mc{5}{c|}{Execution} & \mc{14}{c}{Analysis state changes} \\
\thr{T1} & \thr{T2} & \thr{T3} & \thr{T4} & \thr{T5} & \OwnerElement{x}{1} & \OwnerElement{y}{1} & \OwnerElement{m}{1} & \OwnerElement{o}{1} & \OwnerElement{q}{1} & \OwnerElement{p}{1} & \OwnerElement{r}{1} & \OwnerElement{x}{2} & \OwnerElement{y}{2} & \OwnerElement{m}{2} & \OwnerElement{o}{2} & \OwnerElement{q}{2} & \OwnerElement{p}{2} & \OwnerElement{r}{2} \\\hline
& \Acquire{q}{1} 				&&&& \sna & \sna & \sna & \sna & \setE{\HB}{\HBThread{T2}} & \sna & \sna & \sna & \sna & \sna & \sna & \sna & \sna & \sna \\
\Write{x}{1} 					&&&&& \setE{\HB}{\HBThread{T1}} & \sna & \sna & \sna & \sna & \sna & \sna & \sna & \sna & \sna & \sna & \sna & \sna & \sna \\
\Acquire{m}{1} 					&&&&& \sna & \sna & \setE{\HB}{\HBThread{T1}} & \sna & \sna & \sna & \sna & \sna & \sna & \sna & \sna & \sna & \sna & \sna \\
\Release{m}{1} 					&&&&& \setE{\HB}{\HBLock{m}{1}} & \sna & \setE{\HB}{\HBLock{m}{1}} & \sna & \sna & \sna & \sna & \sna & \sna & \sna & \sna & \sna & \sna & \sna \\
\Acquire{o}{1} 					&&&&& \sna & \sna & \sna & \setE{\HB}{\HBThread{T1}} & \sna & \sna & \sna & \sna & \sna & \sna & \sna & \sna & \sna & \sna \\
\Release{o}{1} 					&&&&& \setE{\HB}{\HBLock{o}{1}} & \sna & \setE{\HB}{\HBLock{o}{1}} & \setE{\HB}{\HBLock{o}{1}} & \sna & \sna & \sna & \sna & \sna & \sna & \sna & \sna & \sna & \sna \\
& \Acquire{o}{2} 				&&&& \setE{\HB}{\HBThread{T2}} \setE{\CCP}{\CCPThread{T2}{o}{1}} & \sna & \setE{\HB}{\HBThread{T2}} \setE{\CCP}{\CCPThread{T2}{o}{1}} & \setE{\HB}{\HBThread{T2}} \setE{\CCP}{\CCPThread{T2}{o}{1}} & \sna & \sna & \sna & \sna & \sna & \sna & \setE{\HB}{\HBThread{T2}} & \sna & \sna & \sna \\
& \Release{o}{2} 				&&&& $\setminus \setE{\CCP}{\CCPThread{T2}{o}{1}}$ & \sna & $\setminus \setE{\CCP}{\CCPThread{T2}{o}{1}}$ & $\setminus \setE{\CCP}{\CCPThread{T2}{o}{1}}$ & \setE{\HB}{\HBLock{o}{2}} & \sna & \sna & \sna & \sna & \sna & \setE{\HB}{\HBLock{o}{2}} & \sna & \sna & \sna \\
& \Write{y}{1} 					&&&& \sna & \setE{\HB}{\HBThread{T2}{\;}\LSLock{q}{1}} & \sna & \sna & \sna & \sna & \sna & \sna & \sna & \sna & \sna & \sna & \sna & \sna \\
& \Release{q}{1}\tikzmark{1} 	&&&& \setE{\HB}{\HBLock{q}{1}} & \sna & \setE{\HB}{\HBLock{q}{1}} & \setE{\HB}{\HBLock{q}{1}} & \setE{\HB}{\HBLock{q}{1}} & \sna & \sna & \sna & \sna & \sna & \setE{\HB}{\HBLock{q}{1}} & \sna & \sna & \sna \\
& & \tikzmark{2}\Acquire{q}{2} 	&&& \setE{\HB}{\HBThread{T3}} \setE{\CCP}{\CCPThread{T3}{q}{1}} & \setE{\HB}{\HBThread{T3}} \setE{\CCP}{\CCPThread{T3}{q}{1}} & \setE{\HB}{\HBThread{T3}} \setE{\CCP}{\CCPThread{T3}{q}{1}} & \setE{\HB}{\HBThread{T3}} \setE{\CCP}{\CCPThread{T3}{q}{1}} & \setE{\HB}{\HBThread{T3}} \setE{\CCP}{\CCPThread{T3}{q}{1}}  & \sna & \sna & \sna & \sna & \sna & \setE{\HB}{\HBThread{T3}} \setE{\CCP}{\CCPThread{T3}{q}{1}} & \setE{\HB}{\HBThread{T3}} & \sna \\ 
& & & \Acquire{m}{2} 			&& \setE{\HB}{\HBThread{T4}} \setE{\CCP}{\CCPThread{T4}{m}{1}} & \sna & \setE{\HB}{\HBThread{T4}} \setE{\CCP}{\CCPThread{T4}{m}{1}} & \sna & \sna & \sna & \sna & \sna & \sna & \setE{\HB}{\HBThread{T4}} & \sna & \sna & \sna & \sna \\
& & & \Acquire{p}{1} 			&& \sna & \sna & \sna & \sna & \sna & \setE{\HB}{\HBThread{T4}} & \sna & \sna & \sna & \sna & \sna & \sna & \sna & \sna \\
& & & \Release{p}{1} 			&& \setE{\HB}{\HBLock{p}{1}} \setE{\CCP}{\CCPThread{p}{m}{1}} & \sna & \setE{\HB}{\HBLock{p}{1}} \setE{\CCP}{\CCPThread{p}{m}{1}} & \sna & \sna & \setE{\HB}{\HBLock{p}{1}} & \sna & \sna & \sna & \setE{\HB}{\HBLock{p}{1}} & \sna & \sna & \sna & \sna \\ 
& & \Acquire{r}{1} 				&&& \sna & \sna & \sna & \sna & \sna & \sna & \setE{\HB}{\HBThread{T3}} & \sna & \sna & \sna & \sna & \sna & \sna & \sna \\
& & \Release{r}{1} 				&&& \setE{\HB}{\HBLock{r}{1}} \setE{\CCP}{\CCPLock{r}{q}{1}} & \setE{\HB}{\HBLock{r}{1}} \setE{\CCP}{\CCPLock{r}{q}{1}} & \setE{\HB}{\HBLock{r}{1}} \setE{\CCP}{\CCPLock{r}{q}{1}} & \setE{\HB}{\HBLock{r}{1}} \setE{\CCP}{\CCPLock{r}{q}{1}} & \setE{\HB}{\HBLock{r}{1}} \setE{\CCP}{\CCPLock{r}{q}{1}} & \sna & \setE{\HB}{\HBLock{r}{1}} & \sna & \sna & \sna & \setE{\HB}{\HBLock{r}{1}} \setE{\CCP}{\CCPLock{r}{q}{1}} & \setE{\HB}{\HBLock{r}{1}} & \sna & \sna \\
& & & & \Acquire{p}{2} 			& \setE{\HB}{\HBThread{T5}} \setE{\CCP}{\CCPThread{T5}{m}{1}{\;}\CCPThread{T5}{p}{1}} & \sna & \setE{\HB}{\HBThread{T5}} \setE{\CCP}{\CCPThread{T5}{m}{1}{\;}\CCPThread{T5}{p}{1}} & \sna & \sna & \setE{\HB}{\HBThread{T5}} \setE{\CCP}{\CCPThread{T5}{p}{1}} & \sna & \sna & \sna & \setE{\HB}{\HBThread{T5}} \setE{\CCP}{\CCPThread{T5}{p}{1}} & \sna & \sna & \setE{\HB}{\HBThread{T5}} & \sna \\
& & & & \Release{p}{2} 			& $\setminus \setE{\CCP}{\CCPThread{T5}{p}{1}}$ & \sna & $\setminus \setE{\CCP}{\CCPThread{T5}{p}{1}}$ & \sna & \sna & $\setminus \setE{\CCP}{\CCPThread{T5}{p}{1}}$  & \sna & \sna & \sna & $\setminus \setE{\CCP}{\CCPThread{T5}{p}{1}}$ & \sna & \sna & \setE{\HB}{\HBLock{p}{2}} & \sna \\
& & & \Acquire{r}{2} 			&& \setE{\CCP}{\CCPThread{T4}{q}{1}{\;}\CCPThread{T4}{r}{1}} & \setE{\HB}{\HBThread{T4}} 								 & \setE{\CCP}{\CCPThread{T4}{q}{1}{\;}\CCPThread{T4}{r}{1}} & \setE{\HB}{\HBThread{T4}}								 & \setE{\HB}{\HBThread{T4}} 								 & \sna & \setE{\HB}{\HBThread{T4}} 		& \sna & \sna & \sna & \setE{\HB}{\HBThread{T4}} 								 & \setE{\HB}{\HBThread{T4}} 		 & \sna & \setE{\HB}{\HBThread{T4}} \\
								&&&&& 														 & \setE{\CCP}{\CCPThread{T4}{q}{1}{\:}\CCPThread{T4}{r}{1}} & 															 & \setE{\CCP}{\CCPThread{T4}{q}{1}{\;}\CCPThread{T4}{r}{1}} & \setE{\CCP}{\CCPThread{T4}{q}{1}{\;}\CCPThread{T4}{r}{1}} & 		& \setE{\CCP}{\CCPThread{T4}{r}{1}} &	   &	  &		 & \setE{\CCP}{\CCPThread{T4}{q}{1}{\;}\CCPThread{T4}{r}{1}} & \setE{\CCP}{\CCPThread{T4}{r}{1}} & 		& \\
& & & \Release{r}{2} 			&& $\setE{\CCP}{\CCPLock{r}{m}{1}} \setminus \setE{\CCP}{\CCPThread{T4}{r}{1}}$ & $\setminus \setE{\CCP}{\CCPThread{T4}{r}{1}}$ & $\setE{\CCP}{\CCPLock{r}{m}{1}} \setminus \setE{\CCP}{\CCPThread{T4}{r}{1}}$ & $\setminus \setE{\CCP}{\CCPThread{T4}{r}{1}}$ & $\setminus \setE{\CCP}{\CCPThread{T4}{r}{1}}$ & \setE{\HB}{\HBLock{r}{2}} & $\setminus \setE{\CCP}{\CCPThread{T4}{r}{1}}$ & \sna & \sna & \setE{\HB}{\HBLock{r}{2}} & $\setminus \setE{\CCP}{\CCPThread{T4}{r}{1}}$ & $\setminus \setE{\CCP}{\CCPThread{T4}{r}{1}}$ & \sna & \setE{\HB}{\HBLock{r}{2}} \\
& & & \Release{m}{2} 			&& \setE{\CCP}{\CCPThreadPlain{\{\HBThread{T5}{\;}\HBLock{p}{}{\;}\HBLock{m}{}\}}{q}{1}} & \setE{\HB}{\HBLock{m}{2}} \setE{\CCP}{\CCPLock{m}{q}{1}} & \setE{\CCP}{\CCPThreadPlain{\{\HBThread{T5}{\;}\HBLock{p}{}{\;}\HBLock{m}{}\}}{q}{1}} & \setE{\HB}{\HBLock{m}{2}} \setE{\CCP}{\CCPLock{m}{q}{1}} & \setE{\HB}{\HBLock{m}{2}} \setE{\CCP}{\CCPLock{m}{q}{1}} & \setE{\HB}{\HBLock{m}{2}} & \setE{\HB}{\HBLock{m}{2}} & \sna & \sna & \setE{\HB}{\HBLock{m}{2}} & \setE{\HB}{\HBLock{m}{2}} \setE{\CCP}{\CCPLock{m}{q}{1}} & \setE{\HB}{\HBLock{m}{2}} & \sna & \setE{\HB}{\HBLock{m}{2}} \\
								&&&&& $\setminus \setE{\CCP}{\CCPThreadPlain{\{\HBThread{T4}{\;}\HBThread{T5}{\;}\HBLock{r}{}{\;}\HBLock{p}{}\}}{m}{1}}$ & & $\setminus \setE{\CCP}{\CCPThreadPlain{\{\HBThread{T4}{\;}\HBThread{T5}{\;}\HBLock{r}{}{\;}\HBLock{p}{}\}}{m}{1}}$ \\
& & & & \Write{x}{2} 			& \setE{\HB}{\HBThreadPlain{\xi}} \setE{\CCP}{\CCPThreadPlain{\xi}{q}{1}} & \sna & \sna & \sna & \sna & \sna & \sna & \setE{\HB}{\HBThread{T5}} & \sna & \sna & \sna & \sna & \sna & \sna \\
& & \Write{y}{2} 				&&& \sna & \setE{\HB}{\HBThreadPlain{\xi}} \setE{\CCP}{\CCPThreadPlain{\xi}{q}{1}} & \sna & \sna & \setE{\CP}{\CPThread{T3}} & \sna & \sna & \sna & \setE{\HB}{\HBThread{T3}{\;}\LSLock{q}{2}} & \sna & \sna & \sna & \sna & \sna \\
& & \Release{q}{2} 				&&& \setE{\CP}{\CPThread{T3}{\;}\CPThread{T4}{\;}\CPThread{T5}{\;}\CPLock{p}{\;}\CPLock{r}{\;}\CPLock{m}{\;}\CPLock{q}{\;}\CPThreadPlain{\xi}} & \setE{\CP}{\CPThread{T3}{\;}\CPThread{T4}{\;}\CPLock{r}{\;}\CPLock{m}{\;}\CPLock{q}{\;}\CPThreadPlain{\xi}} & \setE{\CP}{\CPThread{T3}{\;}\CPThread{T4}{\;}\CPThread{T5}{\;}\CPLock{p}{\;}\CPLock{r}{\;}\CPLock{m}{\;}\CPLock{q}} & \setE{\CP}{\CPThread{T3}{\;}\CPThread{T4}{\;}\CPLock{r}{\;}\CPLock{m}{\;}\CPLock{q}} & \setE{\CP}{\CPThread{T4}{\;}\CPLock{r}{\;}\CPLock{m}{\;}\CPLock{q}} & \sna & \setE{\HB}{\HBLock{q}{2}} & \sna & \sna & \sna & \setE{\CP}{\CPThread{T3}{\;}\CPThread{T4}{\;}\CPLock{r}{\;}\CPLock{m}{\;}\CPLock{q}} & \setE{\HB}{\HBLock{q}{2}} & \sna & \sna \\
 								&&&&& $\setminus \setE{\CCP}{\CCPThreadPlain{\{\HBThread{T3}{\;}\HBThread{T4}{\;}\HBThread{T5}{\;}\HBThreadPlain{\xi}\}}{q}{1}}$& $\setminus \setE{\CCP}{\CCPThreadPlain{\{\HBThread{T3}{\;}\HBThread{T4}{\;}\HBThreadPlain{\xi}\}}{q}{1}}$& $\setminus \setE{\CCP}{\CCPThreadPlain{\{\HBThread{T3}{\;}\HBThread{T4}{\;}\HBThread{T5}\}}{q}{1}}$& $\setminus \setE{\CCP}{\CCPThreadPlain{\{\HBThread{T3}{\;}\HBThread{T4}\}}{q}{1}}$& $\setminus \setE{\CCP}{\CCPThreadPlain{\{\HBThread{T3}{\;}\HBThread{T4}\}}{q}{1}}$& & & & & & $\setminus \setE{\CCP}{\CCPThreadPlain{\{\HBThread{T3}{\;}\HBThread{T4}\}}{q}{1}}$ \\
 								&&&&& $\setminus \setE{\CCP}{\CCPThreadPlain{\{\HBLock{r}{}{\;}\HBLock{p}{}{\;}\HBLock{m}{}\}}{q}{1}}$							& $\setminus \setE{\CCP}{\CCPThreadPlain{\{\HBLock{r}{}{\;}\HBLock{m}{}\}}{q}{1}}$	& $\setminus \setE{\CCP}{\CCPThreadPlain{\{\HBLock{r}{}{\;}\HBLock{p}{}{\;}\HBLock{m}{}\}}{q}{1}}$							& $\setminus \setE{\CCP}{\CCPThreadPlain{\{\HBLock{r}{}{\;}\HBLock{m}{}\}}{q}{1}}$	& $\setminus \setE{\CCP}{\CCPThreadPlain{\{\HBLock{r}{}{\;}\HBLock{m}{}\}}{q}{1}}$	& & & & & & $\setminus \setE{\CCP}{\CCPThreadPlain{\{\HBLock{r}{}{\;}\HBLock{m}{}\}}{q}{1}}$
\end{tabular}
\textlink{1}{2}{\CP}
\caption{An execution in which \CPOrdered{\Write{x}{1}}{\Write{x}{2}} that shows more complex \CCP transfer than Figure~\ref{fig:ruleb-transfer-needed}.
For space, the figure omits \set updates for \OwnerElement{x}{2} and \OwnerElement{y}{2}.}
\label{fig:ComplexExample} 
\end{sidewaysfigure}

%% file: Removal.tex
\section{Removing Obsolete \Sets and Detecting \CP-RACES}
\label{sec:removal}
\label{sub:sec:termination}

\Raptor maintains \sets for every variable access and lock acquire.
Without removing \sets, the analysis state's size would be proportional to trace length, 
which---since the analysis iterates over 
all non-empty \set owners at acquire and release events---would
be unscalable in terms of both space and time.
Fortunately, for real (non-adversarial) program executions,
most \set owners become 
\emph{obsolete}---meaning that they will not be needed again---relatively quickly.
\Raptor detects obsolete \set owners
and removes each owner's \sets, saving both space and time.

\later{
\mike{I wonder what's the worst-case live size (\ie, \# of non-obsolete locksets)?
What's the worst-case running time for the algorithm?}
}

\subsection{Removing Obsolete \emph{Variable Access} \Sets and Detecting \CP-Races}

A variable access \set owner \OwnerElement{x}{i} or \OwnerElementRead{x}{i}{T}
becomes obsolete once the analysis determines
whether or not the corresponding access (\Write{x}{i} or \Read{x}{i}{T}, respectively) 
is involved in a \CP-race with the next access.
Detecting \CP-races is thus naturally part of checking for obsolete \set owners.

Algorithm~\ref{alg:varRemoval} shows the conditions for determining 
whether \OwnerElement{x}{i} or \OwnerElementRead{x}{i}{T} 
is obsolete or has a \CP-race with another access to variable \code{x}.
For write owner \OwnerElement{x}{i},
if \raptor has determined that 
\CPOrdered{\Write{x}{i}}{\Write{x}{i+1}} or \POOrdered{\Write{x}{i}}{\Write{x}{i+1}},
then according to the \invPO and \invCP invariants (Figure~\ref{Fig:invariants}),
$\xi \in \CPLockSet{x}{i} \cup \POLockSet{x}{i}$.
For \raptor to later determine that \CPOrdered{\Write{x}{i}}{\Write{x}{i+1}},
according to the \invCP and \invCCPconstraint invariants there must be some \OwnerElement{m}{j}
such that $\CCPThreadPlain{\xi}{m}{j} \in \CCPLockSet{x}{i}$.
Line~\ref{line:varRemoval:check-xi} checks these conditions and reports a race
if \raptor has determined \PO and \CP order between the conflicting write accesses
has not and will not be established.
Similarly, lines~\ref{line:varRemoval:write-read-race} and \ref{line:varRemoval:read-write-race}
check the conditions for reporting write--read and read--write \races, respectively.


After \Write{x}{i+1} has executed,
\OwnerElement{x}{i} or \OwnerElementRead{x}{i}{T} is obsolete
if \raptor definitely will not in the future determine that \Write{x}{i} or \Read{x}{i}{T}, respectively, is \CP ordered to a following conflicting access.
To determine whether future \CP ordering is possible,
according to the \invCP and \invCCPconstraint invariants
there must exist a \CCPThreadPlain{\xi}{m}{j} or \CCPThreadPlain{\xiT{T}}{m}{j} in \CCPLockSet{x}{i} or \CCPLockSetRead{x}{i}{T}.
Lines~\ref{line:varRemoval:remove-read-sets}--\ref{line:varRemoval:remove-write-sets} check these conditions,
and remove obsolete \sets for \OwnerElement{x}{i} and \OwnerElementRead{x}{i}{T}.

\input{VariableRemoval}


When the execution terminates (\ie, after the last event in the observed total order),
we assume that no thread holds any lock,\footnote{If an execution does not satisfy this condition,
\raptor can simulate the release of all held locks, by
performing the pre-release and release algorithms (Algorithms~\ref{alg:pre-release} and \ref{alg:release}) for each held lock.}
so the \CCP \sets for all owners are empty by the \PCPconstraintinv (Figure~\ref{Fig:invariants}).
Thus for every \conflicts{\Write{x}{i}}{\Write{x}{i+1}},
\conflicts{\Write{x}{i}}{\Read{x}{i}{t}},
and \conflicts{\Read{x}{i}{t}}{\Write{x}{i+1}} pair
for which \goldCP has not already ruled out a \race 
(\ie, $\xi \in \CPLockSet{x}{i} \cup \POLockSet{x}{i}$,
$\xiT{t} \in \CPLockSet{x}{i} \cup \POLockSet{x}{i}$, and
$\xi \in \CPLockSetRead{x}{i}{t} \cup \POLockSetRead{x}{i}{t}$, respectively),
\raptor eventually reports a \CP-race.

\paragraph{Correctness of detecting \CP-races.}

\newcommand{\xiStar}{\ensuremath{\xi_{*}}\xspace}

Now that we know how \raptor detects \CP-races,
we can prove that \raptor is a sound and complete \CP-race detector.

\begin{theorem}
An execution has a \CP-race if and only if \GoldCP reports a \CP-race for the execution.
\label{thm:sound-and-complete}
\end{theorem}


\begin{proof}

We prove the forward direction ($\Rightarrow$) and backward direction ($\Leftarrow$) in turn:

\paragraph{Forward direction (completeness).}

Suppose a trace \tr has a \CP-race, but \raptor does not report a \CP-race.
Without loss of generality, let $e$ and $e'$ be access (read/write) events such that 
$e \totalOrder e'$ and \conflicts{e}{e'}.
Let $\rho$ be the set owner for event $e$ (\eg, \OwnerElement{x}{i} for \Write{x}{i}), and
\xiStar be $\xi$ or \xiT{T} depending on whether $e'$ is a write or a read by \thr{T}, respectively.
Then $\NotCPOrdered{e}{e'} \land
\NotPOOrdered{e}{e'}$, but
$\xiStar \in \CPLockSetPlain{\rho}{} \cup \POLockSetPlain{\rho}{}$
at program termination.
According to Theorem~\ref{thm:invariants},
by the \invCP and \invPO invariants,
$\xiStar \notin \CPLockSetPlain{\rho}{} \land \xiStar \notin \POLockSetPlain{\rho}{}$
at program termination, which is a contradiction.

\paragraph{Backward direction (soundness).}

Suppose \raptor reports a \CP-race between access events $e$ and $e'$ such that $e \totalOrder e'$, but no \CP-race exists.
Let $\rho$ be the set owner for event $e$ (\eg, \OwnerElement{x}{i} for \Write{x}{i}), and
\xiStar be $\xi$ or \xiT{T} depending on whether $e'$ is a write or a read by \thr{T}, respectively.
Then $\xiStar \notin \CPLockSetPlain{\rho}{} \cup \POLockSetPlain{\rho}{}$
at program termination,
but $\CPOrdered{e}{e'} \lor \POOrdered{e}{e'}$.

Since $\xiStar \notin \POLockSetPlain{\rho}{}$ at program termination, 
according to Theorem~\ref{thm:invariants}, by the \POinv,
\NotPOOrdered{e}{e'}.
Thus $\CPOrdered{e}{e'}$.
By the \CPinv, at program termination,
\newcommand\et{\ensuremath{e^\Omega}}%
\[\CPThreadPlain{\xiStar} \in \CPLockSetPlain{\rho}{} \lor
\big(\exists \HBLock{n}{k} \mid \CCPLockPlain{\xiStar}{n}{k} \in \CCPLockSetPlain{\rho}{} \land
\exists j \mid \CPOrdered{\Release{n}{k}}{\Acquire{n}{j}} \totalOrder \et\big)\]
where $\et$ represents a final ``termination'' event.
Since $\xiStar \notin \CPLockSetPlain{\rho}{}$,
$\exists \HBLock{n}{k} \mid \CCPLockPlain{\xiStar}{n}{k} \in \CCPLockSetPlain{\rho}{}$.
By the \CCPconstraintinv, at $\et$, 
$\exists l \mid \Acquire{n}{l} \totalOrder \et \land \Release{n}{l} \not\totalOrder \et$.
However, an execution releases all held locks before terminating, 
so $\Release{n}{l} \totalOrder \et$, which is a contradiction.
\end{proof}

\subsection{Removing Obsolete \emph{Lock Acquire} \Sets}

\Raptor uses \sets for lock instances, 
such as \CPLockSet{m}{j},
for detecting \CP-ordered critical sections and
tracking \CP ordering from \Acquire{m}{j} to establish \RuleB.
Once \OwnerElement{m}{j}'s \sets' elements can no longer trigger \RuleB,
\OwnerElement{m}{j} is obsolete.

Algorithm~\ref{alg:lockRemoval} shows the condition for whether \OwnerElement{m}{j} is obsolete.
Lock owner \OwnerElement{m}{j} is \emph{not} obsolete if
any \set owner's \CCPLockSetPlain{\rho}{} \set contains \CCPLockPlain{\rhoprime}{m}{i} ($i \le j$) 
or might contain it at some later event
(indicated by \HBLock{m}{j} being in \HBLockSetPlain{\rho}{})---unless
$\CPLock{m} \in \CPLockSetPlain{\rho}{}$,
in which case $\PCPLockPlain{\rhoprime}{m}{i} \in \PCPLockSetPlain{\rho}{}$ would be superfluous.
Line~\ref{line:lockRemoval:check} shows the exact condition;
if it evaluates to true,
then the pre-release algorithm (Algorithm~\ref{alg:pre-release}) 
definitely will not use \OwnerElement{m}{j}'s \sets anymore, and so the removal algorithm removes each \set.
\notes{Note that the condition includes checking for \HBLock{m}{j} in \HBLockSetPlain{\rho}{}
because the acquire algorithm (Algorithm~\ref{alg:acquire}) adds \PCPThread{T}{m}{j}
to \PCPLockSetPlain{\rho}{} if $\HBLock{m}{j} \in \HBLockSetPlain{\rho}{}$.}%

\input{LockRemoval}

%% file: VariableRemoval.tex
\begin{algorithm}[ht] \caption{\hfill Detect obsolete owners and remove \sets and report \CP-races for \OwnerElement{x}{i} and \OwnerElementRead{x}{i}{t} (for all threads \thr{t})} 
\begin{algorithmic} [1]
	\If {$\Write{x}{i+1}$ has executed} \label{line:varRemoval:check-executed}
		\lIf {$\xi \notin \CPLockSet{x}{i} \cup \POLockSet{x}{i} \land \nexists \HBLock{m}{j} \mid \PCPThreadPlain{\xi}{m}{j} \in \PCPLockSet{x}{i}$}{Report \race between \Write{x}{i} and \Write{x}{i+1}}  \label{line:varRemoval:check-xi} \label{line:varRemoval:no-race-remove-locksets}
		\ForAll{threads \thr{t} such that \Read{x}{i}{t} executed} \label{line:varRemoval:forall-threads}
			\lIf {$\xiT{t} \notin \CPLockSet{x}{i} \cup \POLockSet{x}{i} \land \nexists \HBLock{m}{j} \mid \PCPThreadPlain{\xiT{t}}{m}{j} \in \PCPLockSet{x}{i}$}{Report \race between \Write{x}{i} and \Read{x}{i}{t}} \label{line:varRemoval:write-read-race}
			\lIf {$\xi \notin \CPLockSetRead{x}{i}{t} \cup \POLockSetRead{x}{i}{t} \land \nexists \HBLock{m}{j} \mid \PCPThreadPlain{\xi}{m}{j} \in \PCPLockSetRead{x}{i}{t}$}{ Report \race between \Read{x}{i}{t} and \Write{x}{i+1}}\label{line:varRemoval:read-write-race}
			\lIf {$\nexists \HBLock{m}{j} \mid \PCPThreadPlain{\xi}{m}{j} \in \PCPLockSetRead{x}{i}{t}$}
			{Remove $\POLockSetRead{x}{i}{t}^+$, $\HBLockSetRead{x}{i}{t}^+$, $\CPLockSetRead{x}{i}{t}^+$, $\PCPLockSetRead{x}{i}{t}^+$} \label{line:varRemoval:remove-read-sets}
		\EndFor
		\If {$\nexists \HBLock{m}{j} \mid \PCPThreadPlain{\xi}{m}{j} \in \PCPLockSet{x}{i} \land \forall \text{threads \thr{t} } \nexists \HBLock{m}{j} \mid \PCPThreadPlain{\xi_\thr{t}}{m}{j} \in \PCPLockSet{x}{i}$}  \label{line:varRemoval:check-ccp} \label{line:varRemoval:report-race}
			\State Remove $\POLockSet{x}{i}^+$, $\HBLockSet{x}{i}^+$, $\CPLockSet{x}{i}^+$, $\PCPLockSet{x}{i}^+$ \label{line:varRemoval:remove-write-sets}
			{\smaller \Comment{No \race (unless reported above)}}
		\EndIf
	\EndIf
\end{algorithmic} \label{alg:varRemoval}
\end{algorithm}

%% file: LockRemoval.tex
\begin{algorithm}[ht] \caption{\hfill Detect obsolete owner and remove \sets for \OwnerElement{m}{j}}
	\begin{algorithmic} [1]
		\If {\Release{m}{j} has executed}
			\If {$\nexists \rho \mid \rho \ne \HBLock{m}{j} \land \big(\HBLock{m}{j} \in \HBLockSetPlain{\rho}{} \lor \LSLock{m}{j} \in \HBLockSetPlain{\rho}{} \lor (\exists i \le j \mid \PCPLockPlain{\sigma}{m}{i} \in \PCPLockSetPlain{\rho}{})\big) \land \CPLock{m} \notin \CPLockSetPlain{\rho}{}$}\label{line:lockRemoval:check}
				\State Remove $\POLockSet{m}{j}^+$, $\HBLockSet{m}{j}^+$, $\CPLockSet{m}{j}^+$, $\PCPLockSet{m}{j}^+$
			\EndIf
		\EndIf
	\end{algorithmic}
	\label{alg:lockRemoval}
\end{algorithm}

%% file: Evaluation.tex
\section{Evaluation}
\label{sec:eval}

This section evaluates the performance 
and \CP-race coverage of an implementation of \raptor.

\input{Implementation}

\subsection{Methodology}

We evaluate \raptor on two sets of benchmarks:

\begin{itemize}

\item Benchmarks evaluated by Smaragdakis \etal~\cite{causally-precedes} 
that we were able to obtain and run.
Of their benchmarks that we do \emph{not} evaluate here, 
all except \bench{StaticBucketMap} execute fewer than 1,000 events~\cite{causally-precedes}.

\item The DaCapo benchmarks~\cite{dacapo-benchmarks-conf}, 
version 9.12 Bach, with small workload size.
We exclude programs that do not run out of the box with RoadRunner.

\end{itemize}

\noindent
The experiments execute on a quiet system with four Intel Xeon E5-4620 8-core processors (32 cores total)
and 128 GB of main memory, running Linux 2.6.32.

\paragraph{Datalog \CP implementation.}

We have extended our \raptor implementation 
to generate a trace of events in a format 
that can be processed by Smaragdakis \etal's Datalog \CP implementation~\cite{causally-precedes}, 
which they have shared with us.
The \raptor and Datalog \CP implementations thus analyze identical executions.
Our experiments run the Datalog \CP implementation 
with a bounded window size of 500, 5,000, and 50,000 events;
Smaragdakis \etal\ used a window size of 500 events~\cite{causally-precedes}.

\paragraph{FastTrack implementation.}

RoadRunner implements the state-of-the-art \FTFull \HB analysis~\cite{fasttrack}.
We added counting of events to the \FTFull implementation in order to report event counts,
which also makes \FTFull's performance more directly comparable with \raptor's (which also counts events).

\later{
\paragraph{\GoldCP \HB implementation.}

To compare \GoldCP's performance, we have implemented a variant of \goldCP that
we call \emph{\goldHB} that is based on \goldCP but that \emph{tracks only \HB
ordering}. This analysis is thus similar to
Goldilocks~\cite{goldilocks-pldi-2007}, but implemented in the context of our
\goldCP implementation (\ie, implemented on top of RoadRunner, using the same
data structures and underlying functionality as our \goldCP implementation).

The \HB algorithm is significantly simpler than the \CP algorithm and allows for a
greatly simplified implementation. Aggresive optimizations are possible, but for
the best comparison the \goldHB implementation only differs slightly from the
\goldCP implementation. \goldHB does not need to track \CP or \PCP elements as
they are specific to detecting the \CP relation. Additionally, \goldHB does not
track \sets for lock instances since the recursive nature of \RuleB and
\RuleC (of the \CP relation) is not part of the \HB algorithm. Lastly, locks do
not have multiple instances because the \HB algorithm does not need to be aware
of the total ordering between events. 
}

\later{
\subsection{Performance}

This section compares the performance of our \goldCP implementation,
the FastTrack happens-before analysis implementation~\cite{fasttrack},
and the Datalog \CP implementation~\cite{causally-precedes}.

\paragraph{Online analyses.}

Figure~?? compares the performance of online analyses for \HB and \CP data race detection.
For comparison purposes, the evaluation includes our implementation of \goldHB that is similar to the Goldilocks happens-before
race detector~\cite{goldilocks-pldi-2007}.
\ldots

\paragraph{Datalog \CP scalability.}
\label{Sec:eval:datalog-scalability}

Smaragdakis \etal's Datalog \CP implementation is limited to detecting \races within
a bounded \emph{window} of events~\cite{causally-precedes}.
By default, the Datalog \CP implementation uses a window of 500 events~\cite{causally-precedes}.
Figure~\ref{fig:datalog-windowTimes} shows the execution time of the Datalog \CP implementation
for window sizes from 100 to 50,000 events for the DaCapo \bench{pmd} benchmark.
Note that the y-axis is logarithmic. The figure shows that run time increases by
over two orders of magnitude when the window size increases from 500 to 50,000 events.
For a window size of 50,000 events, the run time is about 3 hours;
for 100,000 events, the run time exceeded a day, so we killed the experiment.
As Section~?? shows,
the Datalog \CP implementation needs a window size of 1,000,000--3,500,000 events
in order to detect the \races in \bench{pmd}.
\mike{The story might be even worse for other benchmarks.
Once you have more benchmarks running successfully through \goldCP,
you can do the same experiment for other benchmarks that report \races.
But this is not the highest priority experiment; for now, the expermient with \bench{pmd} is pretty convincing on its own. :)}

\begin{figure} [ht]
\centering
\includegraphics[width=0.5\textwidth]{figs/logYDatalog}
\caption{Datalog wall clock time on various window sizes for an execution of \bench{pmd} with 3,680,684 events.}
\label{fig:datalog-windowTimes}
\end{figure}
}

\subsection{Coverage and Performance}

This section focuses on answering two key empirical questions.
(1) For real program executions, how many \CP-races does \raptor detect
that \HB detectors cannot detect?
(2) For real program executions, how many \CP-races does \raptor detect that the Datalog \CP implementation
cannot detect due to its bounded analysis window?

Table~\ref{tab:time-race-results} reports execution time and \HB- and \CP-race coverage of \raptor,
compared with \FTFull's time and Datalog \CP's time and \CP-race coverage for various event window sizes.
The \emph{Events} columns report the number of events
(writes, reads, acquires, and releases)
processed by \FTFull
and \raptor (including additional events generated by translation; Section~\ref{sec:implementation}).
\FTFull reports significantly different event counts than \raptor because,
for some of the programs, \raptor is limited to analyzing only events within 2 hours.
\bench{FTPServer} is an exception;
despite running to completion, \raptor reports fewer events than \FTFull,
we believe, because the program executes significantly different numbers of events depending on the execution's timing. 

\notes{
\jake{FTPServer runs on a script that starts the server, 
then a client, then kills the client, then kills the server. 
I did a quick test by changing the sleep times 
for starting the client and killing the server and client
which changes the active threads and event total.
This is my guess for why \FTFull and \Raptor have different active threads and event counts,
but I could do more testing to get a better understanding.
\jake{I did more testing by modifying sleeps and found that as long as the client has time
to generate threads then regardless of the number of active threads the event counts 
do not drop to 63K events. 
I speculate that \raptor reports less races because FTPServer did not finish execution.
In 358s, 63K events were processed but FTPServer was unable to continue execution for some reason.
This would be in line with the other benchmarks that continued processing events until the 2.0h mark
and the event counts were taken from what was processed in 2.0h.}}
}%
Timing-based nondeterminism also leads to different threads counts for \FTFull and \raptor on \bench{FTPServer} and \bench{Jigsaw}.
\notes{
\mike{Also, why does \raptor sometimes have significantly fewer threads, \eg, \bench{avrora} and \bench{eclipse}---is
that because \raptor doesn't get through enough of the program such that it's spawned all of it's threads,
or are the workloads different?
\jake{I came to the same conclusion that fewer threads are spawned because the program
doesn't get through enough of the program.
I checked that the workloads are the same. 
I also tried executing with different RoadRunner settings, such as availableProcessors,
but couldn't get the same behavior that Raptor is reporting.}}
}%
\bench{avrora} and \bench{eclipse} spawn fewer threads in executions with \raptor than with \FTFull
because the \raptor executions do not make enough progress in 2 hours to have spawned all threads.

In our experiments,
the Datalog \CP implementation processes the same execution trace as \raptor,
by processing traces generated by \raptor during its online analysis.
\Raptor reports significantly different event counts for some of the programs 
also evaluated by Smaragdakis \etal~\cite{causally-precedes};
this discrepancy makes sense because, at least in several cases,
we are using different workloads than prior work.


\begin{table*}[t]
\newcommand{\events}[2]{#1K}
\newcommand{\rtime}[2]{#1~s}
\newcommand{\rcounts}[2]{\textbf{#2} (#1)}
\newcommand{\rzero}{\textbf{0}}
\newcommand{\rna}{\textbf{N/A}}
\newcommand\mynotes[1]{}
\small
\smaller
\centering
\begin{tabular}{@{}l@{\;}|@{\;}r@{\;}r@{\;}r|@{\;}r@{\;}r@{\;}r@{\;\;}|r@{\;\;}r@{\;\;}r|@{\;}r@{\;\;}r|@{\;}r@{\;\;}r|@{\;}r@{\;\;}r@{}}
 & \mc{3}{c|@{\;}}{\bf \FTFull~\cite{fasttrack}}  & \mc{6}{c|@{\;}}{\bf \Raptor} & \mc{6}{c@{}}{\bf Datalog \CP~\cite{causally-precedes}} \\
 & & & 							 & \mc{6}{c|@{\;}}{}		    & \mc{2}{c|@{\;}}{\bf w = 500} & \mc{2}{c|@{\;}}{\bf w = 5,000} & \mc{2}{c@{}}{\bf w = 50,000} \\
 & \#Thr & Events & Time & \#Thr & Events & Time & HB\&CP & HB & CP & Time & CP & Time & CP & Time & CP \\\hline
\bench{elevator} 	& 3  & \events{29}{29,497} 			& \rtime{25}{25.35}	& 3  & \events{30}{29,500}		& 28 s 	& \rzero 			& \rzero & \rzero & 38 s & \rzero & 587 s & \rzero & $>$2.0 h & \rna\\
\bench{FTPServer} 	& 12 & \events{743}{742,597} 		& \rtime{55}{55.01} & 13 & \events{63}{62,512}		& 358 s & \rcounts{766}{31} & \rcounts{371}{28} & \rcounts{395}{22} & 117 s 	& \rcounts{27}{6} 	& 1.4 h\mynotes{5,147 s}& \rcounts{440}{22}\mynotes{$\ge$15}& $>$2.0 h 				& \rna\\
\bench{hedc} 		& 7  & \events{7}{6,702} 			& \rtime{4.7}{4.66} & 7  & \events{8}{7,886} 		& 2.5 s & \rcounts{37}{5} 	& \rcounts{20}{4} 	& \rcounts{17}{1} 	& 31 s 		& \rzero 			& 96 s 					& \rcounts{1}{1} 					& 102 s 				& \rcounts{1}{1}\\
\bench{Jigsaw} 		& 70 & \events{1,774}{1,774,391} 	& \rtime{37}{37.19} & 69 & \events{361}{360,687}	& 2.0 h & \rcounts{34}{19} 	& \rcounts{15}{13} 	& \rcounts{19}{6}	& 294 s 	& \rcounts{1}{1}	& $>$2.0 h 				& \rna 								& $>$2.0 h 				& \rna\\
\bench{philo} 		& 6  & \events{$<1$}{607} 			& \rtime{4.3}{4.29}	& 6  & \events{$<$1}{703} 		& 3.1 s & \rzero 			& \rzero 			& \rzero 			& 15 s 		& \rzero 			& 15 s 					& \rzero 							& 15 s 					& \rzero\\
\bench{tsp} 		& 9  & \events{923,477}{923,476,983}& \rtime{11}{10.46} & 9  & \events{164}{164,167}	& 2.0 h & \rcounts{83}{2} 	& \rcounts{1}{1} 	& \rcounts{82}{2} 	& 462 s 	& \rcounts{2}{1} 	& $>$2.0 h 				& \rna 								& $>$2.0 h 				& \rna\\\hline
\bench{avrora} 		& 4  & \events{694,582}{694,582,487}& \rtime{22}{21.77} & 2  & \events{311}{310,830}	& 2.0 h & \rzero 			& \rzero 			& \rzero 			& 161 s 	& \rzero 			& 1.4 h\mynotes{4,989 s}& \rzero 							& $>$2.0 h 				& \rna\\
\bench{eclipse} 	& 18 & \events{13,131}{13,130,942}	& \rtime{29}{29.05}	& 4  & \events{321}{320,510}	& 2.0 h & \rcounts{2}{1} 	& \rcounts{2}{1} 	& \rzero 			& 392 s 	& \rzero 			& $>$2.0 h 				& \rna 								& $>$2.0 h 				& \rna\\
\bench{jython} 		& 2  & \events{79,042}{79,041,786}	& \rtime{37}{36.59} & 2  & \events{413}{412,616}	& 2.0 h & \rcounts{3}{1} 	& \rzero 			& \rcounts{3}{1} 	& 278 s 	& \rzero 			& $>$2.0 h 				& \rna 								& $>$2.0 h 				& \rna\\
\bench{pmd} 		& 2  & \events{3,681}{3,680,644} 	& \rtime{12}{11.56} & 2  & \events{1,805}{1,804,994}& 2.0 h & \rzero 			& \rzero 			& \rzero 			& 188 s 	& \rzero 			& 1.6 h\mynotes{5,771 s}& \rzero 							& $>$2.0 h 				& \rna\\
\bench{tomcat} 		& 1  & \events{81}{81,262}			& \rtime{3.0}{3.0} 	& 1  & \events{87}{86,936} 		& 65 s 	& \rzero 			& \rzero 			& \rzero 			& 42 s 		& \rzero 			& 296 s 				& \rzero 							& 1.6 h\mynotes{5,896 s}& \rzero\\
\bench{xalan} 		& 8  & \events{112,753}{112,753,201}& \rtime{16}{16.31}	& 8  & \events{2,639}{2,638,597}& 2.0 h	& \rcounts{47}{5} 	& \rcounts{17}{3} 	& \rcounts{30}{5} 	& 0.28 h	& \rcounts{11}{5} 	& 0.48 h\mynotes{1,721 s}& \rcounts{11}{5}					& $>$2.0 h 				& \rna 
\end{tabular}

\caption{Dynamic execution characteristics and detected \HB- and \races reported for the evaluated programs.
\emph{Events} is the number of program events processed by \FTFull and \raptor,
and \emph{\#Thr} is the execution's thread count.
\emph{Time} is wall clock execution time, reported in seconds or hours and rounded to two significant figures or the nearest integer if $\ge$100.
For each race column (\emph{HB\&CP}, \emph{HB}, \emph{CP}), the first number is distinct static races,
and the second number (in parentheses) is dynamic races.
Each run of \FTFull, \Raptor, and Datalog \CP is limited to 2 hours;
\emph{N/A} columns indicate that Datalog \CP did not finish in 2 hours and thus did not report any \races.}
\label{tab:time-race-results}
\end{table*}

\paragraph{Run-time performance.}
For \FTFull, the \emph{Time} column reports execution time for optimized \HB analysis.
For \raptor, the \emph{Time} column reports execution time for \raptor
to perform its analysis
(which dominates execution time)
while also generating a trace of events processed by the Datalog \CP implementation.
We execute \raptor on each program until it either terminates normally or has executed for two hours.
For Datalog \CP, the \emph{Time} column reports the cost of executing Datalog \CP 
on the trace generated by the same \raptor execution.

\paragraph{Race coverage.}
The \emph{HB\&CP} column reports all races detected by \raptor;
the \raptor implementation reports all \CP-races and identifies whether it is also an \HB-race.
The \emph{HB} column reports \HB-races (which are always also \CP-races),
and the \emph{CP} column reports \emph{\CP-only} races, which are
\CP-races that are not also \HB-races.
For each race column, the first number is \emph{distinct static races};
a static race is an unordered pair of static source locations 
(each source location is a source method and bytecode index).
The second number (in parentheses) is \emph{dynamic} races reported.
\notes{
\jake{The static counts allow for the same static race to be both a \HB-race and \CP-only race.
The static counts are calculated from the dynamic counts
but does not remove \CP-only static races that are also \HB static races.
So the counts are very literal.}
}%

Datalog \CP, which reports \CP-only races by filtering out \CP-races that are also \HB-races,
reports fewer \CP-only races
than \raptor because of its bounded event window sizes of 500, 5,000, and 50,000 events.
Datalog \CP does not report any \CP-only races if we terminate it after 2 hours (\emph{N/A}).
With a larger window size (5,000 or 50,000 events), 
Datalog \CP sometimes finds more races than at smaller window sizes,
but it often cannot complete within 2 hours.

By comparing dynamic event identifiers reported by Datalog \CP with those reported by \raptor,
we have verified that 
(1) \raptor detects all \CP-only races detected by Datalog \CP, and
(2) Datalog \CP detects each \CP-only race reported by \raptor that we would expect it to find
(\ie, \CP-only races separated by fewer than 500, 5,000, or 50,000 events).

Note that for \bench{FTPServer} with a 5,000-event window, 
Datalog \CP reports more dynamic \CP-only races than \raptor. 
This phenomenon occurs because Datalog \CP reports some races
that \raptor filters out by ordering accesses it detects as racing (Section~\ref{sec:implementation}),
\eg, if \Write{x}{i} \CP-races with \Write{x}{i+2}, 
\raptor will report the \Write{x}{i}--\Write{x}{i+1} and
\Write{x}{i+1}--\Write{x}{i+2} \CP-races, 
but Datalog \CP will also report the \Write{x}{i}--\Write{x}{i+2} race.
Since Datalog \CP reports only dynamic event numbers for \CP-races,
we have not been able to determine whether the extra 45 dynamic \CP-only races 
reported by Datalog \CP represent any additional \emph{distinct static} \CP-only races 
beyond the 22 distinct static \CP-only races reported by \raptor. 
As a further sanity check of \raptor's results,
we have cross-referenced \raptor's reported \emph{\HB}-races 
with results from running RoadRunner's FastTrack implementation~\cite{fasttrack} on the same workloads.

\paragraph{Race characteristics.}
Table~\ref{tab:CPRace-distances} shows the ``event distance'' for the \CP-only races detected by \raptor.
The event distance of a \CP-race is the distance in the observed total order of events between the race's two accesses.
For each static race in \raptor's \emph{\CP} column from Table~\ref{tab:time-race-results},
Table~\ref{tab:CPRace-distances} reports the range of event distances 
for all dynamic instances of that static race.
The table shows that many detected dynamic \CP-only races 
have event distances over 500 events and sometimes over 5,000 and 50,000 events;
\raptor finds these dynamic \CP-only races, but Datalog \CP often does not,
either due to window size or timeout.
For some static \CP-only races, 
every dynamic occurrence has an event distance exceeding 500, 5,000, or 50,000 events,
so Datalog \CP does not detect the static \CP-only race at all, 
corresponding to \CP-only races missed by Datalog \CP in
Table~\ref{tab:time-race-results}.
(A few dynamic \CP-only races have event distances of just less than 500 events, 
but Datalog \CP does not find them with a window size of 500 events 
because additional events outside of the window are required to determine \CP order.)

\begin{table}[t]
\small
\centering
\begin{tabular}{@{}l@{\;}|@{\;\;}l@{}}
 & Event distance range for each static \race\\\hline
\bench{FTPServer} & 8--1,570; 203--1,472; 253--1,426; 293--1,836; 349--852; 459--2,308; 463--2,309; 495--1,586; \\
&                   499--2,175; 525--2,415; 510--1,601; 607--1,959; 613--2,505; 760--1,931; 807--2,921; 968--968; \\
&                   1,023--1,929; 1,397--1,462; 1,453--1,783; 1,736--1,736; 2,063--2,379; 2,079--2,079\\
\bench{Jigsaw} & 186,137--324,861; 186,147--324,859; 186,161--324,860; 325,804--331,405; 322,805--359,773; \\
&				323,020--359,894\\
\bench{hedc} & 2,067--3,774\\
\bench{tsp} & 177--14,303; 939--939\\
\bench{jython} & 157,736--157,739\\
\bench{xalan} & 4--16; 15--32; 15--48; 18--52; 134--149\\
\end{tabular}
\caption{Event distances for detected \races. For each static \CP-only race reported by \raptor, 
the table reports a range of event distances for all dynamic occurrences of the static \CP-only race.}
\label{tab:CPRace-distances}
\end{table}

\medskip
\noindent
In summary, \raptor handles execution traces of 100,000s--1,000,000s of events within two hours,
and it finds \CP-races that the prior \CP analysis cannot find.

%% file: Implementation.tex
\subsection{Implementation}
\label{sec:implementation}

Our implementation of the \raptor analysis is built on
\emph{RoadRunner},\footnote{\url{https://github.com/stephenfreund/RoadRunner}, version 0.3}
a dynamic analysis framework for concurrent Java
programs~\cite{roadrunner}
that implements the high-performance FastTrack \hbFull race detector~\cite{fasttrack}.
RoadRunner instruments Java bytecode dynamically at class loading time,
generating events for memory accesses
(loads and stores to field and array elements) and
synchronization operations (lock acquire, release, wait, and resume;
thread fork and join; and \code{volatile} read and write).
Our implementation is open source and publicly available.\footnote{%
\url{https://github.com/PLaSSticity/Raptor.git}}

\paragraph{Handling non-lock synchronization.}

The \raptor analysis handles variable read and write events
and lock acquire and release events as depicted 
in Algorithms~\ref{alg:write}--\ref{alg:lockRemoval}.
Our implementation conservatively \CP-orders conflicting \code{volatile} accesses
(according to \RuleA of Definition~\ref{def:cp})
by translating each \code{volatile} access to 
a critical section surrounding the access on a lock unique to the variable at run time.
By similar translation, the implementation establishes \CP order from
\code{static} class initializer to \code{static} class accessed events (\cf~\cite{jvm-spec}).
The implementation establishes \CP order for thread fork (parent to child)
and join (child to parent) by generating critical sections containing
conflicting accesses unique to the threads involved.
Smaragdakis \etal\ translate these synchronization operations 
similarly before feeding them to their Datalog \CP implementation~\cite{causally-precedes}.

\paragraph{Removing obsolete \sets and reporting \races.}

The implementation follows the logic from 
Algorithms~\ref{alg:varRemoval} and \ref{alg:lockRemoval}
to remove obsolete \set owners and report \CP-races.
However, instead of executing these algorithms directly 
(\eg, periodically passing over all non-obsolete \sets 
and explicitly clearing obsolete \set owners),
the implementation performs \emph{reference counting} 
to identify \CP-races and remove obsolete \sets.
The implementation tracks the numbers of remaining 
\CCPLockPlain{\xi}{m}{j} and \CCPLockPlain{\xiT{T}}{m}{j} elements
in \CCPLockSet{x}{i} (or \CCPLockSetPlain{\OwnerElementRead{x}{i}{T}}{}).
If any of these counts drop to zero,
and the expected $\xi$ or \xiT{T} element(s) are not in 
\CPLockSet{x}{i} (or \CPLockSetPlain{\OwnerElementRead{x}{i}{T}}{}),
then the implementation reports a \CP-race.
After reporting the race, 
the implementation adds the corresponding $\xi$ or $\xiT{T}$ element 
to \CPLockSet{x}{i} (or \CPLockSetPlain{\OwnerElementRead{x}{i}{T}}{}),
effectively simulating a \CP-race-free execution up to the current event
and avoiding reporting false \CP-races downstream; 
Datalog \CP behaves similarly~\cite{causally-precedes}.
If all counts drop to zero, 
the implementation concludes that any races between \Write{x}{i} (or \Read{x}{i}{T})
and a following access have already been detected or cannot occur,
and thus its \set owner \OwnerElement{x}{i} (or \OwnerElementRead{x}{i}{T}) is obsolete. 

The implementation removes an obsolete \set owner
(Algorithms~\ref{alg:varRemoval} and \ref{alg:lockRemoval})
by removing all strong references to it,
allowing it to be garbage collected~\cite{goetz-weak-refs-2005}.
Variable owners are kept alive by the locks corresponding to
\CCPLockPlain{\xi}{m}{j} and \CCPLockPlain{\xiT{T}}{m}{j} elements 
in \CCPLockSet{x}{i} (or \CCPLockSetPlain{\OwnerElementRead{x}{i}{T}}{}) referencing the variable.
Lock owners are kept alive by variables containing 
\CCPLockPlain{\xi}{m}{j} and \CCPLockPlain{\xiT{T}}{m}{j} elements 
in \CCPLockSet{x}{i} (or \CCPLockSetPlain{\OwnerElementRead{x}{i}{T}}{}).
At a release event, 
the implementation removes \CCPLockPlain{\xi}{m}{j} and \CCPLockPlain{\xiT{T}}{m}{j} elements
from lock and variable owners along with the lock's references to variables. 

\paragraph{Optimization.}


Our prototype implementation of \raptor is largely unoptimized.
\notes{While we think there is significant opportunity for optimization to reduce
unnecessary or redundant work,
we also believe that designing and implementing effective optimizations would be a major undertaking (\eg, designing the ``right'' fast-path lookups
while avoiding huge memory overheads).}%
We have however implemented the following optimization out of necessity.
Before the pre-release algorithm (Algorithm~\ref{alg:pre-release}) 
iterates over all active \set owners $\rho$,
it pre-computes the following information:
(1) the maximum $j$ such that $\exists l \ge j \mid \thr{T} \in \CPLockSet{m}{l}$; and
(2) for each $j$, the set of \HBLock{n}{k} 
such that $\exists l \ge j \mid \PCPThread{T}{n}{k} \in \PCPLockSet{m}{l}$.
This pre-computation enables quick lookups at
lines~\ref{line:prerel:trigger-condition}--\ref{line:prerel:end-CCP-transfer-condition} 
of Algorithm~\ref{pre-releaseEvent}, 
significantly outperforming an unoptimized pre-release algorithm.

\later{
The implementation also performs ``fast path'' checks that identify and skip analysis of ``redundant'' program accesses;
an access is redundant if the same or stronger access (writes are stronger than reads) to the same variable
occurred on the same thread without intervening synchronization.
The implementation identifies redundant accesses by checking if the \sets for the prior access
(\OwnerElement{x}{i-1} for \Write{x}{i}; \OwnerElementRead{x}{i-1}{T} for \Read{x}{i}{T})
contain only elements added at lines~\ref{line:write:initialize-PO}--\ref{line:write:initialize-HB} of Algorithm~\ref{alg:write}
or lines~\ref{line:read:initialize-PO}--\ref{line:read:initialize-HB} of Algorithm~\ref{alg:read}.
\jake{I was mistaken. The original implementation used to gather the results in the paper
does not have logic for ``fast paths''.}
}

%% file: RelatedWorks.tex
\section{Related Work}
\label{sec:related-works}

Sections~\ref{sec:introduction} and \ref{sec:definitions} covered
\hbFull (\HB) analysis~\cite{happens-before,fasttrack}, prior work's \cpFull (\CP) analysis~\cite{causally-precedes},
and online predictive analyses~\cite{wcp,vindicator}.
This section covers other analyses and approaches.

\paragraph{Offline predictive analysis using SMT constraints.}


In contrast with partial-order-based predictive analyses~\cite{causally-precedes,wcp,vindicator},
\emph{SMT-based} approaches encode constraints that explore alternative schedule interleavings
as satisfiability modulo theories (SMT) constraints that can be solved by employing
an SMT solver~\cite{rvpredict-pldi-2014,maximal-causal-models,rdit-oopsla-2016,said-nfm-2011,jpredictor,ipa}.
SMT-based approaches generate constraints that are superlinear in execution size,
and they cannot scale beyond bounded windows of execution 
(\eg, 500--10,000 events).

\paragraph{Dynamic non-predictive analysis.}
\label{Sec:related:goldilocks}
\label{Sec:related:lockset-analysis}

\HbFull analysis~\cite{fasttrack,multirace,goldilocks-pldi-2007} and
other sound non-predictive analyses~\cite{frost,valor-oopsla-2015,ifrit} 
cannot predict data races in \emph{other} executions.

As Section~\ref{Sec:overview:locksets} mentioned, 
the \emph{HB} components of \raptor are similar
to prior work's \emph{Goldilocks} analysis~\cite{goldilocks-pldi-2007}.
\PCP sets, which are unique to \raptor, incur significant algorithmic complexity compared with Goldilocks.

Goldilocks and \raptor use per-variable ``locksets'' or \sets, respectively, 
to track the \HB relation (and in \raptor's case, the \CP relation) soundly and completely~\cite{goldilocks-pldi-2007}.
In contrast, \emph{lockset analysis} is a different kind of analysis 
that detects violations of a locking discipline that requires
conflicting memory accesses to hold a common lock~\cite{dinning-schonberg, ocallahan-hybrid-racedet-2003, eraser, object-racedet-2001,
choi-racedet-2002, racedet-escape-nishiyama-2004}.
Although lockset analysis can predict data races in other executions,
it is inherently unsound (reports false positives)
since not all violations of a locking discipline are data races 
(\eg, accesses ordered by fork, join, or notify--wait synchronization).
Prior work hybridizes happens-before and lockset analyses for 
performance or accuracy reasons~\cite{ocallahan-hybrid-racedet-2003,racetrack,multirace},
but it inherently cannot soundly predict data races in other executions.

\paragraph{Exposing data races.}

Some prior work exposes data races by exploring multiple thread interleaving
schedules~\cite{randomized-scheduler,chess,racageddon,racefuzzer,drfinder,mcr}.
Similarly, other analyses perturb interleavings by pausing threads
to expose data races~\cite{datacollider,racemob,racechaser-caper-cc-2017}.
In contrast, \raptor detects data races that are possible 
in \emph{other} executions from a \emph{single} observed execution (\ie, from a single schedule).
\Raptor and these approaches are potentially complementary
since \raptor can detect more data races than non-predictive analysis for a given execution.

\paragraph{Static analysis.}

Static data race detection analysis can achieve completeness, 
detecting all of a program's
data races~\cite{naik-static-racedet-2006, naik-static-racedet-2007, locksmith, racerx, relay-2007}.
However,
it is inherently unsound (reports false races).
Developers generally avoid unsound analyses
because each reported race---whether true or false---takes substantial time to
investigate~\cite{pacer-2010, literace, benign-races-2007, fasttrack, microsoft-exploratory-survey,billion-lines-later}.

\paragraph{Prioritizing data races.}

Prior work exposes erroneous \emph{behaviors} due to data races,
in order to prioritize races that are demonstrably harmful~\cite{adversarial-memory,
benign-races-2007, relaxer, datacollider, portend-asplos12, racefuzzer, prescient-memory}.
However, \emph{all} data races are problematic because they lead to ill-defined
semantics, as researchers have argued
convincingly~\cite{memory-models-cacm-2010,c++-memory-model-2008,java-memory-model,data-races-are-pure-evil,jmm-broken,out-of-thin-air-MSPC14}.
In any case, prioritizing data races is complementary to our work, 
which tries to detect as many (true) data races as possible.

%% file: Conclusion.tex
\section{Conclusion}
\label{sec:conclusion}

\Raptor is the first analysis that computes the \cpFull (\CP) relation online.
We have proved that \CP maintains invariants needed to track \CP soundly and completely.
Our evaluation
shows that \raptor can analyze execution traces of 100,000s--1,000,000s of events within two hours,
allowing it to find \CP-races that the prior \CP analysis cannot find.

%% file: acks.tex
\begin{acks}

We thank Steve Freund for help with using and modifying RoadRunner;
Yannis Smaragdakis for help with Datalog CP and experimental infrastructure;
and Yannis Smaragdakis, Jaeheon Yi, Swarnendu Biswas, and Man Cao for helpful discussions, suggestions, and other feedback.

This material is based upon work supported by the National Science Foundation
under Grants CSR-1218695, CAREER-1253703, CCF-1421612, and XPS-1629126.


\end{acks}





%% file: Correctness.tex
\section{Correctness: \Raptor Tracks \CP Soundly and Completely}
\label{sec:correctness}


%
%
%
%
%
%

\input{Proof-Details}

%% file: Proof-Details.tex
%

This section proves that \raptor soundly and completely tracks \CP,
by proving Theorem~\ref{thm:invariants}: that \raptor maintains Figure~\ref{Fig:invariants}'s invariants after every event.
For brevity, our proof \emph{only shows that \goldCP maintains the \CPinv}.
For the other invariants, it is comparatively straightforward to see that they hold:

\begin{itemize}

\item \Raptor maintains \HB \sets in much the same way as Goldilocks~\cite{goldilocks-pldi-2007},
which provably computes HB soundly and completely~\cite{goldilocks-tr-with-proofs}.
Thus it is fairly straightforward to see that \raptor maintains the \invPO, \invHB, \invHBindex, and \invHBcriticalsection invariants.

\item The analysis maintains the \CPruleAinv by updating \CP \sets 
as soon as a conflicting write or read executes
(lines~\ref{line:write:forall-locks}--\ref{line:write:endfor-locks} of Algorithm~\ref{alg:write} and
lines~\ref{line:read:RuleA:start}--\ref{line:read:RuleA:end} of Algorithm~\ref{alg:read}).

\item The \PCPconstraintinv holds because \PCP elements of the form \PCPLockPlain{\rhoprime}{m}{j}
only exist during critical sections on \code{m}, by the following argument.
At \Release{m}{i}, Algorithm~\ref{alg:release}
removes all elements \PCPLockPlain{\rhoprime}{m}{j} from every $\PCPLockSetPlain{\rho}{}^+$.
When $\nexists \rhoprime' \mid \PCPLockPlain{\rhoprime'}{m}{j} \in \PCPLockSetPlain{\rho}{}$, 
only Algorithm~\ref{alg:acquire} (for \Acquire{m}{i}) adds elements of the form \PCPLockPlain{\rhoprime}{m}{j} to
$\PCPLockSetPlain{\rho}{}^+$.

\end{itemize}

\noindent
We first prove two lemmas that we will use to prove Theorem~\ref{thm:invariants}:

\begin{lemma}
Let $e$ be any \emph{release} event, \ie, $e = \Release{m}{i}$ by thread \thr{T}.
If Figure~\ref{Fig:invariants}'s invariants hold
before \goldCP's pre-release algorithm (Algorithm~\ref{alg:pre-release}) executes,
then after the pre-release algorithm executes,
the invariants still hold, for the moment in time \emph{before} $e$ executes.

\label{lem:pre-release}
\end{lemma}

\newcommand\rhsSet[1]{\ensuremath{\mathit{RHS}^{#1}}}
\newcommand\lhsSet[1]{\ensuremath{\mathit{LHS}^{#1}}}

\begin{empty}

\setlist[description]{leftmargin=\parindent,labelindent=\parindent,listparindent=\parindent}
\setlist[itemize]{leftmargin=\parindent,labelindent=\parindent}

\begin{proof}

Let $e = \Release{m}{i}$ by thread \thr{T}.
Let $\rho$ be any \set owner.
Let \erho be the event corresponding to $\rho$, \ie,
$\erho = \Write{x}{h}$ if $\rho = \HBLock{x}{h}$, or
$\erho = \Acquire{m}{h}$ if $\rho = \HBLock{m}{h}$.

We define the following abbreviations for the left- and right-hand sides of the \CPinv:

Let $\lhsSet{} = \CPLockSetPlain{\rho}{} \cup \myset{\rhoprime \mid \big( \exists \HBLock{n}{k} \mid \PCPLockPlain{\rhoprime}{n}{k} \in \PCPLockSetPlain{\rho}{} \land \exists l \mid \CPOrdered{\Release{n}{k}}{\Acquire{n}{l}} \totalOrder e\big)}$.

Let $\lhsSet{+} = \CPLockSetPlain{\rho}{}^+ \cup \myset{\rhoprime \mid \big( \exists \HBLock{n}{k} \mid \PCPLockPlain{\rhoprime}{n}{k} \in \PCPLockSetPlain{\rho}{}^+ \land \exists l \mid \CPOrdered{\Release{n}{k}}{\Acquire{n}{l}} \totalOrder e\big)}$.

Let $\rhsSet{} = \myset{\rhoprime \mid \big(\exists e' \mid \comp{\rhoprime}{e'} \land \CPOrdered{\erho}{e'} \totalOrder e\big)}$.



Suppose Figure~\ref{Fig:invariants}'s invariants hold (\ie, $\lhsSet{} = \rhsSet{}$) before the pre-release algorithm executes.

To show $\lhsSet{+} = \rhsSet{}$, we show
$\lhsSet{+} \subseteq \rhsSet{}$
(subset)
and
$\lhsSet{+} \supseteq \rhsSet{}$
(superset)
in turn.


\paragraph{\underline{Subset direction}:}

Let $\rhoprime \in \lhsSet{+}$.

Either $\rhoprime \in \CPLockSetPlain{\rho}{}^+$ or
$\exists \HBLock{n}{k} \mid \PCPLockPlain{\rhoprime}{n}{k} \in \PCPLockSetPlain{\rho}{}^+ \land \exists l \mid \CPOrdered{\Release{n}{k}}{\Acquire{n}{l}} \totalOrder e$
(or both):

  \begin{description}

  \item[Case 1:] $\rhoprime \in \CPLockSetPlain{\rho}{}^+$ 

  If $\rhoprime \in \CPLockSetPlain{\rho}{}$, then $\rhoprime \in \lhsSet{}$, so $\rhoprime \in \rhsSet{}$ because $\lhsSet{} = \rhsSet{}$.

  Otherwise ($\rhoprime \notin \CPLockSetPlain{\rho}{}$), Algorithm~\ref{alg:pre-release} adds $\rhoprime$ to $\CPLockSetPlain{\rho}{}^+$, which can happen only at line~\ref{line:prerel:add-CPSigma}.
  Let $j$ be such that $\PCPLockPlain{\rhoprime}{m}{j} \in \PCPLockSetPlain{\rho}{}$ (line~\ref{line:prerel:CCP-condition})
  and $\CPThread{T} \in \CPLockSet{m}{q} \mid q \ge j$ (line~\ref{line:prerel:trigger-condition}).
  Since $\CPThread{T} \in \CPLockSet{m}{q}$ and $\lhsSet{} = \rhsSet{}$,
  therefore $\thr{T} \in \rhsSet{}$.
  Thus $\exists e' \mid \getThread{e'} = \thr{T} \land \CPOrdered{\Acquire{m}{q}}{e'} \totalOrder e$.
  Since $e' \totalOrder e$ and $\getThread{e'} = \getThread{e}$, \HBOrdered{e'}{e} by the definition of \HB.
  Thus \CPOrdered{\Acquire{m}{q}}{e} by \CP \RuleC and therefore \CPOrdered{\Release{m}{q}}{\Acquire{m}{i}} by \CP \RuleB (recall that $e = \Release{m}{i}$).
  If $q = j$, then \CPOrdered{\Release{m}{j}}{\Acquire{m}{i}}.
  Otherwise, $q > j$ so \HBOrdered{\Acquire{m}{j}}{\Acquire{m}{q}} by the definition of \HB.
  Thus \CPOrdered{\Acquire{m}{j}}{e} and therefore \CPOrdered{\Release{m}{j}}{\Acquire{m}{i}}.
  Since $\PCPLockPlain{\rhoprime}{m}{j} \in \PCPLockSetPlain{\rho}{}$ and \CPOrdered{\Release{m}{j}}{\Acquire{m}{i}},
  $\rhoprime \in \lhsSet{}$ and thus $\rhoprime \in \rhsSet{}$.

  \item[Case 2:] Let $\HBLock{n}{k}$ and $\HBLock{n}{l}$ be such that $\PCPLockPlain{\rhoprime}{n}{k} \in \PCPLockSetPlain{\rho}{}^+ \land
  \CPOrdered{\Release{n}{k}}{\Acquire{n}{l}} \totalOrder e$.
  
  If $\PCPLockPlain{\rhoprime}{n}{k} \in \PCPLockSetPlain{\rho}{}$,
  then $\rhoprime \in \lhsSet{}$, so $\rhoprime \in \rhsSet{}$ because $\lhsSet{} = \rhsSet{}$.

  Otherwise, $\neg \big(\PCPLockPlain{\rhoprime}{n}{k} \in \PCPLockSetPlain{\rho}{} \land
  \CPOrdered{\Release{n}{k}}{\Acquire{n}{l}} \totalOrder e\big)$.
  We know $\CPOrdered{\Release{n}{k}}{\Acquire{n}{l}} \totalOrder e$,
  so therefore $\PCPLockPlain{\rhoprime}{n}{k} \notin \PCPLockSetPlain{\rho}{}$.
  Thus Algorithm~\ref{alg:pre-release} adds \PCPLockPlain{\rhoprime}{n}{k} to $\PCPLockSetPlain{\rho}{}^+$ only at line~\ref{line:prerel:add-dependent-transfers}.
  Let $j$ be such that $\PCPLockPlain{\rhoprime}{m}{j} \in \PCPLockSetPlain{\rho}{}$ (line~\ref{line:prerel:CCP-condition}) and
  $\PCPThread{T}{n}{k} \in \PCPLockSet{m}{q} \mid q \ge j$ (line~\ref{line:prerel:CCP-transfer-condition}).
  Since $\PCPThread{T}{n}{k} \in \PCPLockSet{m}{q}$,
  $\CPOrdered{\Release{n}{k}}{\Acquire{n}{l}} \totalOrder e$, and since $\lhsSet{} = \rhsSet{}$,
  therefore $\thr{T} \in \rhsSet{}$.
  Thus $\exists e' \mid \getThread{e'} = \thr{T} \land \CPOrdered{\Acquire{m}{q}}{e'} \totalOrder e$.
  Since $e' \totalOrder e$ and $\getThread{e'} = \getThread{e}$, \HBOrdered{e'}{e} by the definition of \HB.
  Thus \CPOrdered{\Acquire{m}{q}}{e} by \CP \RuleC and therefore \CPOrdered{\Release{m}{q}}{\Acquire{m}{i}} by \CP \RuleB (recall $e = \Release{m}{i}$).
  If $q = j$, then \CPOrdered{\Release{m}{j}}{\Acquire{m}{i}}.
  Otherwise, $q > j$ so \HBOrdered{\Acquire{m}{j}}{\Acquire{m}{q}} by the definition of \HB.
  Thus \CPOrdered{\Acquire{m}{j}}{e} and therefore \CPOrdered{\Release{m}{j}}{\Acquire{m}{i}}.
  We have thus determined that $\PCPLockPlain{\rhoprime}{m}{j} \in \PCPLockSetPlain{\rho}{} \land \CPOrdered{\Release{m}{j}}{\Acquire{m}{i}} \totalOrder e$, \ie,
  $\rhoprime \in \lhsSet{}$ and thus $\rhoprime \in \rhsSet{}$.

  \end{description}

\paragraph{\underline{Superset direction}:} 

Let $\rhoprime \in \rhsSet{}$.
Since $\lhsSet{} = \rhsSet{}$, $\rhoprime \in \lhsSet{}$.
Since Algorithm~\ref{alg:pre-release} maintains
$\CPLockSetPlain{\rho}{}^+ \supseteq \CPLockSetPlain{\rho}{}$ and
$\PCPLockSetPlain{\rho}{}^+ \supseteq \PCPLockSetPlain{\rho}{}$,
therefore $\rhoprime \in \lhsSet{+}$.
\end{proof}


\noindent
Next we introduce and prove a lemma that helps to show that 
\raptor maintains invariants despite removing \CCP \set elements.
To do so, we define a concept called \emph{\CPdistance}
that measures the number of times that \RuleB and \RuleC must be applied for a \CP ordering
between two critical sections on the same lock ordered by \CP:


\begin{definition*}
The \emph{\CPdistance} of \OwnerElement{m}{j} and \OwnerElement{m}{i}, \dist{m}{i}{j}, 
is defined as follows:
{\small
\[
\dist{m}{i}{j} =
\begin{cases}
\renewcommand{\xxrightarrow}[1]{\raisebox{-2pt}{$\xrightarrow{#1}$}}
0 \textnormal{ \textbf{if} } \exists e, e', i', j' \mid j \le j' < i' \le i \land
\conflicts{e}{e'} \land
\POOrdered{\POOrdered{\Acquire{m}{j'}}{e}}{\Release{m}{j'}} \land
\POOrdered{\POOrdered{\Acquire{m}{i'}}{e'}}{\Release{m}{i'}} \\
1 + \min \, \dist{n}{l}{k} \mid 
      \HBOrdered{\CPOrdered{\HBOrdered{\Acquire{m}{j}}{\Release{n}{k}}}{\Acquire{n}{l}}}{\Release{m}{i}} \textnormal{ \textbf{otherwise}}
\end{cases}
\]
}
\end{definition*}

\begin{lemma}

Let $e$ be any release event, \ie, $e = \Release{m}{i}$ executed by thread \thr{T}.
If there exists $j$ such that \CPOrdered{\Acquire{m}{j}}{e},
and Figure~\ref{Fig:invariants}'s invariants
hold for all execution points up to and including the point immediately before $e$,
then after the pre-release algorithm (Algorithm~\ref{alg:pre-release}) executes
but before the release algorithm (Algorithm~\ref{alg:release}) executes,
one or both of the following hold:
\begin{itemize}
  \item $\exists j' \ge j \mid \thr{T} \in \CPLockSet{m}{j'}$
  \item {\small $\exists j' \ge j \mid \exists \HBLock{n}{k}  \mid \code{n} \ne \code{m} \land
  \PCPLock{T}{n}{k} \in \PCPLockSet{m}{j'} \land \exists l \mid \CPOrdered{\Release{n}{k}}{\Acquire{n}{l}} \totalOrder e \land
  \dist{n}{l}{k} < \dist{m}{i}{j}$}
\end{itemize}

\label{lem:helper}
\end{lemma}

\noindent
The intuition here is that if \CPOrdered{\Acquire{m}{j}}{e} but
\thr{T} is \emph{not} in some $\CPLockSet{m}{j'}$ ($j' \ge j$),
then there exists some \OwnerElement{n}{k} which has lesser \CPdistance than \OwnerElement{m}{j} such that
\CPOrdered{\Acquire{m}{j}}{e} is \PCP dependent on \CPOrdered{\Acquire{n}{k}}{\Release{n}{l}}
(where \OwnerElement{n}{l} is an ongoing critical section).
As a result, the pre-release algorithm can safely transfer any \PCP ordering dependent on \OwnerElement{m}{j'}
to be dependent on \OwnerElement{n}{k}, preserving invariants.

\input{Distance-Proof}

\noindent
We are ready to prove Theorem~\ref{thm:invariants}, which states that after every event, \goldCP
maintains the invariants in Figure~\ref{Fig:invariants}.

%


\begin{proof}
By induction on the observed total order of events (\totalOrder).

\paragraph{\underline{Base case}:}

Let $e$ be a first ``no-op'' event in \totalOrder that precedes all program events and has no effect on analysis state.
Before and thus after $e$,
all \CP and \CCP sets are empty, so the LHS of the \CPinv is $\emptyset$.
The RHS of the \CPinv is $\emptyset$ because there is no earlier event $e' \totalOrder e$.

%

\paragraph{\underline{Inductive step}:}

Suppose the invariants in Figure~\ref{Fig:invariants} hold immediately before an event $e$ executed by thread \thr{T}.
Let $e^+$ be the event immediately after $e$ in the observed total order ($\totalOrder$).
Let $\rho$ be any \set owner.
Let \erho be the event corresponding to $\rho$, \ie,
$\erho = \Write{x}{h}$ if $\rho = \OwnerElement{x}{h}$, 
$\erho = \Read{x}{h}{T}$ if $\rho = \OwnerElementRead{x}{h}{T}$, or
$\erho = \Acquire{m}{h}$ if $\rho = \OwnerElement{m}{h}$.

We define the following abbreviations for the left- and right-hand sides of the \CPinv:

Let $\lhsSet{} = \CPLockSetPlain{\rho}{} \cup \myset{\rhoprime \mid \big( \exists \HBLock{n}{k} \mid \PCPLock{\rhoprime}{n}{k} \in \PCPLockSetPlain{\rho}{} \land
\exists j \mid \CPOrdered{\Release{n}{k}}{\Acquire{n}{j}} \totalOrder e\big)}$.

Let $\lhsSet{+} = \CPLockSetPlain{\rho}{}^+ \cup \myset{\rhoprime \mid \big( \exists \HBLock{n}{k} \mid \PCPLock{\rhoprime}{n}{k} \in \PCPLockSetPlain{\rho}{}^+ \land
\exists j \mid \CPOrdered{\Release{n}{k}}{\Acquire{n}{j}} \totalOrder e^+\big)}$.

Let $\rhsSet{} = \myset{\rhoprime \mid \big(\exists e' \mid \comp{\rhoprime}{e'} \land \CPOrdered{\erho}{e'} \totalOrder e\big)}$.

Let $\rhsSet{+} = \myset{\rhoprime \mid \big(\exists e' \mid \comp{\rhoprime}{e'} \land \CPOrdered{\erho}{e'} \totalOrder e^+\big)}$.

Inductive hypothesis: Suppose Figure~\ref{Fig:invariants}'s invariants hold before $e$, \ie,
$\lhsSet{} = \rhsSet{}$.

If $e$ is a release,
we consider the analysis state \emph{before} $e$ to be the state \emph{after} the pre-release algorithm (Algorithm~\ref{alg:pre-release}) has executed for $e$,
\ie, let \lhsSet{} reflect the state \emph{after} the pre-release algorithm has executed.
By Lemma~\ref{lem:pre-release}, Figure~\ref{Fig:invariants}'s invariants (\ie, \lhsSet{} = \rhsSet{}) hold at this point.

To show $\lhsSet{+} = \rhsSet{+}$, we show
$\lhsSet{+} \subseteq \rhsSet{+}$ (subset) and
$\lhsSet{+} \supseteq \rhsSet{+}$ (superset) in turn.

\paragraph{\underline{Subset direction}:}

Let $\rhoprime \in \lhsSet{+}$.

Either $\rhoprime \in \CPLockSetPlain{\rho}{}^+$ or 
$\exists \HBLock{n}{k} \mid \PCPLockPlain{\rhoprime}{n}{k} \in \PCPLockSetPlain{\rho}{}^+ \land \exists j \mid \CPOrdered{\Release{n}{k}}{\Acquire{n}{j}} \totalOrder e^+$ 
(or both):

	\begin{description}
	
	\item[Case 1:] $\rhoprime \in \CPLockSetPlain{\rho}{}^+$
	
	If $\rhoprime \in \CPLockSetPlain{\rho}{}$, then by the inductive hypothesis, $\rhoprime \in \rhsSet{}$, \ie, $\exists e' \mid \comp{\rhoprime}{e'} \land \CPOrdered{\erho}{e'} \totalOrder e$.
	Since $e \totalOrder e^+$, $\exists e' \mid \CPOrdered{\erho}{e'} \totalOrder e^+$, so $\rhoprime \in \rhsSet{+}$.

	Otherwise ($\rhoprime \notin \CPLockSetPlain{\rho}{}$), the analysis adds $\rhoprime$ to $\CPLockSetPlain{\rho}{}^+$:
	
	  \begin{description}

	  \item[Case 1a:] $e = \Write{x}{i}$.
	  Algorithm~\ref{alg:write} adds $\rhoprime$ to $\CPLockSetPlain{\rho}{}^+$ 
	  at line~\ref{line:write:CPEdge-established}, \ref{line:write:CPEdge-established:read}, 
	  \ref{line:write:add-existingCP-xi} or \ref{line:write:add-existingCP-xi:read}:
	
		\begin{itemize}
		\item
		Lines~\ref{line:write:CPEdge-established} and \ref{line:write:CPEdge-established:read}
		execute within line~\ref{line:write:forall-locks}'s \textbf{for} loop,
		so let $\code{m} \in \heldBy{T}$.
		Since line~\ref{line:write:CPEdge-established} or \ref{line:write:CPEdge-established:read} executes, 
		the \textbf{if} condition in line~\ref{line:write:LSlock-check} or \ref{line:write:CPEdge-established:read}
		evaluates to true, so let $h < i$ be such that
		(1) $\thr{T} \notin \POLockSet{x}{h}$ or $\thr{T} \notin \POLockSetRead{x}{h}{t}$, and
		(2) let $j$ be such that $\LSLock{m}{j} \in \HBLockSet{x}{h}$ 
		or $\LSLock{m}{j} \in \HBLockSetRead{x}{h}{t}$, respectively.
		Lines~\ref{line:write:CPEdge-established} and \ref{line:write:CPEdge-established:read}
		add \CPThread{T} to $\CPLockSet{m}{j}^+$,
		so $\rho = \OwnerElement{m}{j}$ and $\rhoprime = \thr{T}$.

		By the inductive hypothesis on $\LSLock{m}{j} \in \HBLockSet{x}{h}$ 
		or $\LSLock{m}{j} \in \HBLockSetRead{x}{h}{t}$ (\HBcriticalsectioninv),
		$\POOrdered{\POOrdered{\Acquire{m}{j}}{\allowbreak e''}}{\Release{m}{j}} \land e'' \totalOrder e$,
		where $e''$ is \Write{x}{h} or \Read{x}{h}{t}.
		Because $\code{m} \in \heldBy{T}$ and $\thr{T} \notin \POLockSet{x}{h}$
		or $\thr{T} \notin \POLockSetRead{x}{h}{t}$,
		let $k > j$ be such that $\POOrdered{\POOrdered{\Acquire{m}{k}}{e}}{\Release{m}{k}}$.
		By \CP \RuleA, \CPOrdered{\Release{m}{j}}{\Acquire{m}{k}}.
		By \CP \RuleC, \CPOrdered{\Acquire{m}{j}}{e}.
		Since $\getThread{e} = \thr{T} \land \CPOrdered{\Acquire{m}{j}}{e} \totalOrder e^+$,
		therefore $\thr{T} \in \rhsSet{+}$.

		\item Line~\ref{line:write:add-existingCP-xi} adds $\xi$ to $\CPLockSetPlain{\rho}{}^+$,
		so $\rho = \OwnerElement{x}{i-1}$, $\erho = \Write{x}{i-1}$, $\rhoprime = \xi$, and
		$\CPThread{T} \in \CPLockSetPlain{\rho}{}$ (line~\ref{line:write:CPEdge-Exists}).
		By the inductive hypothesis (\CPinv), let $e'$ be such that $\getThread{e'} = \thr{T} \land \CPOrdered{ \Write{x}{i-1}}{e'} \totalOrder e$.
		By \CP \RuleC, since $\getThread{e'} = \thr{T} = \getThread{e}$,
		$\CPOrdered{\Write{x}{i-1}}{e} \totalOrder e^+$.
		Note that $e = \Write{x}{i}$ is \exi from Figure~\ref{Fig:invariants}.
		Therefore $\xi \in \rhsSet{+}$.
		
		\item Line~\ref{line:write:add-existingCP-xi:read} adds $\xi$ to $\CPLockSetPlain{\rho}{}^+$,
		so $\rho = \OwnerElementRead{x}{i}{t}$, $\erho = \Read{x}{i}{t}$, $\rhoprime = \xi$, and
		$\CPThread{T} \in \CPLockSetPlain{\rho}{}$ (line~\ref{line:write:CPEdge-Exists:read}).
		By the inductive hypothesis (\CPinv), let $e'$ be such that 
		$\getThread{e'} = \thr{T} \land \CPOrdered{ \Read{x}{i}{t}}{e'} \totalOrder e$.
		By \CP \RuleC, since $\getThread{e'} = \thr{T} = \getThread{e}$,
		$\CPOrdered{\Read{x}{i}{t}}{e} \totalOrder e^+$.
		Note that $e = \Write{x}{i}$ is \exi from Figure~\ref{Fig:invariants}.
		Therefore $\xi \in \rhsSet{+}$.  
	
		\end{itemize}
		
	  \item[Case 1b:] $e = \Read{x}{i}{T}$.
	  Algorithm~\ref{alg:read} adds $\rhoprime$ to $\CPLockSetPlain{\rho}{}^+$ at line~\ref{line:read:RuleA:add} or \ref{line:read:establishCP}:
	  	
		\begin{itemize}
	  	  \item Line~\ref{line:read:RuleA:add} executes within line~\ref{line:read:RuleA:start}'s \textbf{for} loop,
	  	  so let $\code{m} \in \heldBy{T}$. Since line~\ref{line:read:RuleA:add} executes,
	  	  its \textbf{if} condition evaluates to true, so let $h < i$ be such that 
	  	  (1) $\thr{T} \notin \POLockSet{x}{h}$, and
	  	  (2) let $j$ be such that $\LSLock{m}{j} \in \HBLockSet{x}{h}$.
	  	  Line~\ref{line:read:RuleA:add} adds \thr{T} to $\CPLockSet{m}{j}^+$,
	  	  so $\rho = \OwnerElement{m}{j}$ and $\rhoprime = \thr{T}$.
	  	  
	  	  By the inductive hypothesis on $\LSLock{m}{j} \in \HBLockSet{x}{h}$ (\HBcriticalsectioninv),
	  	  $\POOrdered{\Acquire{m}{j}}{\POOrdered{\Write{x}{h}}{\Release{m}{j}}} \land \Write{x}{h} \totalOrder e$.
	  	  Because $\code{m} \in \heldBy{T}$ and $\thr{T} \notin \POLockSet{x}{h}$,
	  	  let $k > j$ be such that 
	  	  \POOrdered{\Acquire{m}{k}}{\POOrdered{e}{\Release{m}{k}}}.
	  	  By \CP \RuleA, \CPOrdered{\Release{m}{j}}{\Acquire{m}{k}}.
	  	  By \CP \RuleC, \CPOrdered{\Acquire{m}{j}}{e}.
	  	  Since $\getThread{e} = \thr{T} \land \CPOrdered{\Acquire{m}{j}}{e} \totalOrder e^+$,
	  	  therefore $\thr{T} \in \rhsSet{+}$.
	  	  
	  	  \item Line~\ref{line:read:establishCP} adds \xiT{T} to $\CPLockSetPlain{\rho}{}^+$,
	  	  so $\rho = \OwnerElement{x}{i}$, $\erho = \Write{x}{i}$, $\rhoprime = \xiT{T}$, 
	  	  and $\thr{T} \in \CPLockSetPlain{\rho}{}$ (line~\ref{line:read:establishCP}).
	  	  By the inductive hypothesis (\CPinv),
	  	  let $e'$ be such that $\getThread{e'} = \thr{T} \land \CPOrdered{\Write{x}{i}}{e'} \totalOrder e$.
	  	  By \CP \RuleC, since $\getThread{e'} = \thr{T} = \getThread{e}$,
	  	  $\CPOrdered{\Write{x}{i}}{e} \totalOrder e^+$.
	  	  Note that $e = \Read{x}{i}{T}$ is \exiT{T} from Figure~\ref{Fig:invariants}.
	  	  Therefore $\xiT{T} \in \rhsSet{+}$.
	  	  
	  	\end{itemize}

	  \item[Case 1c:] $e = \Acquire{m}{i}$.
	  Algorithm~\ref{alg:acquire} adds $\rhoprime$ to $\CPLockSetPlain{\rho}{}^+$ only at line~\ref{line:acq:add-CPThread},
		which adds \CPThread{T} to $\CPLockSetPlain{\rho}{}^+$,
		so $\rhoprime = \CPThread{T}$ and $\CPLock{m} \in \CPLockSetPlain{\rho}{}$ (line~\ref{line:acq:CP-condition}).
		By the inductive hypothesis (\CPinv), $\exists j \mid \CPOrdered{\erho}{\Release{m}{j}} \totalOrder e$.
		By \CP \RuleC, \CPOrdered{\erho}{e} (since $e = \Acquire{m}{i}$).
		Because $\getThread{e} = \thr{T} \land \CPOrdered{\erho}{e} \totalOrder e^+$, therefore $\CPThread{T} \in \rhsSet{+}$.
	
	  \item[Case 1d:] $e = \Release{m}{i}$.
	  Algorithm~\ref{alg:release} adds $\rhoprime$ to $\CPLockSetPlain{\rho}{}^+$ only at line~\ref{line:rel:add-CPlock},
		which adds \CPLock{m} to $\CPLockSetPlain{\rho}{}^+$,
		so $\rhoprime = \CPLock{m}$ and $\CPThread{T} \in \CPLockSetPlain{\rho}{}$ (line~\ref{line:rel:CP-condition}).
		By the inductive hypothesis (\CPinv), let $e'$ be such that $\getThread{e'} = \thr{T} \land \CPOrdered{\erho}{e'} \totalOrder e$.
		By the definition of HB, \HBOrdered{e'}{e}, and thus $\CPOrdered{\erho}{e} \totalOrder e^+$ by \CP \RuleC.
		Therefore $\CPLock{m} \in \rhsSet{+}$.
	  
	\end{description}

	\item[Case 2:] Let $\OwnerElement{n}{k}$ and $\OwnerElement{n}{j}$ be such that $\CCPLockPlain{\rhoprime}{n}{k} \in \PCPLockSetPlain{\rho}{}^+ \land \CPOrdered{\Release{n}{k}}{\Acquire{n}{j}} \totalOrder e^+$.
	
	If $\CCPLockPlain{\rhoprime}{n}{k} \in \CCPLockSetPlain{\rho}{} \land \CPOrdered{\Release{n}{k}}{\Acquire{n}{j}} \totalOrder e$, then by the inductive hypothesis (\CPinv), $\rhoprime \in \rhsSet{}$, \ie, $\exists e' \mid \comp{\rhoprime}{e'} \land \CPOrdered{\erho}{e'} \totalOrder e$.
	Since $e \totalOrder e^+$, $\CPOrdered{\erho}{e'} \totalOrder e^+$, and thus $\rhoprime \in \rhsSet{+}$.
	
	Otherwise, $\neg \big(\CCPLockPlain{\rhoprime}{n}{k} \in \CCPLockSetPlain{\rho}{} \land \CPOrdered{\Release{n}{k}}{\Acquire{n}{j}} \totalOrder e\big)$.
	In fact, we can conclude that $\CCPLockPlain{\rhoprime}{n}{k} \notin \CCPLockSetPlain{\rho}{}$ using the following reasoning.
	If $\neg\big(\CPOrdered{\Release{n}{k}}{\Acquire{n}{j}} \totalOrder e\big)$,
	then $\Acquire{n}{j} \not\totalOrder e$ (since \CPOrdered{\Release{n}{k}}{\Acquire{n}{j}}),
	and thus $e = \Acquire{n}{j}$ (since $\Acquire{n}{j} \totalOrder e^+$ and $e$ immediately precedes $e^+$)---in which case
	the inductive hypothesis (\CCPconstraintinv) ensures $\CCPLockPlain{\rhoprime}{n}{k} \notin \CCPLockSetPlain{\rho}{}$.
	Thus the analysis adds \CCPLockPlain{\rhoprime}{n}{k} to $\CCPLockSetPlain{\rho}{}^+$:

	  \begin{description}

	  \item[Case 2a:] $e = \Write{x}{i}$.
		Algorithm~\ref{alg:write} adds \CCPLockPlain{\rhoprime}{n}{k} to $\CCPLockSetPlain{\rho}{}^+$ 
		at line~\ref{line:write:add-conditional-xi} or \ref{line:write:add-conditional-xiRead}.
		
		\begin{itemize}
		  \item Line~\ref{line:write:add-conditional-xi} 
		  adds $\CCPLockPlain{\xi}{n}{k}$ to $\CCPLockSetPlain{\rho}{}^+$,
		  so $\rho = \OwnerElement{x}{i-1}$, $\erho = \Write{x}{i-1}$, $\rhoprime = \xi$, and
		  $\CCPThread{T}{n}{k} \in \CCPLockSet{x}{i-1}$ (line~\ref{line:write:forall-dependent-locks}).
		  By the inductive hypothesis (\CPinv), let $e'$ be such that $\getThread{e'} = \thr{T} \land \CPOrdered{\Write{x}{i-1}}{e'} \totalOrder e$.
		  By the definitions of \HB and \CP, since $\getThread{e'} = \thr{T} = \getThread{e}$,
		  $\CPOrdered{\Write{x}{i-1}}{e} \totalOrder e^+$.
		  Note that $e = \Write{x}{i}$ is \exi from Figure~\ref{Fig:invariants}.
		  Therefore $\xi \in \rhsSet{+}$.
		
		  \item
		  Line~\ref{line:write:add-conditional-xiRead}
		  adds $\CCPLockPlain{\xi}{n}{k}$ to $\CCPLockSetPlain{\rho}{}^+$,
		  so $\rho = \OwnerElementRead{x}{i-1}{t}$, 
		  $\erho = \Read{x}{i-1}{t}$, $\rhoprime = \xi$, and
		  $\CCPThread{T}{n}{k} \in \CCPLockSetRead{x}{i-1}{t}$ (line~\ref{line:write:add-conditional-xiRead}).
		  By the inductive hypothesis (\CPinv),
		  let $e'$ be such that $\getThread{e'} = \thr{T} \land \CPOrdered{\Read{x}{i-1}{t}}{e'} \totalOrder e$.
		  By the definitions of \HB and \CP,
		  since $\getThread{e'} = \thr{T} = \getThread{e}$,
		  $\CPOrdered{\Read{x}{i-1}{t}}{e} \totalOrder e^+$.
		  Note that $e = \Write{x}{i}$ is \exi from Figure~\ref{Fig:invariants}
		  Therefore $\xi \in \rhsSet{+}$.
		  
		\end{itemize}

	  \item[Case 2b:] $e = \Read{x}{i}{T}$.
	  Algorithm~\ref{alg:read} adds \CCPLockPlain{\rhoprime}{n}{k} to $\CCPLockSetPlain{\rho}{}^+$ at 
		line~\ref{line:read:establishCCP}.
	    So $\rho = \OwnerElement{x}{i}$, $\erho = \Write{x}{i}$, $\rhoprime = \xiT{T}$,
	    and $\CCPThread{T}{n}{k} \in \CCPLockSet{x}{i}$ (line~\ref{line:read:establishCCP}).
	    By the inductive hypothesis (\CPinv),
	    let $e'$ be such that 
	    $\getThread{e'} = \thr{T} \land \CPOrdered{\Write{x}{i}}{e'} \totalOrder e$.
	    By the definitions of \HB and \CP, 
	    since $\getThread{e'} = \thr{T} = \getThread{e}$,
	    $\CPOrdered{\Write{x}{i}}{e} \totalOrder e^+$.
	    Note that $e = \Read{x}{i}{T}$ is \exiT{T} from Figure~\ref{Fig:invariants}.
	    Therefore $\xiT{T} \in \rhsSet{+}$.
	  
	  \item[Case 2c:] $e = \Acquire{m}{i}$.
	  Algorithm~\ref{alg:acquire} adds \CCPLockPlain{\rhoprime}{n}{k} to $\CCPLockSetPlain{\rho}{}^+$ at
		line~\ref{line:acq:add-transfer-CCPThread} or \ref{line:acq:add-CCPThread}.
	
		\begin{itemize}
		
		\item If line~\ref{line:acq:add-transfer-CCPThread} adds \PCPLockPlain{\rhoprime}{n}{k} to $\PCPLockSetPlain{\rho}{}^+$,
		then $\rhoprime = \thr{T}$ and
		$\PCPLock{m}{n}{k} \in \PCPLockSetPlain{\rho}{}$ (line~\ref{line:acq:CCP-condition}).
		By the inductive hypothesis (\CPinv), $\exists l \mid \CPOrdered{\erho}{\Release{m}{l}} \totalOrder e$.
		By \CP \RuleC, \CPOrdered{\erho}{e} since $e = \Acquire{m}{i}$.
		Since $\getThread{e} = \thr{T} \land \CPOrdered{\erho}{e} \totalOrder e^+$, therefore $\CPThread{T} \in \rhsSet{+}$.
		
		\item If line~\ref{line:acq:add-CCPThread} adds \PCPLockPlain{\rhoprime}{n}{k} to $\PCPLockSetPlain{\rho}{}^+$,
		then $\HBLock{n}{k} \in \HBLockSetPlain{\rho}{}$
		(recall that $\HBLock{n}{k} \in \HBLockSetPlain{\rho}{}$ if $\LSLock{n}{k} \in \HBLockSetPlain{\rho}{}$)
	    according to line~\ref{line:acq:HB-condition}.
		Furthermore, $\code{n} = \code{m}$ and $\rhoprime = \thr{T}$.
		By the inductive hypothesis (\invHB, \invHBindex, and \invHBcriticalsection invariants), $\HBOrdered{\erho}{\Release{n}{k}} \totalOrder e$.
		Since we know $\CPOrdered{\Release{n}{k}}{\Acquire{n}{j}} \totalOrder e^+$ (from the beginning of Case 2)
		and $e = \Acquire{n}{i}$ (since $\code{n} = \code{m}$),
		therefore \CPOrdered{\erho}{e} by \CP \RuleC (regardless of whether $j = i$ or $j < i$).
		Since $\getThread{e} = \thr{T} \land \CPOrdered{\erho}{e} \totalOrder e^+$,
		$\CPThread{T} \in \rhsSet{+}$.
		
		\end{itemize}

	  \item[Case 2d:] $e = \Release{m}{i}$.
		Algorithm~\ref{alg:release} adds \PCPLockPlain{\rhoprime}{n}{k} to $\PCPLockSetPlain{\rho}{}^+$ only at
		line~\ref{line:rel:add-dependent-locks}, so $\rhoprime = \CPLock{m}$ and
		$\PCPThread{T}{n}{k} \in \PCPLockSetPlain{\rho}{}$ (line~\ref{line:rel:CCP-condition}).
		By the inductive hypothesis (\CPinv),
		$\exists l \mid \CPOrdered{\erho}{\Release{m}{l}} \totalOrder e$ since $\rhoprime = \code{m}$.
		Since \HBOrdered{\Release{m}{l}}{\Acquire{m}{i}}
		and $e = \Release{m}{i}$ and $\getThread{e} = \thr{T}$,
		by \CP \RuleC, $\CPOrdered{\erho}{e} \totalOrder e^+$.
		Therefore $\CPLock{m} \in \rhsSet{+}$.

	  \end{description}

	\end{description}

\paragraph{\underline{Superset direction}:}

Let $\rhoprime \in \rhsSet{+}$. Either $\rhoprime \in \rhsSet{}$ or not:

  \begin{description}

  \item[Case 1:] $\rhoprime \in \rhsSet{}$

  Then $\rhoprime \in \lhsSet{}$ by the inductive hypothesis.
  If $e$ is a write, read, or acquire, then $\rhoprime \in \lhsSet{+}$ because Algorithms~\ref{alg:write}, \ref{alg:read}, and \ref{alg:acquire}
  maintain $\CPLockSetPlain{\rho}{}^+ \supseteq \CPLockSetPlain{\rho}{}$ and $\CCPLockSetPlain{\rho}{}^+ \supseteq \PCPLockSetPlain{\rho}{}$, $\rhoprime \in \lhsSet{+}$.
  (Line~\ref{line:read:resetRead} of Algorithm~\ref{alg:read} overwrites 
  $\CPLockSetRead{x}{i}{\code{T}}^+$ and $\CCPLockSetRead{x}{i}{T}^+$,
  which is correct because, in essence, the latest \Read{x}{i}{T} ``replaces'' any previous \Read{x}{i}{T}.
  Thus for $\rho = \OwnerElement{x}{i}$ or $\rho = \OwnerElementRead{x}{i}{T}$,
  no $\rhoprime \in \rhsSet{}$ exists.)
  But we must consider $e = \Release{m}{i}$ since Algorithm~\ref{alg:release} removes some \PCP elements.
  Since $\rhoprime \in \lhsSet{}$,
  $\rhoprime \in \CPLockSetPlain{\rho}{}$ or $\exists \HBLock{n}{k} \mid \PCPLockPlain{\rhoprime}{n}{k} \in \PCPLockSetPlain{\rho}{} \land
  \exists j \mid \CPOrdered{\Release{n}{k}}{\Acquire{n}{j}} \totalOrder e$:

    \begin{description}

    \item[Case 1a:] If $\rhoprime \in \CPLockSetPlain{\rho}{}$, then $\rhoprime \in \CPLockSetPlain{\rho}{}^+$
    because Algorithm~\ref{alg:release} maintains $\CPLockSetPlain{\rho}{}^+ \supseteq \CPLockSetPlain{\rho}{}$.

    \item[Case 1b:] Let $\HBLock{n}{k}$ be such that $\PCPLockPlain{\rhoprime}{n}{k} \in \PCPLockSetPlain{\rho}{} \land
    \exists j \mid \CPOrdered{\Release{n}{k}}{\Acquire{n}{j}} \totalOrder e$.

    If $\code{n} \ne \code{m}$,
    then Algorithm~\ref{alg:release} does \emph{not} remove $\PCPLockPlain{\rhoprime}{n}{k}$ from $\PCPLockSetPlain{\rho}{}^+$,
    so $\PCPLockPlain{\rhoprime}{n}{k} \in \PCPLockSetPlain{\rho}{}^+$.

    Otherwise, $\code{n} = \code{m}$.
    Since $\exists j \mid \CPOrdered{\Release{n}{k}}{\Acquire{n}{j}} \totalOrder e$ and $e = \Release{m}{i}$,
    therefore \CPOrdered{\Acquire{m}{k}}{\Acquire{m}{i}} by \CP \RuleC and \RuleB.
    By the inductive hypothesis (\CPinv) and Lemma~\ref{lem:helper},
    $\thr{T} \in \CPLockSet{m}{k}$ or
    $\exists \HBLock{o}{l} \mid \code{o} \ne \code{m} \land \PCPThread{T}{o}{l} \in \PCPLockSet{m}{k} \land
    \exists h \mid \CPOrdered{\Release{o}{l}}{\Release{o}{h}} \totalOrder e$.

      \begin{description}

      \item [Case 1b(i):] $\thr{T} \in \CPLockSet{m}{k}$

      In the pre-release algorithm (Algorithm~\ref{alg:pre-release}),
      line~\ref{line:prerel:trigger-condition} evaluated to true for
      $\PCPLockPlain{\rhoprime}{m}{l} \in \PCPLockSetPlain{\rho}{} \mid l \le k$ (line~\ref{line:prerel:CCP-condition}).
      Thus line~\ref{line:prerel:add-CPSigma} executed, so $\rhoprime \in \CPLockSetPlain{\rho}{}$.
      Since the release algorithm (Algorithm~\ref{alg:release}) maintains $\CPLockSetPlain{\rho}{}^+ \supseteq \CPLockSetPlain{\rho}{}$,
      therefore $\rhoprime \in \lhsSet{+}$.

      \item [Case 1b(ii):] Let \HBLock{o}{l} be such that
      $\code{o} \ne \code{m} \land \PCPThread{T}{o}{l} \in \PCPLockSet{m}{k} \land
      \exists h \mid \CPOrdered{\Release{o}{l}}{\Release{o}{h}} \totalOrder e$.

      Then the pre-release algorithm (Algorithm~\ref{alg:pre-release}) 
      executed line~\ref{line:prerel:add-dependent-transfers}
      for $\PCPLockPlain{\rhoprime}{m}{k} \in \PCPLockSetPlain{\rho}{}$ (line~\ref{line:prerel:CCP-condition})
      and $\PCPThread{T}{o}{l} \in \PCPLockSet{m}{k}$ (line~\ref{line:prerel:CCP-transfer-condition}).
      Line~\ref{line:prerel:add-dependent-transfers} added
      \PCPLockPlain{\rhoprime}{o}{l} to $\PCPLockSetPlain{\rho}{}$.
      The release algorithm (Algorithm~\ref{alg:release}) does \emph{not} remove
      \PCPLockPlain{\rhoprime}{o}{l} because $\code{o} \ne \code{m}$ (line~\ref{line:rel:removal}).
      Therefore $\PCPLockPlain{\rhoprime}{o}{l} \in \PCPLockSetPlain{\rho}{}^+$ and thus $\rhoprime \in \lhsSet{+}$.

      \end{description}

    \end{description}


  \item [Case 2:] $\rhoprime \notin \rhsSet{}$

  Thus $\exists  e' \mid \comp{\rhoprime}{e'} \land \CPOrdered{\erho}{e'} \totalOrder e^+$ but
  $\nexists e' \mid \comp{\rhoprime}{e'} \land \CPOrdered{\erho}{e'} \totalOrder e$,  which implies
  $\comp{\rhoprime}{e} \land \CPOrdered{\erho}{e}$.
  $e$ is either a write, read, acquire, or release:

    \begin{description}

    \item [Case 2a:] $e = \Write{x}{i}$

    By the definition of \CP, since $e$ is not an acquire,
    $\exists e' \mid \getThread{e'} = \thr{T} \land \CPOrdered{\erho}{e'} \totalOrder e$,
    \ie, $\thr{T} \in \rhsSet{}$.
    According to the definition of \comp{\rhoprime}{e},
    either (1) $\rhoprime = \thr{T}$ or 
    (2) $\rhoprime = \xi$ (note that $e = \Write{x}{i}$ is \exi from Figure~\ref{Fig:invariants}).
    However, $\rhoprime \notin \rhsSet{}$, so $\rhoprime = \xi$.
    Therefore $\rho = \OwnerElement{x}{i-1}$, $\erho = \Write{x}{i-1}$, and \CPOrdered{\Write{x}{i-1}}{e};
    or $\rho = \OwnerElementRead{x}{i-1}{T}$, $\erho = \Read{x}{i-1}{T}$, and \CPOrdered{\Read{x}{i-1}{T}}{e}.
    By the inductive hypothesis on $\thr{T} \in \rhsSet{}$ (\CPinv),
    $\CPThread{T} \in \CPLockSet{x}{i-1}$,
    $\CPThread{T} \in \CPLockSetRead{x}{i-1}{t}$,
    $\exists \HBLock{n}{k} \mid \CCPLock{T}{n}{k} \in \CCPLockSet{x}{i-1} \land \exists j \mid \CPOrdered{\Release{n}{k}}{\Acquire{n}{j}} \totalOrder e$, or
    $\exists \HBLock{n}{k} \mid \CCPLock{T}{n}{k} \in \CCPLockSetRead{x}{i-1}{t} \land \exists j \mid \CPOrdered{\Release{n}{k}}{\Acquire{n}{j}} \totalOrder e$.

      \begin{description}

      \item[Case 2a(i):] $\CPThread{T} \in \CPLockSet{x}{i-1}$

		Then Algorithm~\ref{alg:write}'s line~\ref{line:write:CPEdge-Exists} evaluates to true,
		so line~\ref{line:write:add-existingCP-xi} adds $\xi$ to $\CPLockSet{x}{i-1}^+$;
		thus $\xi \in \lhsSet{+}$.
		
		\item[Case 2a(ii):] $\CPThread{T} \in \CPLockSetRead{x}{i-1}{t}$
		
		Line~\ref{line:write:CPEdge-Exists:read} evaluates to true,
		so line~\ref{line:write:CPEdge-Exists:read} adds $\xi$ to \CPLockSetRead{x}{i-1}{t};
		thus $\xi \in \lhsSet{+}$.

		\item[Case 2a(iii):] Let \HBLock{n}{k} be such that $\CCPThread{T}{n}{k} \in \CCPLockSet{x}{i-1} \land \exists j \mid \CPOrdered{\Release{n}{k}}{\Acquire{n}{j}} \totalOrder e$.

		Since $\PCPThread{T}{n}{k} \in \PCPLockSet{x}{i-1}$  matches line~\ref{line:write:forall-dependent-locks} in Algorithm~\ref{alg:write},
		line~\ref{line:write:add-conditional-xi} adds \CCPLockPlain{\xi}{n}{j} to $\CCPLockSet{x}{i-1}^+$.
		So $\exists j \mid \CPOrdered{\Release{n}{k}}{\Acquire{n}{j}} \totalOrder e \totalOrder e^+$;
		thus $\xi \in \lhsSet{+}$.
		
		\item[Case 2a(iv):] Let \HBLock{n}{k} be such that $\CCPThread{T}{n}{k} \in \CCPLockSetRead{x}{i-1}{t} \land \exists j \mid \CPOrdered{\Release{n}{k}}{\Acquire{n}{j}} \totalOrder e$.

		Since $\CCPThread{T}{n}{k} \in \CCPLockSetRead{x}{i-1}{t}$  matches line~\ref{line:write:forall-dependent-locksRead} in Algorithm~\ref{alg:write},
		line~\ref{line:write:add-conditional-xiRead} adds \CCPLockPlain{\xi}{n}{j} to $\CCPLockSetRead{x}{i-1}{t}^+$.
		So $\exists j \mid \CPOrdered{\Release{n}{k}}{\Acquire{n}{j}} \totalOrder e \totalOrder e^+$;
		thus $\xi \in \lhsSet{+}$.

		\end{description}

	\item [Case 2b:] $e = \Read{x}{i}{T}$
	
	By the definition of \CP, since $e$ is not an acquire,
	$\exists e' \mid \getThread{e'} = \thr{T} \land \CPOrdered{\erho}{e'} \totalOrder e$,
	\ie, $\thr{T} \in \rhsSet{+}$.
	According to the definition of \comp{\rhoprime}{e}, either
	(1) $\rhoprime = \thr{T}$ or
	(2) $\rhoprime = \xiT{T}$
	(note that $e = \Read{x}{i}{T}$ is \exiT{T} from Figure~\ref{Fig:invariants}).
	However, $\rhoprime \notin \rhsSet{}$, so $\rhoprime = \xiT{T}$.
	Therefore $\rho = \OwnerElement{x}{i}$, $\erho = \Write{x}{i}$, and
	\CPOrdered{\Write{x}{i}}{e}.
	By the inductive hypothesis on $\thr{T} \in \rhsSet{}$ (\CPinv),
	$\thr{T} \in \CPLockSet{x}{i}$ or
	$\exists \OwnerElement{n}{k} \mid \CCPThread{T}{n}{k} \in \CCPLockSet{x}{i} \land \exists j \mid \CPOrdered{\Release{n}{k}}{\Acquire{n}{j}} \totalOrder e$.
	
			\begin{description}
			\item[Case 2b(i):] $\thr{T} \in \CPLockSet{x}{i}$
			
			Then Algorithm~\ref{alg:read}'s line~\ref{line:read:establishCP} evaluates to true,
			so line~\ref{line:read:establishCP} adds \xiT{T} to $\CPLockSet{x}{i}^+$;
			thus $\xiT{T} \in \lhsSet{+}$.
			
			\item[Case 2b(ii):] Let \OwnerElement{n}{k} be such that $\CCPThread{T}{n}{k} \in \CCPLockSet{x}{i} \land \exists j \mid \CPOrdered{\Release{n}{k}}{\Acquire{n}{j}} \totalOrder e$.
			
			Since $\CCPThread{T}{n}{k} \in \CCPLockSet{x}{i}$ matches line~\ref{line:read:establishCCP} in Algorithm~\ref{alg:read},
			line~\ref{line:read:establishCCP} adds \CCPThreadPlain{\xiT{T}}{n}{j} to \CCPLockSet{x}{i}.
			So $\exists j \mid \CPOrdered{\Release{n}{k}}{\Acquire{n}{j}} \totalOrder e \totalOrder e^+$;
			thus $\xiT{T} \in \lhsSet{+}$.
			 
			\end{description}
	
    \item [Case 2c:] $e = \Acquire{m}{i}$

    Thus $\rhoprime = \CPThread{T}$.
    Since $\CPOrdered{\erho}{e}$ and $\nexists e' \mid \getThread{e'} = \thr{T} \land \CPOrdered{\erho}{e'} \totalOrder e$
    (recall $\rhoprime \notin \rhsSet{}$),
    by \CP \RuleC, $\exists l \mid \HBOrdered{\CPOrdered{\erho}{\Release{m}{l}}}{e}$
    or $\exists l \mid \CPOrdered{\HBOrdered{\erho}{\Release{m}{l}}}{e}$.

      \begin{description}

      \item[Case 2c(i):] $\exists l \mid \HBOrdered{\CPOrdered{\erho}{\Release{m}{l}}}{e}$

      By the inductive hypothesis on $\CPOrdered{\erho}{\Release{m}{l}}$ (\CPinv),
      $\CPLock{m} \in \CPLockSetPlain{\rho}{}$ or
      $\exists \HBLock{n}{k} \mid \PCPLock{m}{n}{k} \in \PCPLockSetPlain{\rho}{} \land
      \exists j \mid \CPOrdered{\Release{n}{k}}{\Acquire{n}{j}} \totalOrder e$.

        \begin{itemize}

        \item $\code{m} \in \CPLockSetPlain{\rho}{}$

        Then Algorithm~\ref{alg:acquire}'s line~\ref{line:acq:CP-condition} evaluates to true,
        so line~\ref{line:acq:add-CPThread} adds \CPThread{T} to $\CPLockSetPlain{\rho}{}^+$;
        thus $\CPThread{T} \in \lhsSet{+}$.

        \item Let $\HBLock{n}{k}$ be such that $\PCPLock{m}{n}{k} \in \PCPLockSetPlain{\rho}{} \land
        \exists j \mid \CPOrdered{\Release{n}{k}}{\Acquire{n}{j}} \totalOrder e$.

        Then $\PCPLock{m}{n}{k} \in \PCPLockSetPlain{\rho}{}$ matches line~\ref{line:acq:CCP-condition} in Algorithm~\ref{alg:acquire},
        so line~\ref{line:acq:add-transfer-CCPThread} adds \PCPLock{T}{n}{k} to $\PCPLockSetPlain{\rho}{}^+$.
        So $\exists j \mid \CPOrdered{\Release{n}{k}}{\Acquire{n}{j}} \totalOrder e \totalOrder e^+$;
        thus $\CPThread{T} \in \lhsSet{+}$.
        
        \end{itemize}

      \item[Case 2c(ii):] $\exists l \mid \CPOrdered{\HBOrdered{\erho}{\Release{m}{l}}}{e}$

      Let $j$ be the minimum value such that $\HBOrdered{\erho}{\Release{m}{j}}$.
      By the inductive hypothesis (\invHB, \invHBindex, and \invHBcriticalsection invariants),
      $\HBLock{m}{j} \in \HBLockSetPlain{\rho}{}$
      (recall that $\HBLock{m}{j} \in \HBLockSetPlain{\rho}{}$ if $\LSLock{m}{j} \in \HBLockSetPlain{\rho}{}$).
      Thus line~\ref{line:acq:HB-condition} evaluates to true for $\rho$,
      so line~\ref{line:acq:add-CCPThread} executes and
      $\PCPThread{T}{m}{j} \in \PCPLockSetPlain{\rho}{}^+$.
      Furthermore, $\CPOrdered{\Release{m}{j}}{\Acquire{m}{i}} \totalOrder e^+$ ($e = \Acquire{m}{i}$).
      Therefore $\thr{T} \in \lhsSet{+}$.

      \end{description}

    \item [Case 2d:] $e = \Release{m}{i}$
    
    Since $e$ is not an acquire, by \CP \RuleC,
    $\exists e' \mid \getThread{e'} = \thr{T} \land \CPOrdered{\erho}{e'} \totalOrder e$,
    \ie, $\thr{T} \in \rhsSet{}$.
    According to the definition of \comp{\rhoprime}{e}, either $\rhoprime = \thr{T}$ or $\rhoprime = \code{m}$.
    But $\rhoprime \ne \thr{T}$ since $\rhoprime \notin \rhsSet{}$,
    so $\rhoprime = \code{m}$.
%
    By the inductive hypothesis (\CPinv) on $\thr{T} \in \rhsSet{}$,
    $\thr{T} \in \CPLockSetPlain{\rho}{}$ or
    $\exists \HBLock{n}{k} \mid \PCPThread{T}{n}{k} \in \PCPLockSetPlain{\rho}{} \land
    \exists j \mid \CPOrdered{\Release{n}{k}}{\Acquire{n}{j}} \totalOrder e$.

      \begin{description}
  
      \item[Case 2d(i):] $\thr{T} \in \CPLockSetPlain{\rho}{}$ 

      Then  Algorithm~\ref{alg:release}'s line~\ref{line:rel:CP-condition} evaluates to true,
      so line~\ref{line:rel:add-CPlock} adds \code{m} to $\CPLockSetPlain{\rho}{}^+$; thus $\code{m} \in \lhsSet{+}$.

      \item[Case 2d(ii):] Let \OwnerElement{n}{k} and \OwnerElement{n}{j} be such that
      $\CCPLock{T}{n}{k} \in \CCPLockSetPlain{\rho}{} \land
      \CPOrdered{\Release{n}{k}}{\Acquire{n}{j}} \totalOrder e$.

      If $\code{n} \ne \code{m}$,
      then $\CCPLock{T}{n}{k} \in \CCPLockSetPlain{\rho}{}$ matches line~\ref{line:rel:CCP-condition} in Algorithm~\ref{alg:release},
      so line~\ref{line:rel:add-dependent-locks} adds \CCPLock{m}{n}{k} to $\CCPLockSetPlain{\rho}{}^+$.
      Since $e \totalOrder e^+$, $\PCPLock{m}{n}{k} \in \CCPLockSetPlain{\rho}{} \land \CPOrdered{\Release{n}{k}}{\Acquire{n}{j}} \totalOrder e^+$,
      \ie, $\code{m} \in \lhsSet{+}$.

      Otherwise ($\code{n} = \code{m}$),
      $\CCPLock{m}{m}{k} \notin \CCPLockSetPlain{\rho}{}^+$ because Algorithm~\ref{alg:release} removes 
      $\CCPLock{m}{m}{k}$ from $\CCPLockSetPlain{\rho}{}^+$ at line~\ref{line:rel:removal}.
      Since $\CPOrdered{\Release{m}{k}}{\Acquire{m}{j}} \totalOrder e$,
      $\CPOrdered{\Acquire{m}{k}}{\Acquire{m}{i}} \totalOrder e$ by \CP \RuleC (since $e = \Release{m}{i}$).
      By the inductive hypothesis (\CPinv) and Lemma~\ref{lem:helper},
      $\thr{T} \in \CPLockSet{m}{k}$ or
      $\exists \HBLock{o}{l} \mid \code{o} \ne \code{m} \land \CCPLock{T}{o}{l} \in \CCPLockSet{m}{k} \land
      \exists h \mid \CPOrdered{\Release{o}{l}}{\Acquire{o}{h}} \totalOrder e$.
    
        \begin{itemize}
      
        \item $\thr{T} \in \CPLockSet{m}{k}$

        In the pre-release algorithm (Algorithm~\ref{alg:pre-release}),
        line~\ref{line:prerel:trigger-condition} evaluated to true
        for $\PCPThread{T}{m}{k} \in \CPLockSetPlain{\rho}{}$ (line~\ref{line:prerel:CCP-condition}).
        Thus line~\ref{line:prerel:add-CPSigma} executed, so $\CPThread{T} \in \CPLockSetPlain{\rho}{}$.
        In the \emph{release} algorithm (Algorithm~\ref{alg:release}), line~\ref{line:rel:CP-condition} evaluates
        to true for $\rho$, so
        $\code{m} \in \CPLockSetPlain{\rho}{}^+$ and thus $\code{m} \in \lhsSet{+}$.

        \item Let \OwnerElement{o}{l} and \OwnerElement{o}{h} be such that
        $\PCPLock{T}{o}{l} \in \PCPLockSet{m}{k} \land
        \CPOrdered{\Release{o}{l}}{\Acquire{o}{h}} \totalOrder e$.

        In the pre-release algorithm (Algorithm~\ref{alg:pre-release}),
        line~\ref{line:prerel:add-dependent-transfers} executed for
        $\PCPThread{T}{m}{k} \in \PCPLockSetPlain{\rho}{}$ (line~\ref{line:prerel:CCP-condition})
        and
        $\PCPLock{T}{o}{l} \in \PCPLockSet{m}{k}$ (line~\ref{line:prerel:CCP-transfer-condition}).
        Line~\ref{line:prerel:add-dependent-transfers} added
        \PCPLock{T}{o}{l} to $\PCPLockSetPlain{\rho}{}$.
        In the \emph{release} algorithm (Algorithm~\ref{alg:release}),
        line~\ref{line:rel:add-dependent-locks} executes for
        $\PCPLock{T}{o}{l} \in \PCPLockSetPlain{\rho}{}$ (line~\ref{line:rel:CCP-condition}),
        adding
        \PCPLock{m}{o}{l} to $\PCPLockSetPlain{\rho}{}^+$.
        In addition, $\CPOrdered{\Release{o}{l}}{\Acquire{o}{h}} \totalOrder e^+$ since $e \totalOrder e^+$.
        Therefore $\code{m} \in \lhsSet{+}$.

        \end{itemize}

      \end{description}

    \end{description}

  \end{description}

\noindent
Since $\lhsSet{+} \subseteq \rhsSet{+} \land \lhsSet{+} \supseteq \rhsSet{+}$,
$\lhsSet{+} = \rhsSet{+}$.
\end{proof}

\end{empty}

%% file: Distance-Proof.tex
\begin{empty}

\begin{proof}


Let $e = \Release{m}{i}$ executed by thread \thr{T}.
Let $j$ be such that \CPOrdered{\Acquire{m}{j}}{e}.
Suppose Figure~\ref{Fig:invariants}'s invariants
hold for all execution points up to and including the point immediately before $e$.

We prove the lemma by induction on the \CPdistance \dist{m}{i}{j}.

\underline{Base case}:
If $\dist{m}{i}{j} = 0$, then by the definition of \CPdistance,
$\exists e', e'', i', j' \mid j \le j' < i' \le i \land
\conflicts{e'}{e''} \land
\POOrdered{\POOrdered{\Acquire{m}{j'}}{e'}}{\Release{m}{j'}} \land
\POOrdered{\POOrdered{\Acquire{m}{i'}}{e''}}{\Release{m}{i'}}$.

By the \invCPruleA invariant, $\CPThread{T'} \in \CPLockSet{m}{j'}$ at \Release{m}{i'}, where $\thr{T'} = \getThread{\Release{m}{i'}}$.
Because \EqHBOrdered{\Release{m}{i'}}{\Release{m}{i}} and the algorithms propogate \HB-ordered events through \CP \sets,
$\thr{T} \in \CPLockSet{m}{j'}$ at $e$.

\underline{Inductive step}:
$\dist{m}{i}{j} > 0$

Suppose the lemma holds true for all \OwnerElement{n}{k} and \OwnerElement{n}{l} such that
$\dist{n}{l}{k} < \dist{m}{i}{j}$.

By the definitions of \CP and \CPdistance,
there exist \OwnerElement{n}{k} and \OwnerElement{n}{l} such that
\HBOrdered{\CPOrdered{\HBOrdered{\Acquire{m}{j}}{\Release{n}{k}}}{\allowbreak\Acquire{n}{l}}}{\Release{m}{i}},
$\dist{m}{i}{j} = 1 + \dist{n}{l}{k}$, and $\code{n} \ne \code{m}$.
Either $\Release{m}{i} \totalOrder \Release{n}{l}$ or $\Release{n}{l} \totalOrder \Release{m}{i}$:

\begin{description}

\item [Case 1:] $\Release{m}{i} \totalOrder \Release{n}{l}$

Since \HBOrdered{\Acquire{m}{j}}{\Release{n}{k}},
at \Acquire{n}{l},
there exists $k' \le k$ such that
Algorithm~\ref{alg:acquire} adds \PCPLock{T'}{n}{k'} to \PCPLockSet{m}{j},
where $\thr{T'} = \getThread{\Acquire{n}{l}}$.
Since \HBOrdered{\Acquire{n}{l}}{\Release{m}{i}}, and the algorithms propagate
\PCP through \HB-ordered events,
at $e = \Release{m}{i}$,
$\PCPLock{T}{n}{k'} \in \PCPLockSet{m}{j}$.

\item [Case 2:] $\Release{n}{l} \totalOrder \Release{m}{i}$

To prove this case, we are interested in the release of a lock like \code{n} with ``minimum distance.''
Let \OwnerElement{o}{g} and \OwnerElement{o}{f} be such that

  \begin{itemize}
  \item $\Release{o}{g} \totalOrder \Release{m}{i}$;
  \item 
there exists $f' \ge f$ and $\rhoprime$ such that, at \Release{o}{g}, \\
$\exists j' \ge j \mid \PCPLockPlain{\rhoprime}{o}{f'} \in \PCPLockSet{m}{j'} \land
\exists e' \mid \CPOrdered{\Acquire{m}{j}}{e'} \land
\comp{\rhoprime}{e'} \land \HBOrdered{e'}{e}$; and
  \item $\dist{o}{g}{f}$ is as minimal as possible.
  \end{itemize}

\noindent
Note that \code{o} may be \code{n}, or else \code{o} is some other ``lower-distance'' lock,
and thus $\dist{o}{g}{f} < \dist{m}{i}{j}$ in any case.
It is possible that $\code{o} = \code{m}$,
in which case $f < j$ and $g < i$.

Since $\dist{o}{g}{f} < \dist{m}{i}{j}$,
by the inductive hypothesis,
at \Release{o}{g} (after the pre-release algorithm but before the release algorithm),
there exists $f' \ge f$ such that either
\begin{itemize}

\item $\thr{T'} \in \CPLockSet{o}{f'}$ or

\item $\exists \OwnerElement{q}{c} \mid \code{q} \ne \code{o} \land
\PCPLock{T'}{q}{c} \in \PCPLockSet{o}{f'} \land
\exists \OwnerElement{q}{d} \mid \CPOrdered{\Release{q}{c}}{\Acquire{q}{d}} \totalOrder \Release{o}{g} \land\\
\dist{q}{d}{c} < \dist{o}{g}{f}$

\end{itemize}

\noindent
where \thr{T'} = $\getThread{\Release{o}{g}}$.

  \begin{description}
  
  \item[Case 2a:] $\thr{T'} \in \CPLockSet{o}{f'}$

  At \Release{o}{g}, the pre-release algorithm (Algorithm~\ref{alg:pre-release})
  adds \CPThreadPlain{\rhoprime} to \CPLockSet{m}{j'} at line~\ref{line:prerel:add-CPSigma} because
  $\PCPLockPlain{\rhoprime}{o}{f'} \in \PCPLockSet{m}{j'}$ and
  $\CPThread{T'} \in \CPLockSet{o}{f'}$ (matching \lines{\ref{line:prerel:CCP-condition}}{\ref{line:prerel:trigger-condition}}).
  Since $\comp{\rhoprime}{e'} \land \HBOrdered{e'}{e}$,
  and the algorithms propagate \CP through \HB-ordered events,
  at $e = \Release{m}{i}$,
  $\CPThread{T} \in \CPLockSet{m}{j'}$.

%

  \item[Case 2b:] Let \OwnerElement{q}{c} and \OwnerElement{q}{d} be such that, at \Release{o}{g}, $\code{q} \ne \code{o} \land
  \PCPLock{T'}{q}{c} \in \PCPLockSet{o}{f'} \land
  \CPOrdered{\Release{q}{c}}{\Acquire{q}{d}} \totalOrder \Release{o}{g} \land
  \dist{q}{d}{c} < \dist{o}{g}{f}$.

  Thus $\Acquire{q}{d} \totalOrder \Release{o}{g} \totalOrder \Release{q}{d}$.

  If $\Release{q}{d} \totalOrder \Release{m}{i}$,
  then that would violate the stipulation above that $\dist{o}{g}{f}$ is minimal.
  Thus either $\Release{m}{i} = \Release{q}{d}$ or $\Release{m}{i} \totalOrder \Release{q}{d}$.

    \begin{description}

    \item [Case 2b(i):] $\Release{m}{i} = \Release{q}{d}$
    
    Thus $\code{m} = \code{q}$ and $i = d$, but
    $c < j$ since $\dist{m}{i}{c} = \dist{q}{d}{c} < \dist{m}{i}{j}$.
    
    At \Release{o}{g}, the pre-release algorithm (Algorithm~\ref{alg:pre-release})
    adds \CCPLockPlain{\rhoprime}{m}{c} to \CCPLockSet{m}{j'} at line~\ref{line:prerel:add-dependent-transfers} because
    $\CCPLockPlain{\rhoprime}{o}{f'} \in \CCPLockSet{m}{j'}$ (matching line~\ref{line:prerel:CCP-condition}) and
    $\CCPThread{T'}{m}{c} \in \CCPLockSet{o}{f'}$ (matching line~\ref{line:prerel:CCP-transfer-condition}).
    Since $\comp{\rhoprime}{e'} \land \HBOrdered{e'}{e}$,
    and the algorithms propagate \CCP through \HB-ordered events,
    at $e = \Release{m}{i}$,
    $\CCPThread{T}{m}{c'} \in \CCPLockSet{m}{j'}$.
    
    By the inductive hypothesis on \OwnerElement{m}{i} and \OwnerElement{m}{c},
    since $\dist{m}{i}{c} < \dist{m}{i}{j}$,
    there exists $c' \ge c$ such that, at \Release{m}{i}, one or both of the following hold:

      \begin{itemize}

      \item $\thr{T} \in \CPLockSet{m}{c'}$

      The pre-release algorithm (Algorithm~\ref{alg:pre-release}) adds
      \CPThread{T} to \CPLockSet{m}{j'} at line~\ref{line:prerel:add-CPSigma} because
      $\CCPThread{T}{m}{c} \in \CCPLockSet{m}{j'}$ (matching line~\ref{line:prerel:CCP-condition}) and
      $\CPThread{T} \in \CPLockSet{m}{c'}$ (matching line~\ref{line:prerel:trigger-condition}).

      \item $\exists \HBLock{r}{a} \mid \code{r} \ne \code{m} \land
      \CCPLock{T}{r}{a} \in \PCPLockSet{m}{c'} \land
      \exists b \mid \CPOrdered{\Release{r}{a}}{\Acquire{r}{b}} \totalOrder \Release{m}{i} \land\\
      \dist{r}{b}{a} < \dist{m}{i}{c}$

      The pre-release algorithm (Algorithm~\ref{alg:pre-release})
      adds \CCPLock{T}{r}{a} to \CCPLockSet{m}{j'} at line~\ref{line:prerel:add-dependent-transfers} because
      $\CCPLock{T}{m}{c'} \in \CCPLockSet{m}{j'}$ (matching line~\ref{line:prerel:CCP-condition}) and
      $\CCPLock{T}{r}{a} \in \CCPLockSet{m}{c'}$ (matching line~\ref{line:prerel:CCP-transfer-condition}).

      \end{itemize}

    \item [Case 2b(ii):] $\Release{m}{i} \totalOrder \Release{q}{d}$
     
    Note that $\code{m} \ne \code{q}$ because $\Acquire{q}{d} \totalOrder \Release{o}{g} \totalOrder \Release{m}{i} \totalOrder \Release{q}{d}$.

    At \Release{o}{g}, the pre-release algorithm (Algorithm~\ref{alg:pre-release})
    adds \CCPLockPlain{\rhoprime}{q}{c} to \CCPLockSet{m}{j'} at line~\ref{line:prerel:add-dependent-transfers} because
    $\CCPLockPlain{\rhoprime}{o}{f'} \in \CCPLockSet{m}{j'}$ (matching line~\ref{line:prerel:CCP-condition}) and
    $\CCPThread{T'}{q}{c} \in \CCPLockSet{o}{f'}$ (matching line~\ref{line:prerel:CCP-transfer-condition}).
    Since $\comp{\rhoprime}{e'} \land \HBOrdered{e'}{e}$,
    and the algorithms propagate \PCP through \HB-ordered events,
    at $e = \Release{m}{i}$,
    $\CCPThread{T}{q}{c'} \in \CCPLockSet{m}{j'}$.
    
    \end{description}

  \end{description}

\end{description}

\noindent
Thus, the lemma's statement holds for each case.
\end{proof}

\end{empty}